\documentclass[english,12pt,twoside,reqno]{article}
\usepackage{a4}
\usepackage[latin1]{inputenc}
\usepackage{babel}
\usepackage{amsfonts}
\usepackage{amsmath}
\usepackage{euscript}
\usepackage{amsthm}
\usepackage{amssymb}
\usepackage{graphicx}
\usepackage{enumerate}
\usepackage{dsfont}
\usepackage{color}
\usepackage{hyperref}
\DeclareGraphicsExtensions{.pdf,.jpg,.png}

\newtheorem{stat}{Statement}[section]
\newtheorem{df}{Definition}[section]
\newtheorem{prop}{Proposition}[section]
\newtheorem{theo}{Theorem}[section]
\newtheorem{lemma}{Lemma}[section]

\newtheorem{cor}{Corollary}[section]

\newtheorem{rkkk}{Remarks}[section]
\newtheorem{rkk}{Remark}[section]

\newcommand{\bb}{\mathbb}

\newcommand{\veps}{\varepsilon}

\newcommand{\noi}{\noindent}

\newcommand{\vphi}{\varphi}

\newcommand{\R}{\bb R}\Large\LARGE

\newcommand{\I}{\vert}

\newcommand{\demo}{\noindent\textit{ Proof -~}}
\newcommand{\findemo}{\hfill $\Box$}
\title{The Coefficient Problem and Multifractality of Whole-Plane SLE \& LLE}
\author{Bertrand Duplantier$\,^{(a,b)}$\thanks{e-mail: \texttt{bertrand.duplantier@cea.fr} -- Partially supported by  grant  ANR-08-BLAN-0311-CSD5.}, 
Nguyen Thi Phuong Chi$\,^{(c)}$\thanks{e-mail: \texttt{ntpchi@gmail.com}},\\ Nguyen Thi Thuy Nga$\thanks{e-mail: \texttt{nguyennga.1a@gmail.com} -- Partially supported by the ``R\'{e}gion Centre'' through the grant APR 
TRUC (Transport, R\'{e}seaux, Croissance, Urbanisme.)}\,^{(c)}$, Michel Zinsmeister$\,^{(c,d)}$\thanks{e-mail: \texttt{zins@univ-orleans.fr} 
}
\\
{}\\
{\it $^{(a)}$Institut de Physique Th\'{e}orique, CEA/Saclay}\\
{\it F-91191 Gif-sur-Yvette Cedex, France} \\
{}\\
{\it $^{(b)}$The Mathematical Sciences Research Institute}\\
{\it 17 Gauss Way,
Berkeley, CA 94720-5070, USA}\\ 
{}\\
{\it $^{(c)}${\sc mapmo}}\\
{\it Universit\'e d'Orl\'eans}\\
{\it B\^{a}timent de math\'{e}matiques, rue de Chartres}\\
{\it 
B.P. 6759 - F-45067 Orl\'{e}ans Cedex 2, France}\\
{}\\
{\it $^{(d)}$ Laboratoire de physique th\'{e}orique de la mati\`{e}re condens\'{e}e}\\ 
{\it UMR CNRS 7600, Tour 12-13/13-23, Bo\^{\i}te 121}\\
{\it 4, Place Jussieu}\\
{\it 75252 Paris Cedex 05, France}}

\begin{document}
\maketitle

\begin{abstract} 
	 Karl L\"owner (later known as Charles Loewner) introduced his famous differential equation in 1923 in order to solve the Bieberbach conjecture for series expansion coefficients of univalent analytic functions at  level $n=3$. His method was revived in 1999 by Oded Schramm when he introduced the {\it Stochastic Loewner Evolution} (SLE), a conformally invariant process which made it possible to prove many predictions from conformal field theory  for critical planar models in statistical mechanics. The aim of this paper is to revisit the Bieberbach conjecture in the framework of {\it SLE processes} and, more generally, {\it L\'evy processes}. The study of their {\it unbounded whole-plane} versions leads to a discrete series of exact results for the expectations of coefficients and their variances, and, more generally, for the derivative moments of some prescribed order $p$. These results are generalized to the  ``oddified'' or $m$-fold conformal maps of whole-plane SLEs or L\'evy--Loewner Evolutions (LLEs). We also study the (average) integral means multifractal spectra of these unbounded whole-plane SLE curves. We prove the existence of a phase transition at a moment order $p=p^*(\kappa)>0$, at which one goes from the bulk SLE$_\kappa$ average integral means spectrum, as predicted by one of us \cite{2000PhRvL..84.1363D} and established by Beliaev and Smirnov \cite{BS}, and valid for $p\leq p^*(\kappa)$, to a new integral means spectrum for $p\geq p^*(\kappa)$, as conjectured in part in Ref. \cite{IL}.  The latter spectrum is furthermore shown to be intimately related, via the associated packing spectrum, to the radial SLE derivative exponents  obtained by Lawler, Schramm and Werner \cite{MR2002m:60159b}, and to the local SLE tip multifractal exponents obtained from quantum gravity in Ref. \cite{MR2112128}. This is  generalized to the integral means spectrum of the $m$-fold transform of the unbounded whole-plane SLE map. A  succinct, preliminary, version of this study first 
appeared in Ref. \cite{Hal-DNNZ}.
\end{abstract}
\section{Introduction}
\subsection{The coefficient problem and Schramm--Loewner evolution}\label{coeff}
 \noi Let $ f(z)=\sum_{n\geq 0}a_nz^n$ be a holomorphic function in the unit disc $\bb D$. We further assume that the function $f$ is injective: what then can be said about the coefficients $a_n$? A trivial observation is that $a_1\neq 0$ and Bieberbach \cite{Bi} proved in 1916 that $$\I a_2\I\leq 2\I a_1\I.$$ In the same paper he famously conjectured that$$\forall n\geq 2,\,\I a_n\I\leq n\I a_1\I,$$ guided by the intuition that the function (afterwards called the \textit{Koebe} function) 
 \begin{equation}
 \label{koebe}\mathcal K(z):=-\sum_{n\geq 1} n (-z)^{n}=\frac{z}{(1+z)^2},
 \end{equation} 
 which is a holomorphic bijection between $\bb D$ and $\bb C\backslash [1/4, +\infty)$, should be extremal. This conjecture was finally proven in 1984 by de Branges \cite{dB}: its proof was made possible by the addition of a new idea (an inequality of Askey and Gasper) to a series of methods and results developed in almost a century of effort.
 It is largely accepted that the earliest important contribution to the proof of Bieberbach's conjecture is the proof  \cite{Lo} by Loewner in 1923 that $\I a_3\I \leq 3\I a_1\I$. De Branges' proof in  1985 \cite{dB} indeed used Loewner's idea in a crucial way, as did many contributors to the proof around that time.
In an Appendix to this article, we recall the proof by Bieberbach for the case $n=2$, and that by Loewner for $n=3$. 
It ends with a brief account of post-Loewner steps towards the proof of Bieberbach's conjecture.

Loewner's ideas go far beyond Bieberbach's conjecture: Oded Schramm \cite{Schr} revived Loewner's method in 1999, introducing {\it randomness} into it, as driven by standard Brownian motion. This field, now called the theory of SLE processes (initially for Stochastic Loewner, now for Schramm--Loewner, Evolution),  provides a unified and rigorous approach to the geometry of conformally invariant processes and critical curves in two-dimensional statistical mechanics. It led to the two Fields medals of W. Werner (for the application of SLE to planar Brownian paths) and of S. Smirnov 
(for application of SLE to critical percolation and Ising models). 

The aim of the present paper is to revisit Bieberbach's conjecture in the framework of SLE theory, that is to study the coefficients of univalent functions coming from the conformal maps associated with this process. We also extend our study to the so-called {\it L\'evy--Loewner Evolution} (LLE), where the Brownian source term in Loewner's equation is generalized to a L\'evy process. (See, e.g.,\cite{2006JSMTE..01..001R,2008JSMTE..01..019O}.)


There exist several variants of $\textrm{SLE}_\kappa$,  known, in a terminology due to Schramm, as {\it chordal}, {\it radial}, or {\it whole-plane}. \textit{The one we adopt in this work is a variant of the whole-plane one}, corresponding to the original setting introduced by Loewner. 
As in the radial case, the whole-plane Loewner process is determined by a function $\lambda: [0,+\infty)\to \partial \mathbb D:=\{z: \vert z\vert=1\}$, called the {\it driving function}, obtained as follows.  
Define $\gamma:\,[0,\infty)\to\bb C$ to be a Jordan arc joining $\gamma(0)$ to $\infty$, and not containing  the origin $0$ (see Fig. \ref{whpl}).  Define then for each $t>0$, the slit domain $\Omega_t=\bb C\backslash \gamma([t,\infty)).$ It is a simply connected domain containing $0$ and we can thus consider the Riemann mapping $f_t: \bb D\to \Omega_t,\,f_t(0)=0, f_t'(0)>0.$ By the Caratheodory convergence theorem, $f_t$ converges as $t\to 0$ to $f:=f_0$,  the Riemann mapping of $\Omega_0$. We may assume without loss of generality that $f'(0)=1$ and, by changing the time $t$ if necessary, choose the normalization $f_t'(0)=e^t$.

The key idea of Loewner is to use the fact that the sequence of domains $\Omega_t$ is increasing, which translates into the fact that $ \Re\left(\frac{\partial f_t}{\partial t}/z\frac{\partial f_t}{\partial z}\right)>0$ or, equivalently, that this quantity is the Poisson integral of a positive measure on the unit circle, actually a probability measure because of the above normalization. Now the fact that the domains $\Omega_t$ are slit domains implies that for every $t$ this probability measure must be a Dirac mass at point $\lambda(t)=f_t^{-1}(\gamma(t))$. It is worthwhile to notice that $\lambda$ is a continuous function. One says that the Loewner chain $(f_t)$ associated with $(\Omega_t)$ is driven by the function $\lambda(t)$, in the sense that $f_t$ satisfies the Loewner differential equation
\begin{equation}\label{loewner}
\frac{\partial f_t}{\partial t}=z\frac{\partial f_t}{\partial z}\frac{\lambda(t)+z}{\lambda(t)-z},\,\,\,z\in \mathbb D.
\end{equation}
\noi It is remarkable that that the Loewner method can be reversed: given a function $\lambda$ which is {\it c\`adl\`ag}, i.e., right continuous with left limits at every point of $\bb R_+$ with values in the unit circle, then the Loewner equation \eqref{loewner}, \textcolor{black}{supplemented by the final condition $f_{t\to +\infty}(z)=z$,}  has a solution $f_t(z)$ which is the Riemann mapping of a domain $\Omega_t$ and the corresponding family is increasing in $t$.\\

\begin{figure}[tb]
\begin{center}
\includegraphics[angle=90,width=.93290\linewidth]{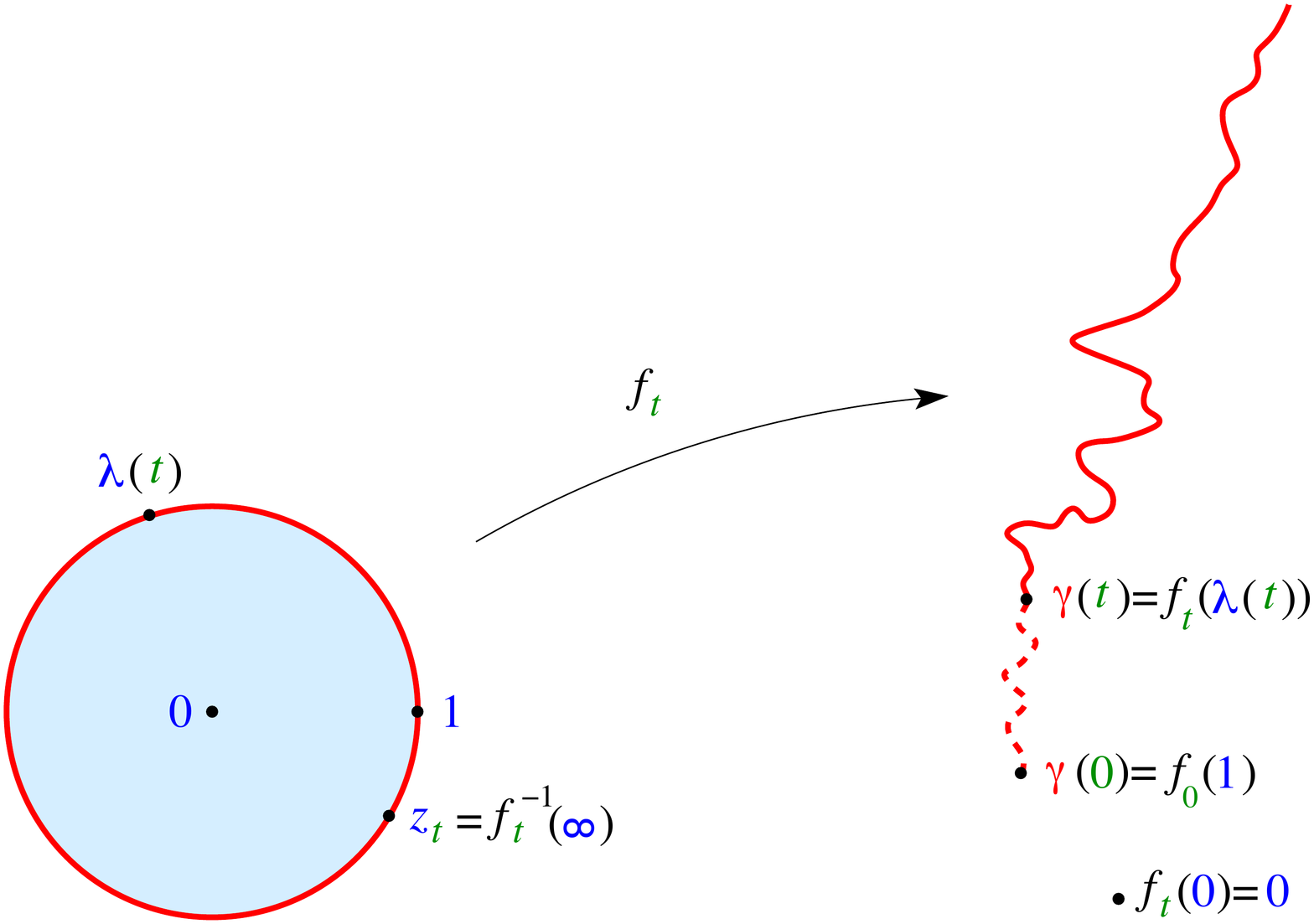}
\caption{{\it Loewner map $z\mapsto f_t(z)$ from $\mathbb D$ to the slit domain $\Omega_t=\bb C\backslash \gamma([t,\infty))$ (here slit by a single curve $\gamma([t,\infty))$ for $\kappa\leq 4$). One has $f_t(0)=0, \forall t\geq 0$. At $t=0$, the driving function $\lambda (0)=1$, so that the image of $z=1$ is at the tip $\gamma(0)=f_0(1)$ of the curve.}}
\label{whpl}
\end{center}
\end{figure}As is well-known, Schramm's fundamental insight was to consider as a particular driving function 
\begin{equation}\label{SLE}\lambda(t)=e^{i\sqrt{\kappa}B_t},\end{equation}
where $\kappa \in [0,\infty)$, and $B_t$ is standard, one-dimensional, Brownian motion, characterized by the three fundamental properties:
\begin{enumerate}[(a)]\item {\it Stationarity:} if $0\leq s \leq t$, then $B_t-B_s$ has the same law as $B_{t-s}$;
\item {\it Markov property:}  if $0\leq s \leq t$, then $B_t-B_s$ is independent of $B_s$;
\item {\it Gaussianity:} $B_t$ has a normal distribution with mean $0$ and variance $t$.
\end{enumerate}
A L\'evy  process $L_t$ provides the generalization 
that is assumed to satisfy only the first two of these properties, the essential difference with Brownian motion being that jumps are then allowed. \textcolor{black}{The corresponding stochastic L\'evy--Loewner evolution (LLE) obeys \eqref{loewner} with a source term
 that generalizes \eqref{SLE}  
 \begin{equation}\label{LLE}
 \lambda(t)=e^{iL_t}.
 \end{equation}}
The characteristic function of a L\'evy process $L_t$ has the form
\begin{equation} \label{Levychar}
\mathbb E(e^{i\xi L_t})=e^{-t\eta(\xi)}\end{equation}
where $\eta$ (called the L\'evy symbol) is a continuous complex function of $\xi\in \mathbb R$,  satisfying (in addition to necessary Bochner type conditions \cite{applebaum})
$\eta(0)=0$, and $\eta(-\xi)=\overline{\eta(\xi)}.$
$\textrm{SLE}_\kappa$ corresponds to a Gaussian characteristic function, and \textcolor{black}{its driving function} is a L\'evy process with symbol 
\begin{equation}\label{esle}
\eta(\xi)=\kappa \xi^2/2.
\end{equation}
More generally, the function
\begin{equation}\label{estable}\eta(\xi)=\kappa \vert \xi\vert^\alpha/2,\,\,\,\alpha\in (0,2]\end{equation}
is the L\'evy symbol of the so-called $\alpha-${\it stable process}.  The normalization here is chosen so that this process \textcolor{black}{gives} $\textrm{SLE}_\kappa$ for $\alpha=2$. Another L\'evy symbol of interest is given (up to  constant factor) by 
$\eta(\xi)=1-({\sin \pi\xi)}/{\pi\xi},$ 
and corresponds to a certain compound Poisson process which serves as a model for a dendritic growth process; this aspect will be developed in a forthcoming paper (see also \cite{Johansson2009238}). 

\textcolor{black}{
The most general form of a L\'evy symbol is given by the well-known L\'evy-Khintchine formula (which makes precise the Bochner-type conditions  mentioned above). It states that a L\'evy symbol (in dimension one) has the necessary form
$$ \eta(\xi)=ib\xi+a^2\xi^2-\int_{\bb R\backslash\{0\}} \left[e^{i\xi y}-1-i\xi y1_{[-1,1]}\right]d\nu(y),$$
where $a,b\in \mathbb R$, and $\nu$ is a measure on $\bb R\backslash\{0\}$ such that
$$ \int_{\bb R\backslash\{0\}}(1\wedge y^2) d\nu(y)<\infty.$$}
In the examples above $\eta$ is a real, therefore even function, a property which we will assume throughout, except in the beginning of Section \ref{analytic}. As we shall see, all the quantities that we will consider depend only on the values of the L\'evy symbol at {\it integer  arguments}; for this reason we shall use the ``sequence" notation: $\eta_k:=\eta(k),\,k\in\bb Z$.

The associated conformal maps,  obeying  \eqref{loewner}, are denoted by $f_t$, and  in this work, we study their coefficients $a_n(t)$, which are random variables, defined by the normalized series expansion: 
\begin{equation}\label{defcoeff}
f_t(z)=e^t\big(z+\sum_{n\geq2} a_n(t)z^n\big).
\end{equation} 
Section \ref{analytic} starts with the computation, in terms of the L\'evy symbols $\eta_k, k\in \mathbb Z$, of $\mathbb E(a_n)$ for all $n$, and of $\mathbb E(\I a_n\I^2)$ for small $n$,  for a general L\'evy--Loewner evolution process $f_t$. Note that a similar idea already appeared in Ref. \cite{2010JSP...139..108K}, where A. Kemppainen  studied in detail the coefficients associated with the Schramm--Loewner evolution, using a stationarity property of SLE \cite{2006JPhA...39L.657K}. However, the  focus there was on expectations of the moments of those coefficients, rather than on the moments of their moduli.  

We also consider the  associated odd (``oddified'') process, defined as :   
\begin{equation}\label{defoddified}
h_t(z):=z\sqrt{f_t(z^2)/z^2}),
\end{equation} represented by the normalized series expansion: 
\begin{equation}\label{defcoeffb}
e^{-t/2}h_t(z)=z+\sum_{n\geq 1}b_{2n+1}(t)z^{2n+1}.\end{equation} 
\textcolor{black}{The transform \eqref{defoddified} was the key to the proof of the Bieberbach conjecture.  The so-called Littlewood-Paley conjecture that the odd coefficients satisfy $|b_{2n+1}|\leq 1$ (an inequality which implies Bieberbach's) was actually disproved by Fekete and Szeg\H o, but its modification by Robertson claiming that $\sum_{k=1}^n|b_{2k+1}|^2\leq n$ (which also implies Bieberbach's conjecture) was finally proven in de Branges's work; see the historical sketch \ref{appendix} at the end of this paper.}

\textcolor{black}{This transform has  been generalized to the \textit{$m$-fold transform}
 \begin{equation}
 \label{defmfold}
 h_t^{(m)}(z):=z(f_t(z^m)/z^m)^{1/m},\end{equation}  
 defined for $m\in\bb N,\,m\geq 1$ (see below).}

 We find the following results:
\begin{theo} \label{theolle0} Let $(f_t)_{t\geq 0}$ be the Loewner whole-plane process driven by the L\'evy process $L_t$ with L\'evy symbol $\eta$. We write
\vskip -.2cm
$$f_t(z)=e^t\big(z+\sum_{n\geq 2}a_n(t) z^n\big),$$
\vskip -.2cm
 We also consider the oddification of $f_t$,
\vskip -.1cm
$$h_t(z)=z\sqrt{f_t(z^2)/z^2}=e^{t/2}\big(z+\sum_{n\geq 1}b_{2n+1}(t)z^{2n+1}\big)$$
\vskip -.1cm
Then the {\it conjugate} whole-plane L\'evy--Loewner evolution $e^{-iL_t} f_t\big(e^{iL_t}z\big)$ has the {\it same law} as $f_0(z)$, i.e., $e^{i(n-1)L_t}a_n(t)\stackrel{\rm (law)}{=}a_n(0)$.
 Similarly, the  {\it conjugate} oddified whole-plane L\'evy--Loewner evolution $e^{-(i/2)L_t} h_t\big(e^{(i/2)L_t}z\big)$ has the {\it same law} as $h_0(z)$, i.e., $e^{inL_t}b_n(t)\stackrel{\rm (law)}{=}b_n(0)$.\\

Setting  $a_n:=a_n(0)$ and $b_{2n+1}:=b_{2n+1}(0)$, we have
\begin{eqnarray*}
\mathbb E(a_n)&=&\prod_{k=0}^{n-2}\frac{\eta_{k}-k-2}{\eta_{k+1}+k+1},\,\,\, n\geq 2,\\
\mathbb E(b_{2n+1})&=&\prod_{k=0}^{n-1}\frac{\eta_k-k-1}{\eta_{k+1}+k+1},\,\,\,n\geq 1.
\end{eqnarray*}
\end{theo}
\begin{cor}\label{corfprime} In the setting of Theorem \ref{theolle0},
\begin{enumerate}[(i)]
\item if $\eta_1=3$, $\mathbb E(f'_0(z))=1-z$;
\item if $\eta_1=1$ and $\eta_2=4$, $\mathbb E(f'_0(z))=(1-z)^2$;
\item if $\eta_1=2$, $\mathbb E(\textcolor{black}{h}'_0(z))=1-z^2$.
\end{enumerate}\end{cor}
\textcolor{black}{Theorem \ref{theolle0} will be proven in Section \ref{proofs} as the combination of Theorem \ref{theoEan} and Theorem \ref{theoEbn}.}

Direct computations of expectations $\mathbb E(\I a_n\I^2)$ are already quite involved at level $n=4$, and we have used computer assistance in symbolic calculus with {\sc matlab}  for higher coefficients. These computer experiments, briefly explained in Section \ref{computation},  lead to the following statements, explicitly checked up to $n=8$, and proven in Section \ref{proofs}:
\begin{theo} \label{theolle} 
In the same setting as in Theorem \ref{theolle0},
\begin{enumerate}[(i)]
\item If $\eta_1=3$, we have 
$$\mathbb E(\I a_n\I^2)=1,\,\forall n\geq 1;$$
this case covers SLE$_6$.
\item If $\eta_1=1,\,\eta_2=4$, we have 
$$\mathbb E(\I a_n\I^2)=n,\,\forall n\geq 1;$$ this case covers SLE$_2 $.
\item If $\eta_1=2$, we have 
$$\mathbb E(\I b_{2n+1}\I^2)=\frac{1}{2n+1},\,\forall n\geq 1;$$ this case covers \textcolor{black}{the oddified} SLE$_4$.
\end{enumerate}
\end{theo} 
\textcolor{black}{\begin{rkk}
 For $\textrm{SLE}_\kappa$, recall that Eq. \eqref{esle} gives $\eta_1=\kappa/2$,  thus case {\it (i)} includes  $\textrm{SLE}_{\kappa=6}$. Since Eq. \eqref{esle} also gives $\eta_2=2\kappa$,  $\textrm{SLE}_2$ is included in case {\it (ii)}. Case {\it (iii)} includes the oddified $\textrm{SLE}_4$. 
 \end{rkk} 
 \begin{rkk}
In the second case {\it (ii)}, we have noticed for all explicitly computed coefficients ($n\leq 8$), and for all numerically computed ones ($n\leq19$), that the condition $\eta_1=1$ in fact
 suffices for the conclusion $\mathbb E(\I a_n\I^2)=n$ to hold. This property was first conjectured to be valid for any coefficient degree $n$ in Ref.  \cite{Hal-DNNZ}. It has been revisited in Ref. \cite{2013arXiv1301.6508L}. 
\end{rkk}}
  Section \ref{proofs} is devoted to proofs and begins with the computation of $\mathbb E(f_t(z))$ and $\mathbb E(h_t(z))$. We show in particular that these expectations take a  simple, {\it polynomial} form for the two cases above, $\eta_1=3$ and $\eta_1=1,\,\eta_2=4$,  and more generally, when there exists a $k\in \mathbb N$, such that $\eta_k=2+k$. In the odd case, these special values are  $\eta_k=1+k$.  This also yields the derivative expectations $\mathbb E[f'_t(z)]$. These results are used in the remainder of the section,  devoted to proving Theorem \ref{theolle} and obtaining other identities.

\textcolor{black}{After our earlier draft  \cite{Hal-DNNZ} was posted,  cases {\it (i)} and {\it (ii)} of Theorem \ref{theolle} were obtained for SLE in  Ref. \cite{IL}. It used a differential equation obeyed by the moments of $|f'_t(z)|$, and obtained by Hastings's (heuristic) method \cite{PhysRevLett.88.055506}. A resulting double recursion then becomes solvable for $\kappa=6,2$, with some computer assistance.} 

\textcolor{black}{This differential equation appeared in a paper by Beliaev and Smirnov (BS) \cite{BS} (see also Beliaev's dissertation  \cite{BKTH}), for another variant of  whole-plane SLE, along with its extension to the LLE case. The latter allows us to prove cases {\it (i)} and {\it (ii)} for  L\'evy--Loewner evolutions.} 

\textcolor{black}{Starting from the BS equation, 
we provide an analytic method to obtain a series of explicit solutions to that equation. In the case of $\textrm{SLE}_\kappa$, closed-form expressions are obtained for the moments $\mathbb E\big[(f'_t(z))^{p/2}\big]$ and $\mathbb E \big[|f'_t(z)|^p\big]=\mathbb E\big[(f'_t(z))^{p/2}(\overline{f'_t(z)})^{p/2}\big]$, 
 for a special set of values of the parameter $p$ depending on $\kappa$, that includes $p=2$ for $\kappa=2,6$  (see also \cite{Hal-DNNZ,IL}). We next show how to  extend  SLE results directly to the LLE case. 
We  further derive modified BS equations  for the oddified version \eqref{defoddified} of SLE or LLE processes, or for their $m$-fold transforms \eqref{defmfold}. For each value of $m\geq 1$, we construct a set of exact solutions;  in the oddified $m=2$ case, this yields a proof of case {\it (iii)} of Theorem \ref{theolle} for SLE and LLE.}\\ 
 
 {\it We would like to stress  that it is only for the ``inner'' variant of whole-plane SLE or LLE that we have introduced in Ref. \cite{Hal-DNNZ} and study  here, that such explicit, closed-form properties may exist.} 
 
 \textcolor{black}{This phenomenon may have a deeper explanation.}  
\textcolor{black}{This suggests future investigations of more general driving functions. A possible class of examples is
$\lambda(t)=e^{i(L_t+\mu(t))}$, 
where $L_t$ is a L\'evy process and $\mu$ is a function of bounded variation, or perhaps, more restrictively, in the Sobolev class $H^1$. This describes  a deterministic Loewner growth process  perturbed by  random noise. One may imagine this approach yielding insights towards a probabilistic proof of Bieberbach's conjecture.}

\subsection{Integral means spectra of whole-plane SLE}\label{intspec}These results are used in Section \ref{multifractal}, to study the {\it multifractal integral means spectrum} of our whole-plane processes.  Plancherel's theorem yields the easy corollary of Theorem \ref{theolle}:
\textcolor{black}{\begin{cor} \label{corint}For a L\'evy--Loewner evolution with $\eta_1=0$, $\eta_1=1$ and $\eta_2=4$, $\eta_1=3$ (thus including SLE for $\kappa=0,2,6$), and for an oddified LLE with $\eta_1=2$ (thus \textcolor{black}{oddified} SLE for $\kappa=4$), one has, respectively:
$$\mathbb E\left(\frac{1}{2\pi}\int_0^{2\pi}\I f'(re^{i\theta})\I^2d\theta\right)=\frac{1+11r^2+11r^4+r^6}{(1-r^2)^5};\frac{1+4r^2+r^4}{(1-r^2)^4} ; \frac{1+r^2}{(1-r^2)^3}; \frac{1+r^4}{(1-r^4)^2}.$$
\end{cor}}
\noindent The first case is obtained directly from the Koebe function \eqref{koebe}, which coincides with the whole-plane SLE map for $\kappa=0$. We can rephrase these results in terms of the following:  
\begin{df}The integral means spectrum of a conformal mapping $f$ is the function defined on $\bb R$ by
\begin{equation} \label{betadef} \beta(p):=\overline{\lim}_{r\to 1}\frac{\log(\int_{\partial \bb D}\I f'(rz)\I^p \I dz\I)}{\log(\frac{1}{1-r})}.\end{equation}
\end{df}
\noi
In the \textit{stochastic} setting, we define the {\it average}  integral means spectrum 
\begin{df}  \begin{equation}\label{betavdef}\beta(p):=\overline{\lim}_{r\to 1}\frac{\log( \int_{\partial \bb D} \mathbb E\, \I f'(rz)\I^p\, \I dz\I)}{\log(\frac{1}{1-r})}.\end{equation}
\end{df}

\textcolor{black}{The preceding results show that,  {\it in the expectation sense} of definition \eqref{betavdef}, these exponents can be read off as 
$ \beta(2)=5,4,3$ for whole-plane LLE with  $\eta_1=0$, $\eta_1=1$ and $\eta_2=4$, $\eta_1=3$ (thus whole-plane SLE with $\kappa=0,2,6$), respectively. For the oddified LLE with $\eta_1=2$ (thus the oddified whole-plane SLE$_4$), $\beta_2(2)=2$.} 

\begin{rkk}
Define the functions  \begin{eqnarray}\label{gamma00} \gamma_0(p,\kappa)&:=&\frac{1}{2\kappa}\left(4+\kappa-\sqrt{(4+\kappa)^2-8\kappa p}\right),\\ \nonumber \beta_0(p,\kappa)&:=&\frac{\kappa}{2}\gamma_0^2=-p+\frac{4+\kappa}{2}\gamma_0\\ \label{beta000}&=&-p+\frac{4+\kappa}{4\kappa}\left(4+\kappa-\sqrt{(4+\kappa)^2-8\kappa p}\right),\\ \label{tildebeta00}
\hat \beta_0(p,\kappa)&:=&p-\frac{(4+\kappa)^2}{16\kappa}.
\end{eqnarray}
They yield the average integral means spectrum $\bar \beta_0(p,\kappa)$ of the bulk of the outer whole-plane version of SLE$_\kappa$, as given by Eqs. (11) (12) and (14) in Beliaev and Smirnov (BS) \cite{BS}:
 \begin{eqnarray}\label{beta00bar} \bar\beta_0(p,\kappa)&=&\begin{cases}\beta_0(p,\kappa),\,\,\,0\leq p\leq p_0^*(\kappa),\\ 
\hat\beta_0(p,\kappa),\,\,\,p\geq p_0^*(\kappa),\end{cases}\\
\label{p00star}
p_0^*(\kappa)&:=&\frac{3(4+\kappa)^2}{32\kappa}.
\end{eqnarray}
 \end{rkk}
\textcolor{black}{\begin{rkk}
 The above values $\beta(2,\kappa)=5,4,3$ for whole-plane SLE$_{\kappa=0,2,6}$, or $\beta_2(2,\kappa=4)=2$ for oddified SLE$_4$, do not agree with the BS spectrum: they are greater than $1$ while $\beta(2)<1$ for bounded maps (see the discussion after Remark \ref{pneg}). This illustrates the fact that the inner version of the whole-plane SLE  is unbounded with positive probablility.  
 \end{rkk}}
 Motivated by this observation, we determine the multifractal integral means spectrum of our inner version of whole-plane SLE$_\kappa$. To this aim, we perform the singularity analysis near the unit circle of the corresponding BS equation. \textcolor{black}{ The same question for oddified or $m$-fold  symmetrized whole-plane SLE is also natural, since it illustrates how the previously unnoticed part of the multifractal spectrum depends on the role of the point at infinity.  The consideration of the $m$-fold version is further motivated by the work by Makarov \cite{Makanaliz} on the universal spectra, showing very similar phenomena.}

The unbounded whole-plane SLE spectra are given in the following (non-rigorous) statement:
\begin{stat} 
\label{theoMF}   In the unbounded case of the inner whole-plane SLE$_\kappa$ process, $f_{t=0}(z), z\in \mathbb D$, as defined by the Schramm--Loewner equation \eqref{loewner}, and of its $m$-fold transforms, $h_0^{(m)}(z):=z\big[f_0(z^m)/z^m\big]^{1/m}, m\geq 1$,   the respective average integral means spectra $\beta(p,\kappa)$ and $\beta_m(p,\kappa)$  all exhibit a phase transition and are given, for $p\geq 0$, and for $1\leq m\leq 3$, by
\begin{eqnarray}\label{beta}
\beta(p,\kappa):=\beta_1(p,\kappa)&=&\max\left\{\beta_0(p,\kappa),3p-\frac{1}{2}-\frac{1}{2}\sqrt{1+2\kappa p}\right \},\\ 
\label{betam}
\beta_m(p,\kappa)&=& \max \left \{\beta_0(p,\kappa),B_m(p,\kappa)\right \},
\end{eqnarray}
\textcolor{black}{where 
\begin{equation}
\label{Bm}
B_m(p,\kappa):=\left(1+\frac{2}{m}\right)p-\frac{1}{2}-\frac{1}{2}\sqrt{1+\frac{2\kappa p}{m}}
\end{equation} 
is the multifractal spectrum corresponding to the unbounded part of the $m$-fold whole-plane SLE path.}

The first spectrum $\beta_1$ has its transition point, where the second term supersedes the first one, at 
\begin{eqnarray}
\nonumber p^*(\kappa)=p_1^*(\kappa)&:=&
  \frac{1}{16\kappa}\left((4+\kappa)^2-4-2\sqrt{2(4+\kappa)^2+4}\right)\\ \label{pstar} &=&\frac{1}{32\kappa}\left(\sqrt{2(4+\kappa)^2+4}-6\right) \left(\sqrt{2(4+\kappa)^2+4}+2\right),\end{eqnarray} while 
in general:
\begin{eqnarray}\nonumber p_m^*(\kappa)&:=& \frac{m}{8\kappa(m+1)^2}\left( (m+1)(4+\kappa)^2-8m
-4\sqrt{(m+1)(4+\kappa)^2+4m^2}\right)\\ \nonumber
&=&\frac{m}{8\kappa(m+1)^2}\left(\sqrt{(m+1)(4+\kappa)^2+4m^2}-2m-4\right)\\ \label{pmstar}
&&\times\left(\sqrt{(m+1)(4+\kappa)^2+4m^2}+2m\right).\end{eqnarray} 
For $1\leq m\leq 3$, one has $\forall \kappa\geq 0, p_m^*(\kappa)\leq p_0^*(\kappa)$ so that $\beta_0=\bar \beta_0$ in \eqref{betam}.

For $m\geq 4$, the   average integral means spectrum of the unbounded inner whole-plane SLE$_\kappa$ is given by 
\begin{eqnarray}
 \label{betamtilde}
\beta_m(p,\kappa)= \max \left \{\bar \beta_0(p,\kappa),B_m(p,\kappa)\right\},
\end{eqnarray}
with $\bar \beta_0$ defined as in \eqref{beta00bar}-\eqref{tildebeta00}. For $m\geq 4$, the order of the two critical points $p^*_0(\kappa)$ and $p^*_m(\kappa)$ depends on $\kappa$, and is given by \begin{eqnarray}\label{km0}
p_m^*(\kappa)\lesseqqgtr p_0^*(\kappa),\,\,\,\,\, \kappa\lesseqqgtr \kappa_m,\,\,\,\,\kappa_m:=4\frac{m+3}{m-3},\,\,\,\,\, m\geq 4,
\end{eqnarray}
such that for $\kappa\leq \kappa_m$, 
\textcolor{black}{\begin{eqnarray}\label{betamm00}
\beta_m(p,\kappa)=\begin{cases}\beta_0(p,\kappa),\,\,\,\,0\leq p\leq p_m^*(\kappa),\,\,\,
\\ B_m(p,\kappa),\,\,\,\, p_m^*(\kappa)\leq p,\end{cases}
\end{eqnarray}
whereas  for $\kappa\geq \kappa_m$,
\begin{eqnarray}\label{in0}\beta_m(p,\kappa)=\begin{cases}\beta_0(p,\kappa),\,\,\,\,0\leq p\leq  p_0^*(\kappa),\,\,\,
\\ 
\hat \beta_0(p,\kappa),\,\,\,\,p_0^*(\kappa)\leq p\leq p^{**}_m(\kappa),
\\ 
 B_m(p,\kappa),\,\,\,\,p_m^{**}(\kappa)\leq p,\end{cases}
\end{eqnarray}}
where $p_m^{**}(\kappa)$ is the second critical point
\begin{equation}\label{pdoublestar0}
p_m^{**}(\kappa):=m\frac{\kappa^2-16}{32\kappa},\end{equation}  where the last spectrum in \eqref{in0}  supersedes the linear spectrum \eqref{tildebeta00}.

 For $p\leq 0$, the average integral means spectrum, common to all $m$-fold versions of the inner or outer whole-plane SLE, is given, as in Eq. (14) of \cite{BS},  for $-1-{3\kappa}/{8}<p\leq 0$ by the bulk spectrum $\beta_0(p,\kappa)$ \eqref{beta000}, and for $p\leq -1-{3\kappa}/{8}$ by the so-called tip-spectrum \cite{BS,PhysRevLett.88.055506,0911.3983}:
\begin{eqnarray}\label{tip}
\beta_{\textrm{tip}}(p,\kappa)&:=&\beta_0(p,\kappa)-2\gamma_0(p,\kappa)-1=-p-1+\frac{\kappa}{2}\gamma_0(p,\kappa)\\ \label{beta0tip}
&=&-p-1+\frac{1}{4}\left(4+\kappa-\sqrt{(4+\kappa)^2-8\kappa p}\right),\,\,p\leq -1-\frac{3\kappa}{8}.
\end{eqnarray} 
\end{stat}

 For the second order moment case $p=2$, and for the special cases $m=1$, $\kappa=0,2,6$ or $m=2$, $\kappa=4$, the expressions \eqref{beta} and \eqref{betam} above agree with the results stated in Corollary \ref{corint}. \textcolor{black}{The rightmost expression in \eqref{beta}, i.e., $B_{m=1}(p,\kappa)$ in \eqref{Bm},  was conjectured  in Ref. \cite{IL} (see also \cite{2012arXiv1203.2756L,2013JSMTE..04..007L}); as we shall show in Sections \ref{DerivSec} and \ref{derivative}, it is directly related to the radial SLE derivative exponents introduced in Ref. \cite{MR2002m:60159b}, and to the (non-standard) multifractal tip exponents obtained in Ref. \cite{MR2112128}.} 

As mentioned above, there exists a special point \cite{BS} \begin{equation}\label{pkappa1}p=p(\kappa)=p_1(\kappa):=\frac{(6+\kappa)(2+\kappa)}{8\kappa},\end{equation} where an exact expression can be found for $\mathbb E\big[|f_0(z)|^p\big]$ (Theorem \ref{main1});  more generally there exists a series of special points 
\begin{equation}\label{pkappam}
p=p_m(\kappa):=\frac{m(2m+4+\kappa)(2+\kappa)}{2(m+1)^2\kappa},\,\,\, m\geq 1,\end{equation}
 where the $p$-th moment of the $m$-fold transform, $\mathbb E\big[|h_0^{(m)}(z)|^p\big]$, is found in an exact form (Theorems \ref{main3} and \ref{main4}). Note that $p_m^*(\kappa)\leq p_m(\kappa), \forall \kappa \geq 0$ and $p_m^{**}(\kappa)\leq p_m(\kappa), \forall \kappa \geq \kappa_m$.

In this setting, we rigorously prove  the following 
\begin{theo}\label{theoMFc}
The average integral means spectrum $\beta(p,\kappa):=\beta_1(p,\kappa)$ of the unbounded inner whole-plane SLE$_\kappa$ \eqref{loewner} \eqref{SLE} has a \textit{phase transition} at $p^*(\kappa)$ \eqref{pstar} and a special point at $p(\kappa)$ \eqref{pkappa1}
\textcolor{black}{ \begin{eqnarray}\nonumber
\beta(p,\kappa)\begin{cases}=\beta_0(p,\kappa),\,\,\,0\leq p\leq p^*(\kappa);\\ 
\geq 3p-\frac{1}{2}-\frac{1}{2}\sqrt{1+2\kappa p}>\beta_0(p,\kappa),\,\,\,p^*(\kappa) <p<p(\kappa);\\ 
= 3p-\frac{1}{2}-\frac{1}{2}\sqrt{1+2\kappa p}=\frac{(6+\kappa)^2}{8\kappa},\,\,\,p=p(\kappa);\\ 
\leq 3p-\frac{1}{2}-\frac{1}{2}\sqrt{1+2\kappa p},\,\,\,p(\kappa)<p.
\end{cases}
\end{eqnarray}}
\end{theo}
\textcolor{black}{Actually, using a duality method explained in Section \ref{secondsolrig}, we prove a stronger result in the domain $[p^*(\kappa), p(\kappa)]$:} 
 \begin{theo}\label{theoMFrigorous}
The average integral means spectrum $\beta(p,\kappa):=\beta_1(p,\kappa)$ of the unbounded inner whole-plane SLE$_\kappa$ \eqref{loewner} \eqref{SLE} has a \textit{phase transition} at $p^*(\kappa)$ \eqref{pstar} and a special point at $p(\kappa)$ \eqref{pkappa1}, such that (with $\hat p(\kappa):=1+\kappa/2$)
\textcolor{black}{\begin{eqnarray}\nonumber
\beta(p,\kappa)\begin{cases}=\beta_0(p,\kappa),\,\,\,0\leq p\leq p^*(\kappa);\\ 
 =3p-\frac{1}{2}-\frac{1}{2}\sqrt{1+2\kappa p},\,\,\,p^*(\kappa) \leq p \leq \min\{p(\kappa),\hat p(\kappa)\};\\ 
 \geq 3p-\frac{1}{2}-\frac{1}{2}\sqrt{1+2\kappa p},\,\,\, \min\{p(\kappa),\hat p(\kappa)\}\leq  p\leq p(\kappa);\\ 
= 3p-\frac{1}{2}-\frac{1}{2}\sqrt{1+2\kappa p}=\frac{(6+\kappa)^2}{8\kappa},\,\,\,p=p(\kappa);\\ 
\leq 3p-\frac{1}{2}-\frac{1}{2}\sqrt{1+2\kappa p},\,\,\,p(\kappa)<p.
\end{cases}
\end{eqnarray}}
\end{theo}
\begin{rkk}
\textcolor{black}{The duality method we use works only in the domain $p\leq p(\kappa)$. The  presence of the further quantity $\min\{p(\kappa),\hat p(\kappa)\}$ is linked to the possible occurrence of a tip spectrum at $\hat p(\kappa)=1+\kappa/2$ in this duality method. Note that $p(\kappa)=\frac{6+\kappa}{4\kappa}\times \frac{2+\kappa}{2}$ so that 
$\min\{p(\kappa),\hat p(\kappa)\}=\hat p(\kappa)$ for $\kappa\leq 2$, whereas $\min\{p(\kappa),\hat p(\kappa)\}=p(\kappa)$ for $\kappa\geq 2$.}
\end{rkk}
\begin{theo}\label{theoMFcm}
Similarly, the average integral means spectrum of the $m$-fold transform of the unbounded whole-plane SLE map has a phase transition at $p^*_m(\kappa)$ \eqref{pmstar} and a special point at $p_m(\kappa)$  \eqref{pkappam}, such that for $1\leq m\leq 3,\forall \kappa$, or $m\geq 4, \kappa\leq \kappa_m=4(m+3)/(m-3)$, \textcolor{black}{\begin{eqnarray}\nonumber
\beta_m(p,\kappa)\begin{cases}=\beta_0(p,\kappa),\,\,\,0\leq p\leq p_m^*(\kappa);\\ 
\geq B_m(p,\kappa)=\left(1+\frac{2}{m}\right)p-\frac{1}{2}-\frac{1}{2}\sqrt{1+\frac{2\kappa p}{m}} >\beta_0(p,\kappa),\,\,\,p_m^*(\kappa)< p<p_m(\kappa);\\ 
=B_m(p,\kappa)=\frac{(2m+4+\kappa)^2}{2(m+1)^2\kappa},\,\,\,p= p_m(\kappa);\\ 
\leq B_m(p,\kappa),\,\,\, p_m(\kappa)<p.
\end{cases}\end{eqnarray}}
For $m\geq 4$ and $\kappa \geq \kappa_m$, the average integral means spectrum of the $m$-fold transform of the unbounded whole-plane SLE$_\kappa$ map has a phase transition at $p^{**}_m(\kappa)$ \eqref{pdoublestar0}
\textcolor{black}{\begin{eqnarray}\nonumber
\beta_m(p,\kappa)\begin{cases}=\bar \beta_0(p,\kappa),\,\,\,0\leq p\leq p_m^{**}(\kappa);\\ 
\geq B_m(p,\kappa) >\bar\beta_0(p,\kappa),\,\,\,p_m^{**}(\kappa)< p<p_m(\kappa);\\ 
= B_m(p,\kappa)=\frac{(2m+4+\kappa)^2}{2(m+1)^2\kappa},\,\,\,p= p_m(\kappa);\\ 
\leq B_m(p,\kappa),\,\,\, p_m(\kappa)<p.
\end{cases}\end{eqnarray}}
\end{theo}
\begin{rkk}
The phase transition point $p_m^*(\kappa)$ \eqref{pstar} is lower, for $1\leq m\leq 3, \forall \kappa$, or for $ 4\leq m, \kappa\leq \kappa_m$, than the special value $p_0^*(\kappa)=3(4+\kappa)^2/32\kappa$ after which the BS spectrum becomes  linear in $p$ \cite{BS}. The phase transition specific to the unbounded whole-plane SLE$_\kappa$ then supersedes the usual phase transition towards a linear behavior. For $m\geq 4, \kappa\geq \kappa_m$, the situation is reversed, and the linear transition at $p_0^*(\kappa)$ happens before the one specific to the unboundedness of the inner whole-plane SLE$_\kappa$, which thus takes place at the higher value $p_m^{**}(\kappa)$ \eqref{pdoublestar0}.
\end{rkk}
\begin{rkk}\label{pneg}
 As mentioned above, for $p\leq 0$, all the average spectra $\beta(p,\kappa):=\beta_1(p,\kappa),\beta_m(p,\kappa)$, $m\geq 2$, co\"incide with the one derived by Beliaev and Smirnov \cite{BS}, which equals $\beta_0(p,\kappa)$ down to the  phase transition for $p\leq -1-3\kappa/8$ to the  \textit{tip spectrum} \eqref{tip}-\eqref{beta0tip}, as predicted in the multifractal formalism in \cite{PhysRevLett.88.055506} and  proven in an almost sure sense in \cite{0911.3983}. 
\end{rkk}
 In the  $\kappa\to 0$ limit, one has $\lim_{\kappa\to 0}\beta_0(p,\kappa)=0$, and the  spectra:
 \begin{eqnarray}\label{beta0}
\beta(p,\kappa=0)&=&\max\{0,3p-1\},\\ \label{betatilde0}
\beta_2(p,\kappa=0)&=& \max \{0,2p-1 \},\\ \label{betam0}
\beta_m(p,\kappa=0)&=&\max \{0,(1+2/m)p-1\},\end{eqnarray} 
co\"{\i}ncide  with those directly derived (for $p\geq 0$) for the Koebe function \eqref{koebe} and its $m$-fold transforms. 

\textcolor{black}{
The above results are reminiscent of the difference between universal integral means spectra for bounded or unbounded conformal maps \cite{Pommerenke}. Makarov \cite{Makanaliz} has indeed shown that \eqref{beta0}, \eqref{betatilde0}, and \eqref{betam0} give the  universal spectra for general conformal maps (for $p$ large enough). Theorems \ref{theoMFc} and \ref{theoMFcm} show that very similar expressions appear in the whole-plane SLE case.}  

Note also that these integral means spectra at $p=2$ give the asymptotic behaviors of the coefficient second moments: $\mathbb E |a_n|^2 \asymp n^{\beta(2)-3}$ and  $\mathbb E |b_{2n+1}|^2 \asymp n^{\beta_2(2)-3}$ for $n\to \infty$, with $\beta(2)=(11-\sqrt{1+4\kappa})/2$ [for $\kappa\leq 30$] and $\beta_2(2)=(7-\sqrt{1+2\kappa})/2)/2$ [for $\kappa \leq 24$].

Another interesting random variable is the {\it area} of the image of the unit disk
$$\int\int_{\bb D}\I f'(z)\I^2 dxdy=\pi\sum_{n=1}^\infty n\I a_n\I^2.$$
The expectation of this quantity  thus converges for $\beta(2)<1$, i.e., only for $\kappa > 20$, even though  the SLE trace is no longer a simple curve as soon as  $\kappa >4$. Similarly, for the odd case, convergence of the area is obtained for $\beta_2(2) < 1$, hence for $\kappa > 12$. 

\subsection{Derivative exponents}\label{DerivSec}
In the so-called multifractal formalism \cite{FP,1986PhRvA..33.1141H,hentschelP,Mand}, the integral means spectrum  \eqref{betadef}, or its in expectation version $\beta(p)$ \eqref{betavdef}, are related by various (Legendre) transforms to other multifractal spectra, such as the so-called packing spectrum,  the moment spectrum, often written $\tau(q)$ or $\tau(n)$, the generalized dimension spectrum $D(n):=\tau(n)/(n-1)$, and the celebrated multifractal spectrum $f(\alpha)$ (see, e.g., Refs. \cite{BS,MR2112128,1986PhRvA..33.1141H,Makanaliz}). 
 Of particular interest here is the \textit{packing spectrum} \cite{Makanaliz}, defined as
\begin{equation}\label{packing}
s(p):=\beta(p)-p+1. 
\end{equation}
For our unbounded whole-plane SLE$_\kappa$, we have for $p\geq p^*(\kappa)$ \eqref{beta}, \eqref{pstar}:
\begin{eqnarray}\label{betaunb}
\beta(p,\kappa)&=& 3p-\frac{1}{2}-\frac{1}{2}\sqrt{1+2\kappa p},\\ \label{sp}
s(p,\kappa)&=&\beta(p,\kappa)-p+1\\ \label{sunb}&=&2p+\frac{1}{2}-\frac{1}{2}\sqrt{1+2\kappa p}.
\end{eqnarray}
It is then particularly interesting to consider, \textcolor{black}{for each fixed $\kappa$,} the \textit{inverse function} of $s(p,\kappa)$: $p(s,\kappa)=s^{-1}(s,\kappa)$. It has two branches,\begin{eqnarray}\nonumber\label{ppm(s)}
p_{\pm}(s,\kappa):=\frac{s}{2} +\frac{1}{16}\left(\kappa-4\pm\sqrt{(4-\kappa)^2+16 \kappa s}\right),
\end{eqnarray}
which are both defined for $s\geq s_{\textrm{min}}(\kappa):=-(4-\kappa)^2/16\kappa$, where they share the common value $p_{\textrm{min}}(\kappa):=p_{\pm}(s_{\textrm{min}}(\kappa),\kappa)=(\kappa-4)(\kappa+4)/32\kappa$. One has  $p_-(s) \leq p_{\textrm{min}}$, whereas $p_+(s) \geq p_{\textrm{min}}$. Since $p^*(\kappa)>p_{\textrm{min}}(\kappa)$,  the determination that contains 
the ``physical'' branch $p\in[p^*(\kappa),+\infty)$ is  $p_+(s,\kappa)$, hence we retain for $s\geq s(p^*(\kappa),\kappa)$
\begin{eqnarray}\label{ps}
p(s,\kappa)=s^{-1}(s,\kappa)&=&\frac{s}{2} +\frac{1}{16}\left(\kappa-4+\sqrt{(4-\kappa)^2+16 \kappa s}\right)\\ \nonumber
&=&\frac{s}{2} +\frac{\kappa}{8}\mathcal U^{-1}_\kappa(s),\\ \label{Un}
\mathcal U^{-1}_\kappa(s)&:=&\frac{1}{2\kappa}\left(\kappa-4+\sqrt{(4-\kappa)^2+16 \kappa s}\right).
\end{eqnarray}
These expressions then lead  to the following striking observation:
\begin{rkk}\label{tipexponents}The same expression \eqref{ps} appeared earlier  in the set of tip multifractal exponents $x(1\wedge n)$ in Ref. \cite{MR2112128} [Eq. (12.19)], and is identical  (for $n= s$) to  $\lambda_\kappa(1\wedge n):=x(1\wedge n)-x_1$,  $x_1:=(6-\kappa)(\kappa-2)/8\kappa$ [\cite{MR2112128}, Eq. (12.37)]. The bulk critical exponent $x(1\wedge n)$ corresponds geometrically to the extremity of an SLE$_\kappa$ path avoiding a packet of $n$ independent Brownian motions diffusing away from its tip, while $x_1$ is the bulk exponent of the SLE$_\kappa$ single extremity. These local tip exponents differ from the ones associated to the SLE tip multifractal spectrum \eqref{tip}-\eqref{beta0tip} of Refs. \cite{BS,PhysRevLett.88.055506,0911.3983}. Eq. \eqref{ps}  for $p(s,\kappa)$ is also identical to the so-called \textit{derivative exponent} $\nu(b,\kappa)$ (for $b=s$), obtained for radial SLE$_\kappa$ in Ref. \cite{MR2002m:60159b}, Eq. (3.1).  
 \end{rkk}
 The exponents $x(1\wedge n)$ were calculated using the so-called \textit{quantum gravity} method in \cite{2000PhRvL..84.1363D,MR1964687,MR2112128}. The function $\mathcal U^{-1}_\kappa$   \eqref{Un} appears there as the inverse of the so-called Kniznik-Polyakov-Zamolodchikov (KPZ) relation \cite{MR947880} (see also Refs. \cite{MR981529,MR1005268}), which was recently proven rigorously in a probabilistic framework \cite{springerlink:10.1007/s00222-010-0308-1bis,2009arXiv0901.0277D,PhysRevLett.107.131305}. (See also Ref. \cite{rhodes-2008}.) Here it maps 
   a critical exponent in the complex (half-)plane $\mathbb H$, $n(=s)$, corresponding to the boundary scaling behavior of a packet of $n$ independent Brownian motions, to its quantum gravity counterpart  on a random surface coupled to SLE$_\kappa$, $\mathcal U^{-1}_\kappa(s)$.
 
 The \textit{derivative exponents} $\nu(b,\kappa)=x(1\wedge b)-x_1$ also describe the scaling behavior  of the moments of order $b(=s)$ of the modulus of the derivative of the forward \textit{radial} SLE$_\kappa$ map $g_t$  in $\mathbb D$ at large time $t$ \cite{MR2002m:60159b}. (See also Refs. \cite{MR2002m:60159a,MR2153402}.) In Section \ref{derivative}, we give a heuristic explanation of why  \textit{the inverse function $p(s,\kappa)=s^{-1}(s,\kappa)$ of the packing spectrum $s(p,\kappa)$ \eqref{sunb} of the \textit{unbounded} whole-plane SLE  co\"{\i}ncides with the derivative exponents $\nu(s,\kappa)$ of radial SLE}.\\


Recall that $p(s,\kappa)$ is analytically defined only for $s\geq s_{\textrm{min}}(\kappa):=-(4-\kappa)^2/16\kappa$. For $\kappa\leq 4$, one has $p(s,\kappa)\geq 0$ for $s\geq 0$. For $\kappa\geq 4$,  $p\in [p_{\textrm{min}}(\kappa),+\infty]$, with $p_{\textrm{min}}(\kappa)=p(s_{\textrm{min}}(\kappa),\kappa)=(\kappa^2-16)/32\kappa$.   For $s=0$, $p(0,\kappa)=(\kappa-4+|\kappa-4|)/16$, hence $p(0,\kappa)=0$ for $\kappa\leq 4$, and $p(0,\kappa)=(\kappa-4)/8$ for $\kappa\geq 4$.\\

Consider now the $m$-fold version \eqref{defmfold} of the unbounded inner whole-plane SLE. For $1\leq m\leq 3, \forall \kappa$, or for $m\geq 4$ and $\kappa\leq \kappa_m$ \eqref{km0}, we have for $p\geq p_m^*(\kappa)$ \eqref{pmstar} the average  integral means and packing spectra
\begin{eqnarray}\label{betamunb}
\beta_m(p,\kappa)&=& B_m(p,\kappa)=\left(2+\frac{1}{m}\right)p-\frac{1}{2}-\frac{1}{2}\sqrt{1+2\kappa \frac{p}{m}},\\ \nonumber
s_m(p,\kappa)&:=&\beta_m(p,\kappa)-p+1\\ \label{smunb}&=&2\frac{p}{m}+\frac{1}{2}-\frac{1}{2}\sqrt{1+2\kappa \frac{p}{m}}\\ \label{sms1}
&=&s\left(\frac{p}{m},\kappa\right).
\end{eqnarray}
The inverse function, $p_m(s,\kappa):=s_m^{-1}(s,\kappa)$, is therefore simply 
\begin{eqnarray}
p_m(s,\kappa)=s_m^{-1}(s,\kappa)=m\, p(s,\kappa),\,\,\,s\in[s(p_m^*(\kappa)),\infty),\end{eqnarray}
where $p(s,\kappa)$ is the inverse function \eqref{ps} of $s(p,\kappa)$ for $m=1$.

For $m\geq 4$ and $\kappa\geq \kappa_m$, we have the successive integral means and packing spectra
\begin{eqnarray}\label{betamsucunb}
\beta_m(p,\kappa)&=&p-\frac{(4+\kappa)^2}{16\kappa},\,\,\,p_0^*(\kappa)=\frac{3(4+\kappa)^2}{32\kappa}\leq p\leq m\frac{\kappa^2-4}{32\kappa}=p_m^{**}(\kappa),\\ \nonumber
s_m(p,\kappa)&=&\beta_m(p,\kappa)-p+1\\ \label{smin}
&=&-\frac{(4-\kappa)^2}{16\kappa}=s_{\textrm{min}}(\kappa);\\
\beta_m(p,\kappa)&=& B_m(p,\kappa),\,\,\,p\geq m\frac{\kappa^2-4}{32\kappa},\\ \nonumber
s_m(p,\kappa)&=&\beta_m(p,\kappa)-p+1\\ \label{smsucunb}&=&2\frac{p}{m}+\frac{1}{2}-\frac{1}{2}\sqrt{1+2\kappa \frac{p}{m}}
=s\left(\frac{p}{m},\kappa\right).
\end{eqnarray}
Observe that $p_m^{**}(\kappa)=m\, p_{\textrm{min}}(\kappa)=m\, p(s_{\textrm{min}}(\kappa),\kappa)$, so that the inverse function of $s_m(p,\kappa)$, $s_m^{-1}(s,\kappa)$, is now defined in the whole range  $s\geq s_{\textrm{min}}(\kappa)$, and is given  by
\begin{eqnarray}
&&p_m(s,\kappa)=s_m^{-1}(s,\kappa)=m\, p(s,\kappa)\in [m p_{\textrm{min}}(\kappa),\infty), \,\,\,s\in[s_{\textrm{min}}(\kappa),\infty).
\end{eqnarray}
\textcolor{black}{\subsection{Organization}
This article is organized as follows:\\$\bullet$ Section \ref{analytic} deals with the computation at low orders of the coefficients $a_n$ \eqref{defcoeff} of the whole-plane SLE or LLE  maps, or of the coefficients $b_{2n+1}$   \eqref{defcoeffb} of their oddified versions \eqref{defoddified}. This is followed by the evaluation of the single or square expectations of these coefficients. Computer experiments,  symbolic up to order $n=8$, and numerical up to order $n=19$, complete this study.\\
$\bullet$ Section \ref{proofs} deals with the proofs of Theorems \ref{theolle0} and \ref{theolle}. Subsection \ref{subsecexp} establishes Theorems \ref{theoEan} and \ref{theoEbn}, which together constitute Theorem \ref{theolle0}. Subsection \ref{subsecDeriv} deals with the moments of the derivative of the whole-plane SLE map, and establishes the corresponding  Beliaev--Smirnov equation. Special solutions are given by Theorem \ref{main1} and its Corollary \ref{F6-2}, thereby establishing in the  SLE$_{\kappa=6,2}$ case results \textit{(i)} and \textit{(ii)} of Theorem \ref{theolle}, while the same results are extended to the LLE case through Theorem \ref{levy} followed by Remark \ref{theo3.3}.  In Subsection  \ref{subsecodd}, similar  results are proved for the oddified whole-plane Loewner map \eqref{defoddified}.  The proof of result \textit{(iii)} of Theorem \ref{theolle} is obtained in Corollary \ref{F06-2} of Theorem \ref{main3} in the SLE$_{\kappa=4}$  case, and in Proposition \ref{oddeta1} in the LLE case. All these results, namely the existence of a Beliaev--Smirnov-like equation and of special solutions thereof, which yield specific moments in a closed form, are generalized to the $m$-fold Loewner maps \eqref{defmfold} in Subsection \ref{subsecm-fold} .\\ 
$\bullet$ Section \ref{multifractal} deals with the multifractal integral means spectrum of SLE. Subsection \ref{introductionI} describes the general properties of the SLE's harmonic measure spectra, as well as the corresponding universal  spectra. The integral means spectrum of whole-plane SLE is studied in great detail in Subsection  \ref{IMS1}, leading to the proof of Theorem \ref{theoMFc}. The general Theorem \ref{theoMFcm} for $m$-fold whole-plane SLE maps is established in Subsection \ref{subsecSpecm}. The relationship of the novel spectrum for unbounded whole-plane SLE to the so-called derivative exponents of radial SLE is explained in Subsection \ref{derivative}.\\
$\bullet$  Finally, Section  \ref{Appendices}  is comprised of several appendices. The  history of the Bieberbach conjecture is briefly recalled in Subsection \ref{appendix}; some coefficient computations are given in  Subsection \ref{appendcoeff}; a proof of Makarov's Theorem \ref{FMcGo} for the universal spectrum of oddified maps, that parallels that of the Feng-MacGregor Theorem \ref{FMcG}, is given in Subsection \ref{McGo}.}
 
\section*{Acknowledgements}
It is a pleasure to thank Dmitry Beliaev, Michael Benedicks and Steffen Rohde for discussions, and David Kosower for a reading of the manuscript. BD also wishes to thank the MSRI at Berkeley for its kind hospitality during the Program ``Random Spatial Processes'' (January 9, 2012 to May 18, 2012). \textcolor{black}{We thank the referee for a thorough and critical reading of the manuscript, and for numerous suggestions. After completing this work, we learned of related work in Refs. \cite{2012arXiv1203.2756L,2013JSMTE..04..007L,2013arXiv1301.6508L}, which use a different approach.}

\section{Coefficient estimates}\label{analytic}
\subsection{Computation of $a_n$ and $\mathbb E(\vert a_n\vert^2)$ for small $n$}\label{an2}
\subsubsection{Loewner's method}In this paragraph we perform computations for  general Loewner-L\'evy processes. 
 Let us recall that   
\begin{equation}\label{ftexpansion}
f_t(z)=e^t\big(z+\sum_{n\geq 2}a_n(t)z^n\big).
\end{equation}
By expanding both sides of Loewner's equation \eqref{loewner} as power series,  and identifying coefficients, leads one to the set of equations 
\begin{eqnarray}\label{dotaa}
\dot{a}_n(t)-(n-1)a_n(t)=2\sum_{p=1}^{n-1}(n-p)a_{n-p}(t) \bar \lambda^p(t)=2\sum_{k=1}^{n-1}ka_k(t) \bar \lambda^{n-k}(t), \, n\geq 2; 
 \end{eqnarray}where $a_1=1$; the dot means a $t$-derivative, and  $\bar \lambda(t)=1/\lambda(t)$. Specifying for $n=2,3$ gives
\begin{eqnarray}\label{a2}\dot{a}_2-a_2&=&2\bar\lambda ,  \\ \label{a3}
\dot{a}_3-2a_3&=&4a_2\bar{\lambda}+2\bar{\lambda}^2.
\end{eqnarray}
The first differential equation \eqref{a2} (together with the uniform bound, $\forall t \geq 0, |a_2(t)|\leq C_2<+\infty$; see Remark \ref{weakbiber})  yields 
\begin{equation}\label{a2bis} a_2(t)=-2e^t\int_t^{+\infty}e^{-s}\bar{\lambda}(s)ds.
\end{equation}
In a similar way, the second one \eqref{a3}  leads to
\begin{equation}\nonumber\label{a3bis}a_3(t)=-4e^{2t}\int_t^{+\infty} e^{-2s}a_2(s)\bar{\lambda}(s)ds-2e^{2t}\int_t^{+\infty} e^{-2s}\bar{\lambda}^2(s)ds,
\end{equation}
The first integral invoves $\int_t^\infty u_2(s){\dot u}_2(s)ds=-u_2^2(t)/2$,  where $u_2(s):=e^{-s}a_2(s)$. The formula for $a_3$ then reduces to 
\begin{equation}\label{a3ter}a_3(t)=4e^{2t}\left(\int_t^{+\infty} e^{-s}\bar{\lambda}(s)ds\right)^2-2e^{2t}\int_t^{+\infty} e^{-2s}\bar{\lambda}^2(s)ds.
\end{equation}
\subsubsection{Quadratic coefficients}\label{quad}
\begin{prop} \label{theo-a2} For L\'evy--Loewner processes, we have, setting here $a_2:=a_2(0)$, $$\mathbb E(| a_2|^2)=\Re\left(\frac{4}{1+\eta_1}\right).$$
\end{prop}
\textcolor{black}{\demo
Using \eqref{a2bis}, we write 
$$\I a_2(0)\I^2=4\int_0^\infty\int_0^\infty ds du\, e^{-(s+u)} e^{i(L_{u}-L_{s})}.$$
Taking care of the relative order of $s$ and $u$, the characteristic function \eqref{Levychar} of $L_u-L_{s}$ is 
$$ \mathbb E\left[e^{i(L_u-L_s)}\right]=\vartheta(u-s)e^{-\eta_1(u-s)} +\vartheta(s-u)e^{-\bar\eta_1(s-u)},$$
where $\vartheta$ is the Heaviside step distribution; the result follows by integration. \findemo}

 \medskip
For calculations involving the third order term $a_3$ as given by \eqref{a3ter}, and in order to avoid repetitions, we have computed at once $\mathbb E (\I a_3-\mu a_2^2\I^2)$, where $\mu$ is a real constant. The detail of the calculation is given in Appendix  in Section \ref{Appthird}. 
Let us simply state the result here.
\begin{prop} \label{theo-a3mu}If $\mu$ is a real coefficient, then
\begin{align*}
&\mathbb E(\I a_3-\mu a_2^2\I^2)\,=\\
&\Re\left(
\frac{16(1-\mu)^2(4+\eta_2)}{(1+\eta_1)(2+\eta_2)(3+\eta_1)}-\frac{16(1-\mu)(2+\eta_1)}{(1+\eta_1)(2+\eta_2)(3+\eta_1)}+\frac{2}{2+\eta_2}+\frac{8(1-\mu)(1-2\mu)}{(\overline{\eta_1}+1)({\eta_1}+3)}
\right).
\end{align*}
In the real $\eta$ case:
$$ \mathbb E(\I a_3-\mu a_2^2\I^2)=\frac{32(1-\mu)^2(3+\eta_2)-8(1-\mu)(6+2\eta_1+\eta_2)+2(1+\eta_1)(3+\eta_1)}{(1+\eta_1)(2+\eta_2)(3+\eta_1)}.$$
In the SLE case (i.e., for $\eta_\ell=\frac{\kappa}{2}\ell^2$):
$$ \mathbb E(\I a_3-\mu a_2^2\I^2)=\frac{(108-288\mu+192\mu^2)+(88-208\mu+128\mu^2)\kappa+\kappa^2}{(1+\kappa)(2+\kappa)(6+\kappa)}.$$
\end{prop}
\subsubsection{Some corollaries}
The first one gives, for $\mu=0$, the analogue of Loewner's estimate. 
\begin{cor} \label{theo-a3}For L\'evy--Loewner processes with $\eta$ real, we have
 \begin{equation}\label{a3eta}\mathbb E(\I a_3\I^2)=\frac{1}{(1+\eta_1)(3+\eta_1)}\left[24+2\frac{(\eta_1-1)(\eta_1-3)}{2+\eta_2}\right].\end{equation}
 In the SLE case: 
 $$\mathbb E(\I a_3\I^2)=\frac{108+88\kappa+\kappa^2}{(1+\kappa)(2+\kappa)(6+\kappa)}.$$
\end{cor}
\noi Notice the special role played by $\eta_1=1,3$, corresponding to $\kappa=2,6$:  the result no longer depends on $\eta_2$, and equals  $3$ and $1$ respectively.

The second corollary shows that there is no Fekete--Szeg\H o counter-example in the SLE family \textcolor{black}{in the expectation sense.} 
To $f:=f_0$, an $\textrm{SLE}_\kappa$ whole-plane map, we associate its oddified function  as above, that is $h_0(z)=z\sqrt{{f(z^2)}/{z^2}}=z+b_3z^3+b_5z^5+\cdots$  
An easy computation gives $b_5=\frac{1}{2}\left(a_3-\frac{1}{4}a_2^2\right).$ Setting $\mu=\frac{1}{4}$ in the above proposition gives 
$$\mathbb E(\I b_5\I^2)=\Re\left(
\frac{18+9\eta_2-4\eta_1+2\eta_1^2}{4(1+\eta_1)(2+\eta_2)(3+\eta_1)}+\frac{3}{4}\frac{1}{(1+\overline{\eta_1})(3+\eta_1)}
\right).$$ 
In the case of a real $\eta$,  
$$\mathbb E(\I b_5\I^2)=\frac{6+3\eta_2-\eta_1+\eta^2_1/2}{(1+\eta_1)(3+\eta_1)(2+\eta_2)} \left[=
\frac{12+44\kappa+\kappa^2}{(1+\kappa)(2+\kappa)(6+\kappa)}\right],$$ where the last expression has been specified for the SLE case, and  is always less than or equal to $1$ (equality holding only for $\kappa=0$).

Consider the {\it Schwarzian derivative}, $S(z):={f^{[3]}(z)}/{f'(z)}-({3}/{2})({f'^{[2]}(z)}/{f'(z)})^2.$ One obtains $S(0)=6(a_3-a_2^2)$, corresponding to $\mu=1$,  and giving the expected values
$ \mathbb E [\I S(0)\I^2]={72}/({2+\eta_2}),$ and for SLE, 
 $ \mathbb E [\I S(0)\I^2]={36}/({1+\kappa}).$\\
 
\noi A few comments are in order here:\\
-- We noticed that $\mathbb E(\I a_2\I^2)=\mathbb E(\I a_3\I^2)=1$ for $\kappa=6$. We return to this in the next sections after performing some computer experiments.\\
-- For all values of $\kappa$,  $\mathbb E(\I b_5\I^2)\leq 1$: \textcolor{black}{therefore, in this expectation sense,} there is no Fekete--Szeg\H o counterexample in the SLE-family. Using the Schoenberg   property of the L\'evy symbol $\eta$ \cite{applebaum}, it can also be seen that there is no counterexample in expectation for a general L\'evy--Loewner process with real $\eta$. The question remains open for higher order terms or higher moments; this will be studied elsewhere.\\ 
-- It is known that $\I S(0)\I\leq 6$ whenever $f$ is injective. Conversely, if $(1-\I z\I^2)^2\I S(z)\I\leq 2,$ then $f$ is injective; \textcolor{black}{here, the corresponding inequality $|S(0)|\leq 2$ holds  in the sense that $\mathbb E(|S(0)|^2)\leq 4$ for $\kappa\geq 8$.}

\subsubsection{Next order}
The  quadratic expectation of the next order coefficient, $\mathbb E(\left|a_4^2\right|)$, can still be computed by hand, which yields
\begin{eqnarray}\nonumber\mathbb E\left(\left|a_4\right|^2\right)&=&\,\frac{4!2^3}{(\eta_1+1)(\eta_1+3)(\eta_1+5)}\\ \label{a4eta} &&+\frac{4(\eta_1-1)(\eta_1-3)\eta_2(\eta_2-4)(\eta_1+3)}{3(\eta_1+1)(\eta_1+3)(\eta_1+5)(\eta_2+2)(\eta_2+4)(\eta_3+3)},
\end{eqnarray}
and for SLE,
\begin{eqnarray*}\mathbb E\left(\left|a_4\right|^2\right)=\,\frac{8}{9}\frac{\kappa^5 + 104 \kappa ^4 + 4576 \kappa ^3 + 18288 \kappa ^2 + 22896 \kappa + 8640}{(\kappa + 10)(3 \kappa + 2)(\kappa + 6)(\kappa + 1)(\kappa + 2)^2}.
\end{eqnarray*}
 
 The results \eqref{a3eta} and \eqref{a4eta}  obtained so far for $\mathbb E (|a_n|^2)$, $n=3,4$, call for the following observations.
 \begin{rkkk}\label{mainobs}
 --After the first term in the expression for $\mathbb E (|a_n|^2)$, one notes the presence in numerators of the common factors $(\eta_1-1)(\eta_1-3)$, thus vanishing for $\eta_1=1$ or $3$. For $\eta_1=3$ (or $\kappa=6$), the first term, thus $\mathbb E (|a_n|^2)$ itself, equals  $1$; for $\eta_1=1$ (or $\kappa=2$) it equals $n$. \textcolor{black}{We checked explicitly that this holds in symbolic computations up to $n=8$, and in numerical ones up to $n=19$.  (see Appendix \ref{Apphigh} and Eq. \eqref{a5square}); the validity of these observations  for all $n$ was first conjectured in \cite{Hal-DNNZ}.} \\
 --Somehow surprisingly, all the coefficients of the polynomial expansions in $\kappa$ are positive.\\
--For $\kappa \to \infty$ (or $\eta\to\infty$), these expectations vanish as $\kappa^{-1}$.
\end{rkkk}
 {\it All these patterns will be confirmed at higher orders}, to which we now turn. 

\subsection{Computational experiments}\label{computation}
\begin{figure}[tb]
\begin{center}
\includegraphics[angle=0,width=.83290\linewidth]{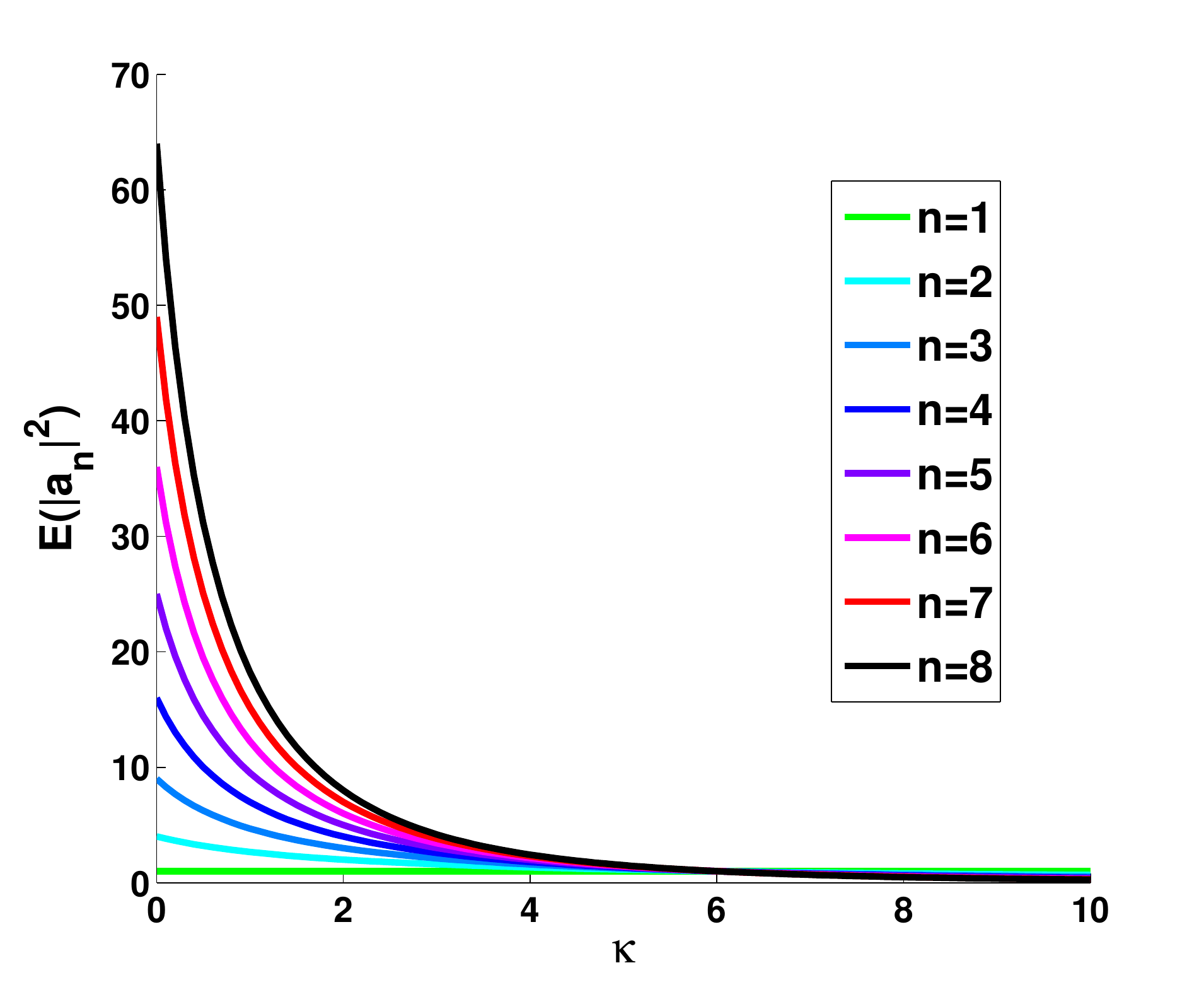}
\caption{{\it Graphs of the $\textrm{SLE}_\kappa$ map $\kappa\mapsto \mathbb E(\I a_n\I^2)$ for $n=1,\cdots,8$.}}
\label{fig1}
\end{center}
\end{figure}
\begin{figure}[tb]
\begin{center}
\includegraphics[angle=0,width=.93290\linewidth]{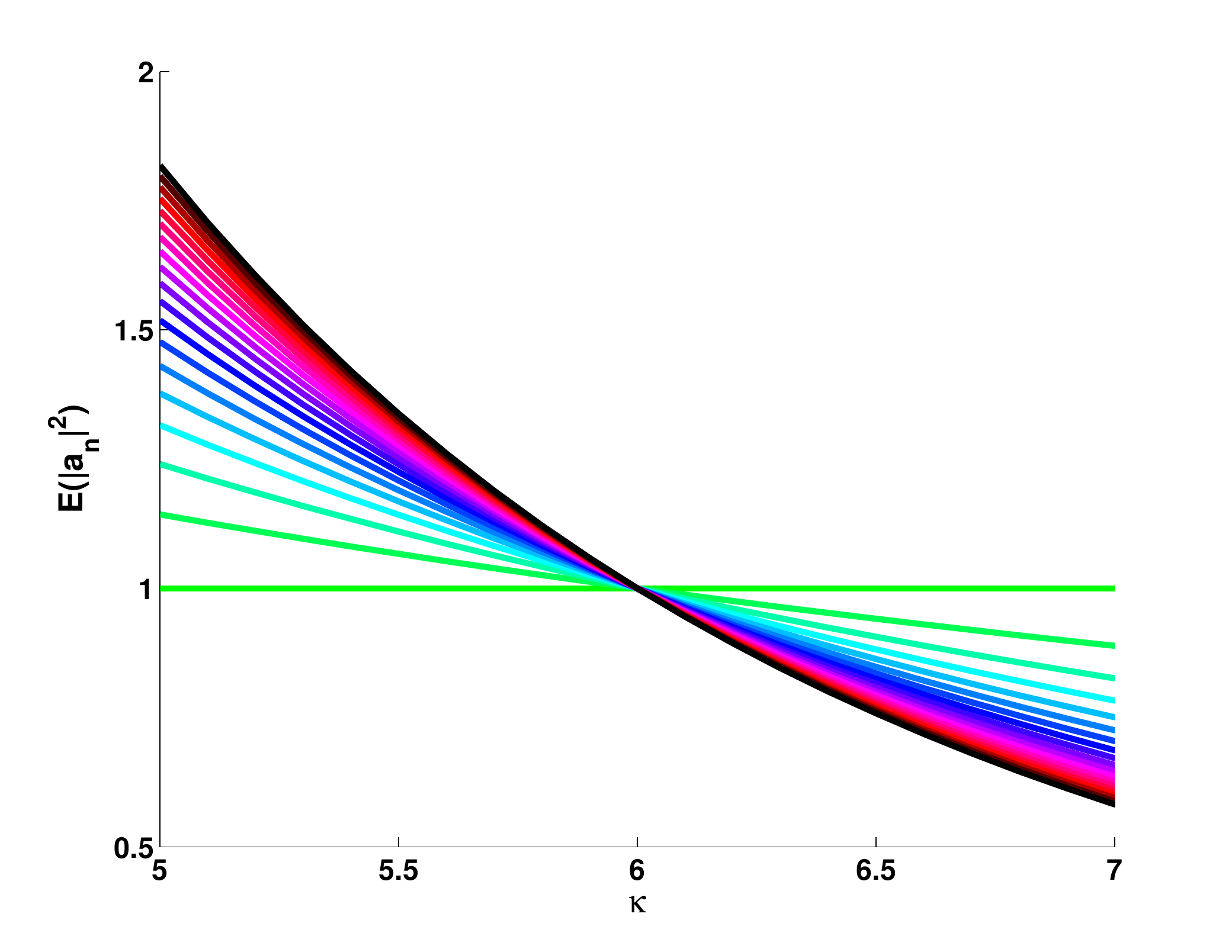}
\caption{{\it Graphs of the $\textrm{SLE}_\kappa$ map $\kappa\mapsto \mathbb E(\I a_n\I^2)$, for $n=1,\cdots,19$, with a zoom near $\kappa=6$.}}
\label{fig1zoom}
\end{center}
\end{figure}As one may see, these computations become more and more involved. 
Moreover, its seems difficult to find a closed formula for all terms. This section is devoted to the description of an algorithm that we have implemented on {\sc matlab} to compute $\mathbb E(\I a_n\I^2)$. This algorithm is divided into two parts: the first encodes the computation of $a_n$, while the second uses it to compute $\mathbb E(\I a_n\I^2)$. Since the important cases of SLE and $\alpha$-stable processes both have real L\'evy symbols $\eta$, we restrict the study to the latter case.

For the encoding of $a_n$, we observe that they are linear combinations of  successive integrals of the form 
 \begin{equation}\label{int}\int_t^\infty ds_1\,e^{-i\alpha_1L_{s_1}-\beta_1s_1}\int_{s_1}^\infty ds_2\,e^{-i\alpha_2L_{s_2}-\beta_2s_2}\ldots\int_{s_{k-1}}^\infty ds_k\, e^{-i\alpha_kL_{s_k}-\beta_ks_k}.\end{equation}
Their expectations  are encoded as
\begin{equation}\label{ab}
(\alpha_1,\beta_1)\ldots(\alpha_k,\beta_k)\quad(1\le k\le n),
\end{equation}
and are explicitly computed by using as above the strong Markov property and the L\'evy characteristic function \eqref{Levychar}:
$$ (\alpha_1,\beta_1)\ldots(\alpha_k,\beta_k)=\prod _{j=0}^{k-1}\left[\beta_k+\beta_{k-1}+\ldots+\beta_{k-j}+\eta(\alpha_k+\alpha_{k-1}+\ldots+\alpha_{k-j})\right]^{-1}.$$
Next, in order to compute $\I a_n\I^2$, we need to evaluate the expectation of products of integrals such as \eqref{int} with complex conjugate of others, that we symbolically denote by
\begin{equation}
\left[(\alpha_1,\beta_1)\ldots(\alpha_k,\beta_k);(-\alpha'_1,\beta'_1)\ldots(-\alpha'_\ell,\beta'_\ell)\right]\quad(1\le k,\ell\le n).\label{form}
\end{equation}
The product integrals may be written as a sum of $(\begin{smallmatrix} k+\ell\\k\end{smallmatrix})$ ordered integrals with $k+\ell$ variables: the $k$ first ones and the $\ell$ last ones are ordered and the number of ordered integrals corresponds to the number of ways of shuffling $k$ cards in the left hand with $\ell$ cards in the right hand. This sum is quite large and, in order to systematically compute it, we write its expectation as the sum of expectations of integrals of the form \eqref{ab}  that begin with a term of type $(\alpha_1,\beta_1)$ or with a term of type $(-\alpha'_1,\beta'_1)$, thus reducing the work to a computation at lower order. 

Using dynamic programing, we performed computations (formal up to $n=8$ and numerical up to $n=19$) on a usual computer. The results are reported in Appendix B, Section \ref{Apphigh}. They fully confirm the validity of Remarks \ref{mainobs}.   

The graphs given in Figure \ref{fig1} for the $\textrm{SLE}_\kappa$ map $\kappa\mapsto \mathbb E(\I a_n\I^2)$, for $n=1,\cdots,8$, illustrate the phenomena described above;  in particular a zoom in Fig. \ref{fig1zoom} for values of $n=1,\cdots,19$ shows the striking constant value $\mathbb E(|a_n^2|)=1$ for $\kappa=6$. Similarly, Fig. \ref{fig2} illustrates the $\textrm{SLE}_\kappa$ map $\kappa\mapsto \mathbb E(\I a_n\I^2)/n$, for $n=1,\cdots,8$, with a zoom  in Fig \ref{fig2zoom} near $\kappa=2$ where $\mathbb E(|a_n^2|)=n$, here for $n=1,\cdots,19$.  
\begin{figure}[tb]
\begin{center}
\includegraphics[angle=0,width=.93290\linewidth]{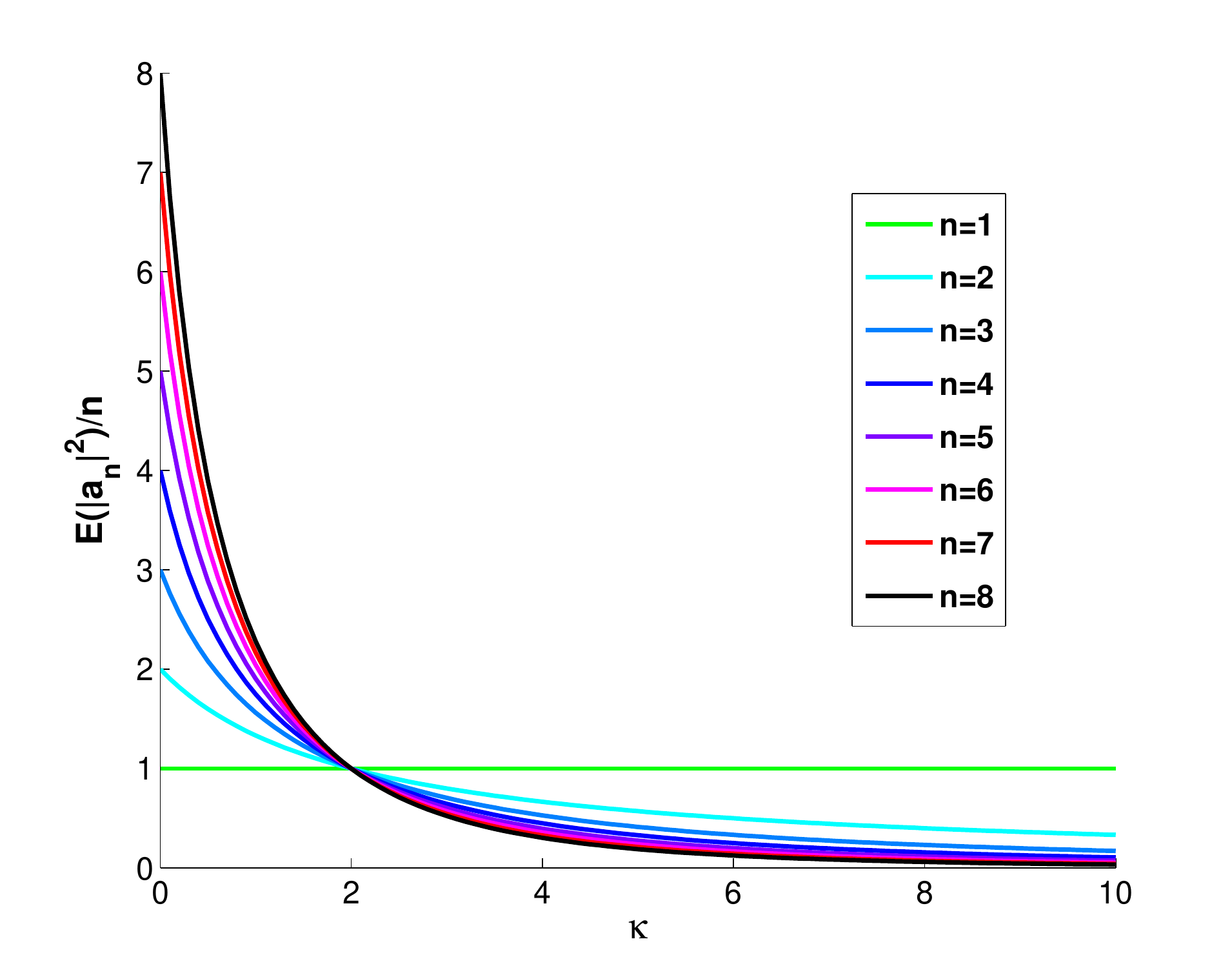}
\caption{{\it Graphs of the $\textrm{SLE}_\kappa$ map $\kappa\mapsto \mathbb E(\I a_n\I^2)/n$, for $n=1,\cdots,8$.}}
\label{fig2}
\end{center}
\end{figure}
\begin{figure}[tb]
\begin{center}
\includegraphics[angle=0,width=.93290\linewidth]{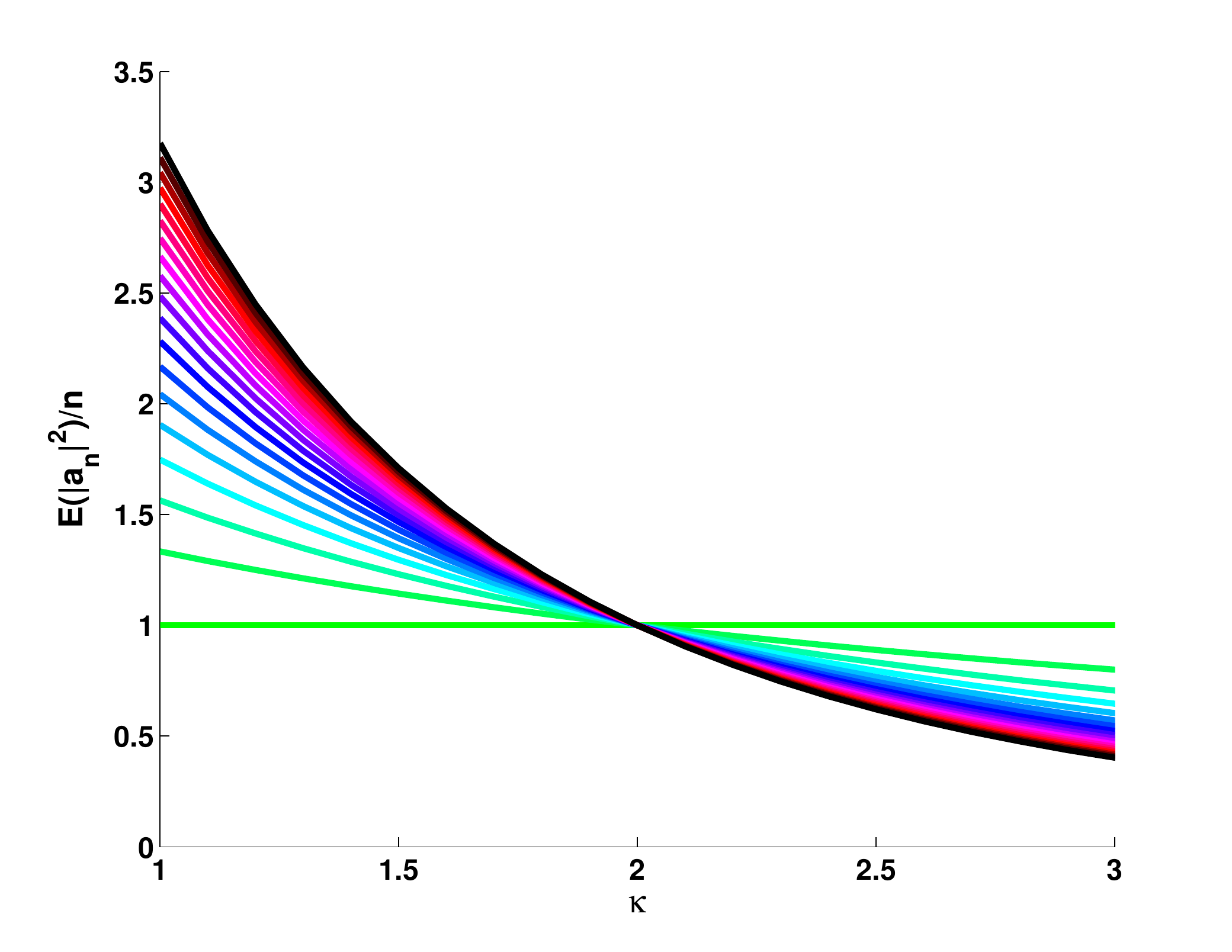}
\caption{{\it Graphs of the $\textrm{SLE}_\kappa$ map $\kappa\mapsto \mathbb E(\I a_n\I^2)/n$, for $n=1,\cdots,19$, with a zoom near $\kappa=2$.}}
\label{fig2zoom}\end{center}
\end{figure} 
\section{Theorems and Proofs}\label{proofs}
\subsection{ Expected conformal maps for L\'evy--Loewner evolutions}\label{subsecexp}
\subsubsection{Expectation of $f_0(z)$}
In this first section, we give an explicit expression for the expectations of the coefficients $a_n(t)$ of the expansion \eqref{ftexpansion} in the L\'evy--Loewner setting, thereby  obtaining the expectation of the map, $\mathbb E[f_0(z)]$, and of its derivative.

 The differential recursion \eqref{dotaa} in Section \ref{analytic} then becomes, for $\lambda(t):=e^{iL_t}$, and in terms of the auxiliary function 
$u_n(t)$, 
\begin{eqnarray}
\label{ua} u_n(t)&:=&a_n(t)e^{-(n-1)t}\\ \label{ind}
 \dot{u}_n(t)&=&2\sum_{k=1}^{n-1}k X_t^{n-k}u_k(t),
\end{eqnarray}
where $X_t$ is defined as
 \begin{equation}\label{Xt}
 X_t:=e^{-t-iL_t}.
 \end{equation}
The recursion \eqref{ind}
 can be rewritten under the simpler form:
\begin{equation}\label{indter}
\dot{u}_n=X_t[\dot{u}_{n-1}+2(n-1)u_{n-1}]. 
\end{equation}
Recall that $u_1=a_1=1$, while the next term of this recursion, as already seen in Eqs. \eqref{a2bis}, is  
\begin{equation}\label{u2}
u_2(t)=-2\int_t^{+\infty} ds X_s.
\end{equation}
Similarly,  
  we can write the general solution $u_n$, for $n\geq 2$, under the form 
\begin{equation}\label{uv}
u_n(t)=-2\int_t^{+\infty} ds X_s v_n(s),
\end{equation}
with $v_2(s)=1$, and rewrite the differential equation \eqref{indter} as an {\it integral} equation
\begin{equation}\label{inteq}
v_n(t)=X_t v_{n-1}(t)-2(n-1)\int_t^{+\infty}ds X_s v_{n-1}(s).
\end{equation}
Define then the multiplicative and integral operators $\mathcal X$ and $\mathcal J$ such that
\begin{eqnarray}\label{X}
\mathcal X v (t)&:=&X_t v(t),\\ \label{J}
\mathcal J v (t)&:=&-2\int_t^{+\infty} ds X_s v(s).
\end{eqnarray}
The solution to \eqref{u2}, \eqref{uv} and \eqref{inteq} can then be written as the operator product 
\begin{eqnarray}\nonumber
u_n&=&  \mathcal J  \circ [\mathcal X + (n-1) \mathcal J]\circ \cdots \circ (\mathcal X + 2\mathcal J )\mathds{1}\\ \label{uop}
&=&\mathcal J \prod_{k=1}^{n-2} \circ (\mathcal X + (k+1)\mathcal J ) \mathds{1},
\end{eqnarray}
 where $\mathds{1}(=v_2)$ is the constant function equal to $1$ on $\mathbb R^+$.
 
Next, recall the strong Markov property of the L\'evy process, which implies the {\it identity in law:} $\forall s\geq t, L_s \stackrel{\rm(law)}{=}L_t +\tilde L_{s-t}$, where $\tilde L_{s'}$ is an independent copy of the L\'evy process, also started at $\tilde L_0=0$. Therefore, the process $X_t$ \eqref{Xt} is, in law,
\begin{equation}\label{XXtilde}
X_s \stackrel{\rm (law)}{=}X_t \tilde X_{s-t}, \forall s\geq t
\end{equation}
where $\tilde X_{s'}:=e^{-s'-\tilde L_{s'}}, s'\geq 0$, is an independent copy of that process, with $\tilde X_{0}=1$.
The operator $\mathcal J$ \eqref{J} can then be written as 
\begin{eqnarray}\label{Jinlaw}
\mathcal J v (t)&\stackrel{\rm (law)}{=}&-2 X_t\int_0^{+\infty} ds \tilde X_{s} v(s+t)\\
&=&\mathcal X\circ \tilde {\mathcal J} v (t),
\end{eqnarray}
with $\tilde {\mathcal J} v (t):=-2\int_0^{+\infty} ds\tilde X_{s} v(s+t)$. By iteration of  the use of the Markov property,  Eq. \eqref{uop} can be rewritten as
\begin{eqnarray}\nonumber
u_n&\stackrel{\rm (law)}{=}&  \mathcal J  \circ \big[\mathcal X \big(1+ (n-2) \tilde{\mathcal J}^{[n-1]}\big)\big]\circ \cdots \circ \big[\mathcal X \big(1+ 2\tilde {\mathcal J}^{[1]}\big)\big]\mathds{1}\\ \label{uoptilde}
&\stackrel{\rm (law)}{=}&\mathcal J \prod_{k=1}^{n-2} \circ [\mathcal X \big(1+ (k+1)\tilde {\mathcal J}^{[k]}\big) \big] \mathds{1},
\end{eqnarray}
where the integral operators $\tilde {\mathcal J}^{[k]}, k=1,\cdots, n-2$, involve successive {\it independent} copies, $\tilde X^{[k]}_{s_k}, k=1,\cdots, n-2$, of the original exponential L\'evy process $X_s$. We therefore arrive at the following explicit representation of the solution \eqref{uop}
\begin{equation}\label{unXn}
u_n(t)\stackrel{\rm (law)}{=}-2\int_t^{+\infty} ds X_s^{n-1}\prod_{k=1}^{n-2}\left(1-2(k+1)\int_0^{+\infty}ds_k \big(\tilde X^{[k]}_{s_k}\big)^{k} \right).
\end{equation}

As mentioned in the introduction, the {\it conjugate} whole-plane L\'evy--Loewner evolution $e^{-iL_t} f_t\big(e^{iL_t}z\big)$ should have the {\it same law} as $f_0(z)$. At order $n$, we are thus interested in the stochastically rotated coefficients: $$e^{i(n-1)L_t} a_n(t)=(X_t)^{-(n-1)} u_n(t).$$ Using again the identity in law \eqref{XXtilde}
in \eqref{unXn}, we arrive at 
\begin{eqnarray}\label{anXn}
e^{i(n-1)L_t}a_n(t)&\stackrel{\rm (law)}{=}&-2\int_0^{+\infty} ds {\tilde X}_s^{n-1}\prod_{k=1}^{n-2}\left(1-2(k+1)\int_0^{+\infty}ds_k \big(\tilde X^{[k]}_{s_k}\big)^{k} \right)\\ \nonumber
&\stackrel{\rm (law)}{=}& a_n(0),
\end{eqnarray}
which, as it must, no longer depends of $t$. 

All factors in \eqref{anXn} involve successive independent copies of the L\'evy process, and their expectations can now be taken {\it independently}. 
Recalling the form  \eqref{Levychar} of the L\'evy characteristic function, we have 
$\mathbb E [({\tilde X}_s)^k]=e^{-(\eta_k+k)s}$. Thus 
\begin{eqnarray}\nonumber
\mathbb E [a_n(0)]&=&-2\int_0^{+\infty} ds\, \mathbb E[{\tilde X}_s^{n-1}]\prod_{k=1}^{n-2}\left(1-2(k+1)\int_0^{+\infty}ds_k \big(\mathbb E\big [\big(\tilde X^{[k]}_{s_k}\big)^{k}\big]\right)\\\label{Ean0}
&=&-2\, \frac{1}{\eta_{n-1}+n-1}\prod_{k=1}^{n-2}\left(1-\frac{2(k+1)}{\eta_{k}+k}\right).\end{eqnarray}
We finally obtain:
\begin{theo}\label{theoEan}
For $n\geq 2$, setting $a_n:=a_n(0)$, 
\begin{eqnarray} \label{Eant}
a_n(0)&\stackrel{\rm (law)}{=}&e^{i(n-1)L_t}a_n(t),\\ \label{Ean}
\mathbb E(a_n)
&=&-2\frac{\prod_{k=1}^{n-2}(\eta_{k}-k-2)}{\prod_{k=1}^{n-1}(\eta_k+k)}=\prod_{k=0}^{n-2}\frac{\eta_{k}-k-2}{\eta_{k+1}+k+1}.
\end{eqnarray}
\end{theo}
\begin{cor}\label{cor1}
The expected  conformal map $\mathbb E [f_0(z)]$ of the whole-plane L\'evy--Loewner evolution, in the setting of Theorem \ref{theolle0}, is  polynomial of degree $k+1$ if there exists a positive $k$ such that $\eta_k=k+2$, has radius of convergence $1$ for an $\alpha$-stable L\'evy process of symbol  $\eta_n={\kappa}n^\alpha/2$, $\alpha\in (0,2]$, except for the Cauchy process $\alpha=1,\kappa=2$, where $\mathbb E[f_0(z)]=ze^{-z}$.
\end{cor}
\begin{proof}
From Theorem \ref{theoEan}, $\mathbb  E [f_0(z)]$ is polynomial if there exists $k \in \mathbb N$ such that $\eta_k=k+2$, as all $\mathbb E(a_n)$ then vanish for $n\geq k+2$. Otherwise, use D'Alembert's criterion, applied here to $$\lim_{n\to \infty} \frac{\left\vert\mathbb E(a_{n+1})\right\vert}{\left\vert\mathbb E(a_{n})\right\vert}
=\lim_{n\to \infty} \frac{\left\vert\eta_{n-1}-n-1\right\vert}{\left\vert\eta_n+n\right\vert}=1,
$$
for an $\alpha$-stable symbol, $\eta_n={\kappa}|n|^\alpha/2, \forall \alpha\in (0,2]$, except if $\alpha=1$ and $\kappa=2$, for which the limit vanishes. In that case, Eq. \eqref{Ean} gives $\mathbb E(a_n)=(-1)^{n-1}/{(n-1)}!$ for $n\geq 2$, thus $\mathbb E [f_0(z)]=ze^{-z}$ and $\mathbb E[f_0'(z)]=(1-z)e^{-z}$.
\end{proof}

\subsubsection{Expectations for the odd map $h_0(z)$}
The oddified  map $h_t(z):=z\sqrt{f_t(z^2)/z^2}$ obeys the Loewner equation
\begin{equation}\label{OLL} \dot{h_t}(z)=\frac{z}{2}h_t'(z)\frac{\lambda(t)+z^2}{\lambda(t)-z^2}.
\end{equation}
Its series expansion 
\begin{equation}e^{-t/2}h_t(z)=z+\sum_{n\geq 1}b_{2n+1}(t)z^{2n+1}
\end{equation}
gives the recursion:
$\dot{b}_{2n+1}=n b_{2n+1}+\sum_{k=0}^{n-1} \bar\lambda^{n-k} (2k+1)b_{2k+1},$ with $b_1=1$.
This is transformed into the set of equations
\begin{eqnarray}
\label{ub} w_n(t)&:=&b_{2n+1}(t)e^{-nt},\,\,\, w_0=1,\\ \nonumber
 \dot{w}_n(t)&=&\sum_{k=0}^{n-1}(2k+1) X_t^{n-k}w_k(t)\\ \label{indb} 
 &=&(2n-1)X_t w_{n-1}(t)+X_t \dot{w}_{n-1}(t).
\end{eqnarray}
The last equation is similar to Eq. \eqref{indter}, except for $2n-1$ here replacing $2n-2$ there, and an index  shifted boundary condition $w_0=1$ replacing $u_1=1$. Its solution can thus be written, as in \eqref{uop}, as the operator product 
\begin{eqnarray}\nonumber
w_n&=& \frac{1}{2} \mathcal J  \circ \big[\mathcal X + \big(n-\frac{1}{2}\big) \mathcal J\big]\circ \cdots \circ \big(\mathcal X + \frac{3}{2}\mathcal J \big)\mathds{1}\\ \label{wop}
&=&\frac{1}{2}\mathcal J \prod_{k=1}^{n-1} \circ \left[\mathcal X + \big(k+\frac{1}{2}\big)\mathcal J \right] \mathds{1},
\end{eqnarray}
 where $\mathds{1}(=w_0)$ is the constant function equal to $1$ on $\mathbb R^+$.
 
As mentioned in the introduction, the {\it conjugate odd} whole-plane L\'evy--Loewner evolution $e^{-(i/2)L_t} h_t\big(e^{(i/2)L_t}z\big)$ should have the {\it same law} as $h_0(z)$. At order $n$, we are thus interested in the stochastically rotated coefficients: $$e^{inL_t} b_{2n+1}(t)=(X_t)^{-n} w_n(t).$$ Comparing \eqref{wop} here to \eqref{uop} above, and adapting from the general formula \eqref{anXn}, we arrive directly at the final identity in law for the odd coefficients
\begin{eqnarray}\label{bnXn}
e^{inL_t}b_{2n+1}(t)&\stackrel{\rm (law)}{=}&-\int_0^{+\infty} ds {\tilde X}_s^{n}\prod_{k=1}^{n-1}\left(1-(2k+1)\int_0^{+\infty}ds_k \big(\tilde X^{[k]}_{s_k}\big)^{k} \right)\\ \nonumber
&\stackrel{\rm (law)}{=}& b_{2n+1}(0),
\end{eqnarray}
which, as it must, no longer depends of $t$. 
 Again, all factors in \eqref{bnXn} involve successive independent copies of the L\'evy process, whose expectations can be taken independently. 
Thus 
\begin{eqnarray}\nonumber
\mathbb E [b_{2n+1}(0)]&=&-\int_0^{+\infty} ds\, \mathbb E[{\tilde X}_s^{n}]\prod_{k=1}^{n-1}\left(1-2(k+1)\int_0^{+\infty}ds_k \big(\mathbb E\big [\big(\tilde X^{[k]}_{s_k}\big)^{k}\big]\right)\\\label{Ebn0}
&=&-\, \frac{1}{\eta_{n}+n}\prod_{k=1}^{n-1}\left(1-\frac{2k+1)}{\eta_{k}+k}\right).\end{eqnarray}
We finally obtain for the odd whole-plane L\'evy--Loewner evolution:
\begin{theo}\label{theoEbn}
For $n\geq 1$, setting $b_{2n+1}:=b_{2n+1}(0)$, 
\begin{eqnarray} \label{Ebnt}
b_{2n+1}(0)&\stackrel{\rm (law)}{=}&e^{inL_t}b_{2n+1}(t),\\ \label{Ebn}
\mathbb E(b_{2n+1})
&=&-\frac{\prod_{k=1}^{n-1}(\eta_{k}-k-1)}{\prod_{k=1}^{n}(\eta_k+k)}=\prod_{k=0}^{n-1}\frac{\eta_k-k-1}{\eta_{k+1}+k+1}.
\end{eqnarray}
\end{theo}
\begin{cor}\label{cor2}
The expected  conformal map $\mathbb E [h_0(z)]$ of the oddified whole-plane L\'evy--Loewner evolution, in the setting of Theorem \ref{theolle0}, is  polynomial of degree $2k+1$ if there exists a positive $k$ such that $\eta_k=k+1$, has radius of convergence $1$ for an $\alpha$-stable L\'evy process of symbol  $\eta_n={\kappa}n^\alpha/2$, $\alpha\in (0,2]$, except for the Cauchy process $\alpha=1,\kappa=2$, where $\mathbb E[h_0(z)]=ze^{-z^2/2}$.
\end{cor}
\subsubsection{Some results for L\'evy--Loewner maps}
In the general case (or in the specific case of $\textrm{SLE}_\kappa$), the formula \eqref{Ean} gives for the first terms:
\begin{eqnarray*}
\mathbb E(a_2)&=&-\displaystyle\frac{2}{\eta_1+1}=-\displaystyle\frac{4}{2+\kappa}\\
\mathbb E(a_3)&=&-\displaystyle\frac{\eta_1-3}{\eta_2+2}\displaystyle\frac{2}{\eta_1+1}=-\displaystyle\frac{\kappa-6}{(1+\kappa)(2+\kappa)}\\
\mathbb E(a_4)&=&-\displaystyle\frac{\eta_2-4}{\eta_3+3}\displaystyle\frac{\eta_1-3}{\eta_2+2}\displaystyle\frac{2}{\eta_1+1}=-\displaystyle\frac{4(\kappa-2)(\kappa-6)}{(6+9\kappa)(1+\kappa)(2+\kappa)}.
\end{eqnarray*}
For the oddified map, \eqref{Ebn} gives
\begin{eqnarray*}
\mathbb E(b_3)&=&-\frac{1}{\eta_1+1}=-\frac{2}{\kappa+2}\\
\mathbb E(b_5)&=&-\frac{\eta_1-2}{(\eta_1+1)(\eta_2+2)}=-\frac{\kappa-4}{2(\kappa+1)(\kappa+2)}\\
\mathbb E(b_7)&=&-\frac{(\eta_1-2)(\eta_2-3)}{(\eta_1+1)(\eta_2+2)(\eta_3+3)}=-\frac{(\kappa-4)(2\kappa-3)}{3(\kappa+1)(\kappa+2)(3\kappa+2)}.
\end{eqnarray*}\\

An interesting further identity, valid for $n\geq 1$, gives  the truncated series
\begin{eqnarray}\nonumber
S_n:=1+2\, \mathbb E(a_2)+\cdots+n\,\mathbb E(a_n)&=&-\frac{1}{2} \big[\eta_{n-1}-(n+1)\big]\,\mathbb E (a_n)\\ \label{sn}
&=&-\frac{1}{2} \big[\eta_{n}+n\big]\,\mathbb E (a_{n+1}).
\end{eqnarray}
Due to the peculiar factorized and recursive form of $\mathbb E(a_n)$ in Theorem \ref{theoEan} (respectively, of $\mathbb E(b_n)$ in Theorem \ref{theoEbn}),  we have seen in Corollary \ref{cor1} (respectively,  \ref{cor2}) that if there exits 
an integer $N$ such that $\eta_N=N+2$ (respectively, $\eta_N=N+1$), $\mathbb E[f_0(z)]$ is polynomial of degree $N+1$ (respectively, $\mathbb E[h_0(z)]$ is polynomial of degree $2N+1$). 

In the first case, $\mathbb E(a_{N+\ell})=0, \forall \ell\geq 2$, and  $\mathbb E[f_0'(z=1)]=S_{N+1}=0$, therefore the derivative $\mathbb E[f'_0(z)]$ necessarily contains the monomial $(1-z)$ as a factor. 

The first such case, $N=1$, gives a L\'evy symbol $\eta_1=3$. This includes in particular the SLE$_\kappa$ process  for $\kappa=6$ (recall then that   $\eta_n=\kappa n^2/2$), for which 
\begin{eqnarray}\nonumber
\mathbb E[f_0(z)]&=&z-z^2/2=\frac{1}{2}\big(1-(1-z)^2\big),\\ \label{fp1}
\mathbb E[f_0'(z)]&=&1-z.
\end{eqnarray}
The $N=2$ case gives $\eta_2=4$.  This includes in particular the SLE$_\kappa$ process for $\kappa=2$, for which 
\begin{eqnarray}\nonumber
\mathbb E[f_0(z)]&=&z-z^2+z^3/3=\frac{1}{3}\big(1-(1-z)^3\big),\\ \label{fp2}
\mathbb E[f_0'(z)]&=&(1-z)^2.
\end{eqnarray}
More generally, the SLE$_\kappa$ expected map, $z\mapsto\mathbb E [f_t(z)]$, is polynomial for the decreasing sequence of values  
$\kappa=\kappa_N:=\frac{2(N+2)}{N^2},\, N\geq 1.$\\

For the  oddified L\'evy--Loewner evolution, the $N=1$ first case gives $\eta_1=2$. This includes in particular the SLE$_\kappa$ for $\kappa=4$, for which 
\begin{eqnarray}\label{phip}
\mathbb E[h_0'(z)]=1-z^2.
\end{eqnarray}
The odd SLE$_\kappa$ expected map $\mathbb E [h_t(z)]$ is polynomial for $\kappa=\tilde \kappa_N:=\frac{2(N+1)}{N^2},\, N\geq 1.$
\subsection{Derivative moments}\label{subsecDeriv}
In this section,  motivated by the observations made in Sections \ref{an2} and \ref{computation}, completed in appendix \ref{Apphigh}, we prove the first part of Theorem \ref{theolle0}, which we recall here:\\

\textcolor{black}{\noindent {\bf Theorem 1.1.} \textit{Cases (i) (ii). Let $(f_t)_{t\geq 0}$ be the Loewner whole-plane process driven by the L\'evy process $L_t$ with L\'evy symbol $\eta$. We write
\begin{eqnarray}
\nonumber
f_t(z)&=&e^t\big(z+\sum_{n\geq 2}a_n(t) z^n\big),\\ \label{f0coeff}
f_0(z)&=&z+\sum_{n\geq 2}a_n(0) z^n,
\end{eqnarray}
The {\it conjugate} whole-plane L\'evy--Loewner evolution $e^{-iL_t} f_t\big(e^{iL_t}z\big)$ has the {\it same law} as $f_0(z)$, i.e., $e^{i(n-1)L_t}a_n(t)\stackrel{\rm (law)}{=}a_n(0)=:a_n$, and $\mathbb E(\I a_n(t)\I^2)=\mathbb E(\I a_n\I^2)$.  
Then:
\begin{enumerate}[(i)]
\item If $\eta_1=3$, we have 
$$\mathbb E(\I a_n\I^2)=1,\,\forall n\geq 1;$$
this case covers SLE$_6$.
\item If $\eta_1=1,\,\eta_2=4$, we have 
$$\mathbb E(\I a_n\I^2)=n,\,\forall n\geq 1;$$ this case covers SLE$_2 $.
\end{enumerate}
}}

\textcolor{black}{This will be proven in several steps, namely for SLE through Theorem \ref{main1} and its Corollary \ref{F6-2}, and for LLE through  Theorem \ref{levy} followed by Remark \ref{theo3.3}. These results will be a by-product of a thorough study of the derivative moments
$F(z):=\mathbb E [|f'_0(z)|^p]$, $p\in \mathbb R$,  of the inner whole-plane SLE or LLE maps. Using \eqref{f0coeff} for $p=2$, one gets the derivative's quadratic moment 
\begin{equation}\label{expansion}
\mathbb E [|f'_0(z)|^2]=\sum_{n,m=1}^{\infty}\mathbb E[a_n\overline{a_m}] nmz^{m-1}\bar z^{m-1},\,\,\,z\in \mathbb D,
\end{equation}
so that its integral means, 
\begin{equation}\label{expansionbis}
\frac{1}{2\pi}\int_0^{2\pi}\mathbb E [|f'_0(|z|e^{i\theta})|^2]=1+\sum_{n\geq 2}\mathbb E(|a_n|^2) n^2 (z\bar z)^{n-1},
\end{equation}
is a generating function for the coefficients' quadratic moments. Our study uses a  partial differential equation satisfied by $F(z)$ above, which is an extension of that derived by Beliaev and Smirnov (BS) in their study of the harmonic measure for SLE \cite{BS}. We follow it by studying  the space of its analytic solutions in the unit disk, among which some special factorized solutions exist. We then develop the same formalism, \textit{i.e.}, the martingale derivation of a BS-like equation for the derivative moments and the construction of special explicit solutions, for the oddified whole-plane SLE and LLE processes \eqref{defoddified}, and for the higher $m$-fold transforms \eqref{defmfold}.}
  \subsubsection{The Beliaev--Smirnov equation} \label{BSderiv} In Ref. \cite{BS}, Beliaev and Smirnov first consider a standard radial {\it (outer)} SLE process ($g_t(z), t\geq 0$), from $\mathbb C\setminus K_t$ to the complement $\mathbb D_{-}$ of the unit disk, where, as usual, $K_t$ denotes the SLE hull at time $t$. This SLE process satisfies a standard ODE, which can be continued to negative times (via a two-sided Brownian motion $B_t$ in the Schramm--Loewner source term $\lambda(t)=e^{i\sqrt{\kappa}B_t}$). The harmonic spectrum is best studied via the inverse map, $g_t^{-1}$, which satisfies a Loewner-type PDE (as the whole-plane evolution considered in this article).  \textcolor{black}{Since the processes $g_t^{-1}$ and $g_{-t}$  have the same law (up to conjugation by $e^{i\sqrt{\kappa}B_t}$), BS \textit{redefine in Ref. \cite{BS} a radial SLE (denoted there by $f_t$) as 
  \begin{equation}
  \tilde f_t(z):=g_{-t}(z)\stackrel{\rm (law)}{=}e^{-i\sqrt{\kappa}B_t}g_t^{-1}(e^{i\sqrt{\kappa}B_t}z),\,t\in \mathbb R,
  \end{equation} thus mapping  $\mathbb D_{-}$ to $\mathbb C\setminus K_{-t}$}.} Then they show that the expectation  
 \begin{equation}\label{tildeF}\tilde{F}(z,t):=\mathbb E(\I \tilde f_t'(z)\I^p),\end{equation}
where $p$ is real, is solution to the differential equation
\begin{eqnarray}\nonumber&&p\frac{r^4+4r^2(1-r\cos \theta)-1}{(r^2-2r\cos \theta+1)^2}\tilde{F}+
\frac{r(r^2-1)}{r^2-2r\cos \theta+1}\tilde{F}_r\\ \label{BSprime}&&-\frac{2r\sin \theta}{r^2-2r\cos \theta+1}\tilde{F}_\theta+\Lambda \tilde{F}-\tilde{F}_t
=0,\end{eqnarray}
with $z=re^{i\theta}$, and where subscripts  represent partial derivatives of $\tilde F$ with respect to $r$, $\theta$ and $t$, and where $\Lambda$ stands for the infinitesimal generator of the SLE driving Brownian process, i.e., $\Lambda=(\kappa/2)\partial^2/\partial\theta^2$. 

To derive \eqref{BSprime}, they consider the {\it martingale} $\mathcal{M}_s:=\mathbb E(\I \tilde f_t(z)'\I^p\mid\mathcal{F}_s)$, where $\mathcal{F}_s$ denotes the $\sigma-$algebra generated by  $\{B_\tau\,;\tau\leq s\}$. From the SLE {\it Markov property}, they show (Lemma 2 in \cite{BS}) that 
$$\mathbb E(\I \tilde f_t'(z)\I^p\mid \mathcal{F}_s)=\I \tilde f_s'(z)\I^p\tilde{F}(z_s,t-s),$$
in terms of the conjugate variable $z_s:=\tilde f_s(z)e^{-i\sqrt{\kappa}B_s}.$   Expressing the fact that the $ds$ drift term vanishes in the  It\^o derivative of the right-hand side then gives equation \eqref{BSprime} above.

The next step in their derivation is to remark that, by {\it stationarity},  the limit  of $e^{-t}\tilde f_t(z)$ as $t\to+\infty$ exists, and has the same law as the value  $\hat f_0(z)$ at proper time zero of the {\it (outer) whole-plane} SLE (denoted by $F_0(z)$ in \cite{BS}). Rewrite  $\tilde F(z,t)$ \eqref{tildeF} above trivially as 
\begin{equation} \label{Fzt}\tilde{F}(z,t)=e^{\sigma pt}\mathbb E(\I e^{-\sigma t}f_t'(z)\I^p),\,\, (\sigma=1),\end{equation} 
to obtain
\begin{equation}\label{Fdef}F(z):=\mathbb E [ |\hat f_0'(z)|^p]=\lim_{t\to \infty} e^{-\sigma pt}\tilde F(z,t).\end{equation}
 Note that this exterior whole-plane map $\hat f_0(z)$, acting on the exterior $\mathbb D_{-}$ of the unit disk, has precisely the same law as the conjugate via $z\mapsto 1/z$ of the  interior whole-plane SLE map $f_0(z)$ that we consider in this article.  
Substituting \eqref{Fzt} into \eqref{BSprime} and taking the large $t$ limit, BS thereby obtain the following equation for $F(z)$
\begin{eqnarray}\nonumber&&p\left(\frac{r^4+4r^2(1-r\cos \theta)-1}{(r^2-2r\cos \theta+1)^2}-\sigma\right)F
+\frac{r(r^2-1)}{r^2-2r\cos \theta+1}F_r\\ \label{BS}&-&\frac{2r\sin \theta}{r^2-2r\cos \theta+1}F_\theta+\Lambda F=0.
\end{eqnarray}
 For our {\it interior} case, we  similarly introduce the function $\tilde f_t, t\geq 0$, as the continuation $g_{-t}$ of the standard {\it inner} radial SLE process $g_t$ to negative times, which has the same law as its inverse map $g_t^{-1}$.  Then the limit $e^t \tilde f_t(z)$ as $t\to +\infty$ exists, and has  the same law as the {\it inner whole-plane} process $f_0(z)$ considered in this work.  This amounts to formally taking $\sigma=-1$ in \eqref{Fzt}, in effect changing the sign of the term $-\sigma pF$ in \eqref{BS} that results from the time-derivative term in \eqref{BSprime}. This simple observation results in:    \begin{prop} \label{pBS1} 
   For the interior whole-plane Schramm--Loewner evolution, as considered here in the setting of Theorem \ref{theolle0}, the expected  moments of the derivative modulus, $F(z)=\mathbb E(|f_0'(z)|^p)$, satisfy the Beliaev--Smirnov equation \eqref{BS} with $\sigma=-1,$ and $\Lambda=(\kappa/2)\partial^2/\partial\theta^2$ the generator of the driving Brownian process.
\end{prop}
Finally, note that the BS derivation for SLE, as recalled above, is also valid for the L\'evy--Loewner evolution, which possesses the same Markov property,  together with the existence of similar whole-plane stationary limits.  Stochastic calculus (and It\^o formula) can be generalized to L\'evy processes \cite{applebaum}, resulting in the same martingale argument. As mentioned by Beliaev and Smirnov in \cite{BS}, one simply has to take  for $\Lambda$ in \eqref{BS} the generator of the driving L\'evy process. 
 We therefore state:  
   \begin{prop} \label{pBS2} For the interior whole-plane L\'evy--Loewner evolution, as considered here in the setting of Theorem \ref{theolle0}, the expected moments of the derivative modulus, $F(z)=\mathbb E(|f_0'(z)|^p)$, satisfy the Beliaev--Smirnov equation \eqref{BS} with $\sigma=-1,$ and $\Lambda$ the generator of the driving L\'evy process.
\end{prop} 


\subsubsection{Whole-plane SLE solutions} \label{wpsle}We first study the SLE$_\kappa$ case. 
Let us switch to $z,\bar z$ variables, instead of polar coordinates, and write $F(z)$ above as 
where \textcolor{black}{\begin{eqnarray} \label{FZZ} F(z)&=&\mathbb E(|f_0'(z)|^p)=F(z,\bar z)\\ \label{z1z2}
F(z_1,z_2)&:=&\mathbb E[(f_0'(z_1))^{p/2}(\bar f_0'(z_2))^{p/2}]\\ \label{zbarz} \bar f_0'(z)&:=&\overline{f_0'(\bar z)}. \end{eqnarray}
Note that the function $F(z_1,z_2)$ is holomorphic in the bi-disk $\mathbb D\times \mathbb D$, or in its inverse $\mathbb D_-\times \mathbb D_-$ for the exterior case (expectation and derivation can be interchanged). This allows one to consider hereafter the variables $z$ and $\bar z$ as formally independent in $F(z,\bar z)=\mathbb E[(f_0'(z))^{p/2}(\bar f_0'(\bar z))^{p/2}]$.} 
Using  $\partial:=\partial_z$, $\overline{\partial}:=\partial_{\bar z}$, Eq. \eqref{BS} then becomes 
\textcolor{black}{\begin{eqnarray}\mathcal P(D)[F(z,\bar z)]&=&0,\label{zz1}\\ \nonumber 	\mathcal P(D)&:=&-\frac{\kappa}{2}(z\partial-\bar z\overline{\partial})^2 +\frac{z+1}{z-1}z\partial +\frac{\bar z+1}{\bar z-1}\bar z\overline{\partial} \\  \label{PD} &&-p\left[\frac{1}{(z-1)^2}+\frac{1}{(\bar z-1)^2}+\sigma- 1\right].\end{eqnarray}}
To study this equation for the interior case $z\in \mathbb D$ and $\sigma=-1$, we shall need the three lemmas below.
\begin{lemma} \label{lem0}The space of formal series \textcolor{black}{$F(z,\bar z)=\sum_{k,\ell\in\mathbb N}a_{k,\ell}z^k\bar z^{\ell}$ in non-negative integer powers of $z,\bar{z}$ with complex coefficients} that are solutions of ($\ref{zz1}$) is one-dimensional.
\end{lemma}
\demo The Lemma is an easy consequence of the two following observations:\\
First,  the differential operator, $\mathcal P(D)$, involved in \eqref{zz1},  is polynomial in $z\partial$ and $\bar{z}\overline{\partial}$, and the monomials $z^k\bar z^{\ell}$ are eigenvectors of the latter two operators.\\
Second, the \textcolor{black}{non-differential term in $\mathcal P(D)$}, may be written as $A(z)+A(\bar{z})$, with $A(0)=0$.\\
Now, assuming that $G$ is a solution of ($\ref{zz1}$) with $G(0,0)=0$, it suffices to prove that, necessarily, $G= 0$. We argue by contradiction: If not,  consider the minimal (necessarily non constant) term $a_{k,\ell} z^k{\bar z}^{\ell}$ in the series, with $a_{k,\ell}\neq 0$ and $k+\ell$ minimal (and non vanishing). Then $\mathcal P(D)[F]$ will have a minimal non-vanishing term of the form
$$-a_{k,\ell}\big[\frac{\kappa}{2}(k-\ell)^2+k+\ell\big]z^k\bar{z}^{\ell},$$
contradicting the fact that $\mathcal P(D)[F]=0$.\\
\begin{lemma}\label{lem1}
The quantity 
$F(z,\bar z=0)
$
satisfies the boundary equation obtained by setting $\bar z=0$ in \eqref{zz1} (here $\sigma=-1$):
\begin{eqnarray} \mathcal P(\partial)[F(z,0)]:=\left\{-\frac{\kappa}{2}(z\partial)^2 +\frac{z+1}{z-1}z\partial -p\left[\frac{1}{(z-1)^2}- 1\right]\right\}F(z,0)=0.
\label{zz0}
\end{eqnarray}
The complex conjugate equation also holds for $F(z=0,\bar z).$
\end{lemma}
\begin{demo} The bi-analytic function $F(z,\bar z)$ 
 has a double series expansion of the form $F(z,\bar z)=\sum_{k\geq 0,\ell\geq 0} a_{k,l}z^k{\bar z}^{\ell}.$ When acting on it with a differential operator $\mathcal P (D)$ as in \eqref{zz1}, the resulting double series must vanish identically, hence all its coefficients $a_{k,\ell}$ must as well. This implies that the variables $z$ and $\bar z$ can be considered as two {\it independent} complex variables in \eqref{zz1}. By symmetry, the complex conjugate equation also holds for $F(z=0,\bar z)$. 
\end{demo}
\begin{lemma} \label{lem2}The action of the operator $\mathcal P(D)$ $(\ref{PD})$ on a function of the factorized form $F(z,\bar{z})=\vphi(z)\overline{\vphi}(\bar{z}) P(z,\bar{z})$, \textcolor{black}{with the definition $\bar\varphi(\bar z):=\overline{\varphi(z)}$,} is by Leibniz's rule given by
\begin{eqnarray}\nonumber\mathcal P(D)[\varphi\bar\varphi P]=&-&\frac{\kappa}{2}\vphi\overline{\vphi}(z\partial-\bar z\bar{\partial})^2 P-\kappa(z\partial-\bar z\overline{\partial})(\vphi\overline{\vphi})(z\partial-\bar z\overline{\partial})P\\ \nonumber
&+&\kappa (z\partial\vphi)(\bar z\overline{\partial}\overline{\vphi})P
+\vphi\overline{\vphi}\frac{z+1}{z-1}z\partial P+\vphi\overline{\vphi}\frac{\bar z+1}{\bar z-1}\bar z\overline{\partial}P\\ \nonumber
&+&\left[-\frac{\kappa}{2}\overline{\vphi}(z\partial)^2\vphi
-\frac{\kappa}{2}\vphi (\bar z\overline{\partial})^2 \overline{\vphi}
+\overline{\vphi}\frac{z+1}{z-1}z\partial\vphi
+\vphi\frac{\bar z+1}{\bar z-1}\bar z\overline{\partial}\,\overline{\vphi}\right]P\\&-&p\left[\frac{1}{(z-1)^2}+\frac{1}{(\bar z-1)^2}+\sigma-1\right]\vphi\overline{\vphi}P.\label{zzb}
\end{eqnarray}
\end{lemma}
\begin{cor} \label{Prad} For the the particular choice: $P(z,\bar z):=P(z\bar z)$, the first line \textcolor{black}{of the r.h.s. of Eq.} \eqref{zzb}  vanishes identically.
\end{cor}
\begin{demo}
$P$ is then {\it radial}, and the differential operator in the first line acts on the angular variable only.  
\end{demo}
\begin{cor} \label{pdelta}In the interior case ($\sigma=-1$), and for the particular choice:
\begin{eqnarray}
\varphi(z)=F(z,\bar z=0),\,\,\,\bar{\varphi}(\bar z)=F(z=0,\bar z),\end{eqnarray}
the last two lines \textcolor{black}{of the r.h.s. of Eq.} \eqref{zzb} in Lemma \ref{lem2} 
 vanish identically.
\end{cor}
\begin{demo} Use Lemma \ref{lem1}. The last two lines of \eqref{zzb} are precisely of the form
$$P\times \left[\bar{\varphi}\mathcal P(\partial)\varphi   +\varphi\,\mathcal P(\overline{\partial}) \bar{\varphi}\right] =0.
$$
 \end{demo} 
\begin{cor}\label{corfond} The function  $$F(z,\bar z)=\mathbb E\big[|f'_0(z)|^p\big]=\mathbb E\big[(f'_0(z))^{p/2}(\overline{f'_0}(\bar z))^{p/2}\big]$$
  is  the unique solution of (\ref{zz1}) such that $F(0,0)=1$. 
   The function  $$\varphi(z)=F(z,0)=\mathbb E\big[(f'_0(z))^{p/2}\big]$$ is the unique solution of  \eqref{zz0} such that $F(z=0,0)=1$. \textcolor{black}{Corollary \ref{pdelta} applies to it.}
  \end{cor}

In the particular case $p=2$, and for SLE$_\kappa$, with $\kappa=6$ or $2$, we have obtained above the derivative expectations \eqref{fp1} and \eqref{fp2}:
\begin{eqnarray}\label{kp6-2}
\varphi(z)=\mathbb E\big[f'_0(z)\big]=(1-z)^\alpha,\,\,\, \alpha=1, \kappa=6;\,\,\,\alpha=2,\kappa=2.
\end{eqnarray}
From Corollary \ref{corfond}, we know that they are annihilated by the boundary operator $\mathcal P(\partial)$ of Lemma \ref{lem1} (with a similar result for  the conjugate quantities), and that  the  two last lines of \eqref{zzb}, equal to $\bar\varphi\mathcal P(\partial)[\varphi]  P+\varphi \mathcal P(\bar\partial)[\bar \varphi]  P$, identically vanish.  

Denote then by $\mathcal P_{\textrm{sing}}(D)$ the {\it singular} operator made of the second line of \eqref{zzb}, which contains the  pole at $z=1,\bar z=1$: 
$$\mathcal P_{\textrm{sing}}(D)[\varphi\bar\varphi P]:=\kappa (z\partial\vphi)(\bar z\overline{\partial}\overline{\vphi})P
+\vphi\overline{\vphi}\frac{z+1}{z-1}z\partial P+\vphi\overline{\vphi}\frac{\bar z+1}{\bar z-1}\bar z\overline{\partial}P.$$
For $\varphi_\alpha(z):=(1-z)^\alpha$, its action gives the factorized form
\begin{eqnarray}\label{Psing}
\mathcal P_{\textrm{sing}}(D)[\varphi_\alpha\bar\varphi_\alpha P]=\varphi_\alpha\bar\varphi_\alpha\left[\frac{\kappa\alpha^2z\bar z}{(1-z)(1-\bar z)}+\frac{z+1}{z-1}z\partial +\frac{\bar z+1}{\bar z-1}\bar z\overline{\partial}\right]P.\end{eqnarray}

Thanks to Lemma \ref{Prad}, it is now natural to look for {\it radial} solutions, $P(z,\bar z)=P(z\bar z)$, that make \eqref{Psing} vanish.
 The resulting equation is simply
 \begin{equation}\label{PP'}
 \frac{P'(z\bar z)}{P(z\bar z)}=\frac{\kappa \alpha^2}{2}\frac{1}{1-z\bar z},
 \end{equation}
 which is immediately solved, for $P(0)=1$, into
 \begin{equation}\label{Pzzbar}
 P(z\bar z)=(1-z\bar z)^{-\kappa \alpha^2/2}.
 \end{equation}
From Corollaries  \ref{Prad} and \ref{pdelta}, we obtain that 
\begin{equation}\label{Fzz}
F(z,\bar z)=\frac{(1-z)^\alpha(1-\bar z)^\alpha}{(1-z\bar z)^{\beta}},\,\,\,\,\beta=\frac{\kappa}{2} \alpha^2,\end{equation}
is the unique solution to the differential equation \eqref{zz1}, such  that $F(0,0)=0$, {\it if and only if ${\varphi}_\alpha(z)=(1-z)^\alpha$ is a solution to the boundary equation} \eqref{zz0} of Lemma \ref{lem1}. 
 From \eqref{kp6-2}, we already know this to hold true for $p=2$, in the two cases $\alpha=1,\kappa=6$, or $\alpha=2,\kappa=2$. 
 
 In the general case, we obtain: 
\begin{eqnarray}\label{Pphialpha}
\mathcal P(\partial)[\vphi_\alpha]&=&A(\alpha) \vphi_\alpha+B(\alpha)\vphi_{\alpha-1}+ C(\alpha) \vphi_{\alpha-2}\\
\label{Aalpha}A(\alpha)&:=&- \frac{\kappa}{2}\alpha^2+\alpha+p,\\ \label{Balpha2}
B(\alpha)&:=&\frac{\kappa}{2}\alpha(2\alpha-1)-3\alpha,\\ \label{Calpha}
C(\alpha)&:=& -\frac{\kappa}{2}\alpha(\alpha-1)+2\alpha-p.  
\end{eqnarray}
Notice that $A+B+C=0$.  The boundary equation \eqref{zz0} $\mathcal P(\partial)[\vphi_\alpha]=0$ thus reduces to the set of equations:
$A(\alpha)=0,B(\alpha)=0$, which is solved into
\begin{eqnarray}\label{kpgen}
\alpha=\alpha(\kappa):=\frac{6+\kappa}{2\kappa},\,\,p=p(\kappa):=\frac{(6+\kappa)(2+\kappa)}{8\kappa}.
\end{eqnarray}
This set of values naturally  includes the above cases \eqref{kp6-2} for $\kappa=6$ and  $\kappa=2$. 
From Corollary \ref{corfond}, we therefore obtain the general result:
\begin{theo} \label{main1}The whole-plane SLE$_{\kappa}$ map $f_0(z)$ has derivative moments  
\begin{eqnarray*}\mathbb E\big[(f'_0(z))^{p/2}\big]&=&(1-z)^\alpha,\\
\mathbb E\big[|f'_0(z)|^p\big]&=&\frac{(1-z)^\alpha(1-\bar z)^\alpha}{(1-z\bar z)^\beta},
\end{eqnarray*}
for the special set of exponents $p=\kappa\alpha(\alpha+1)/6={(6+\kappa)(2+\kappa)}/{8\kappa}$, with $\alpha=(6+\kappa)/2\kappa$ and 
$\beta=\kappa \alpha^2/2=(6+\kappa)^2/8\kappa$.
\end{theo}
\begin{cor}\label{F6-2}
The whole-plane SLE$_{\kappa}$ map $f_0(z)$ has first and second derivative moments, for $\kappa=6$:
\begin{eqnarray*}\mathbb E(f'_0(z))=1-z,\,\,\,\,
\mathbb E(|f'_0(z)|^2)=\frac{(1-z)(1-\bar z)}{(1-z\bar z)^3};
\end{eqnarray*}
for $\kappa=2$: 
\begin{eqnarray*}\mathbb E(f'_0(z))=(1-z)^2,\,\,\,\,
\mathbb E(|f'_0(z)|^2)=\frac{(1-z)^2(1-\bar z)^2}{(1-z\bar z)^4}.
\end{eqnarray*}
\end{cor}
In the setting of Theorem \ref{theolle},  the coefficient $n^2\mathbb E (|a_n|^2)$ is that of the term of order $(z\bar z)^{n-1}$ in the expansion \eqref{expansion} of $\mathbb E(|f'_0(z)|^2)$. It can be obtained directly from the explicit expressions in Corollary \ref{F6-2}, as $\mathbb E (|a_n|^2)=1$ for $\kappa=6$, and $\mathbb E (|a_n|^2)=n$ for $\kappa=2$. Equivalently, we can evaluate the respective integral means \eqref{expansionbis}
$$\frac{1}{2\pi}\int_{\partial\mathbb D}\mathbb E (\I f_0'(zu)\I^2)\, |d u|=\frac{1+z\bar z}{(1-z\bar z)^3};\,\,\,\frac{1+4z\bar z+(z\bar z)^2}{(1-z\bar z)^4},$$
which establishes the equivalent Corollary \ref{corint}. This achieves for SLE$_\kappa$ the proof of cases {\it (i)} and {\it (ii)}  of Theorem \ref{theolle}. 
\textcolor{black}{As mentioned earlier, the expressions for the second moments in Corollary \ref{F6-2} appeared in Ref. \cite{IL},  as computer-assisted solutions to a  double recursion; 
the set  \eqref{kpgen} was also mentioned there, and was further studied in Refs. \cite{2012arXiv1203.2756L,2013JSMTE..04..007L}.}
\subsubsection{L\'evy--Loewner evolution}\label{llesle}
\begin{theo}\label{levy}
If a L\'evy process has its first $m$ symbols ($m\geq 1$) given by $\eta_j=\kappa j^2/2,  1\leq j\leq m,$ with 
$\kappa=6/(2m-1)$, then the associated L\'evy--Loewner map $f_0(z)$ has the same derivative moments of order $p$ as SLE$_\kappa$, for the particular value of the exponent $p=m(m+1)/(2m-1)$, as given in Theorem \ref{main1} with $\alpha=m, \beta=3m^2/(2m-1).$   
\end{theo}
\begin{proof}
For the  whole-plane {L\'evy--Loewner} evolution, the Beliaev--Smirnov  equation \eqref{zz1} becomes \cite{BS}
\begin{equation}\label{LLee}\Lambda F+\frac{z+1}{z-1}z\partial F+\frac{\bar z+1}{\bar z-1}\bar z\overline{\partial} F-p\left[\frac{1}{(z-1)^2}+\frac{1}{(\bar z-1)^2}- 2\right]F=0.\end{equation}
The action of the L\'evy infinitesimal generator $\Lambda$ on a term 
$z^k {\bar z}^\ell$ is $$\Lambda\, [z^k {\bar z}^\ell]=-\eta(k-\ell)z^k {\bar z}^\ell,$$ where $\eta(\cdot)$ here real and even. It is such that $\eta(0)=0$, therefore for any $n\in \mathbb Z$,$$\Lambda\, [(z\bar z)^n]=0.$$
For the set of solutions $F(z,\bar z)$ \eqref{Fzz} of \eqref{zz1}, as  given in Theorem \ref{main1}, we thus have 
\begin{equation} \label{lambdaalpha}\Lambda F(z,\bar z)=\frac{1}{(1-z\bar z)^\beta}\Lambda [(1-z)^\alpha(1-\bar z)^{\alpha}].\end{equation}
If the exponent $\alpha={(6+\kappa)}/{2\kappa}$ equals an {\it integer} $m\geq 1$,  $\vphi_\alpha=(1-z)^\alpha$ is polynomial of order $m$, and $\Lambda [\vphi_\alpha{\bar \vphi}_\alpha]$ contains only the finite set of L\'evy symbols $\{\eta_1,\ldots,\eta_m\}$. If this set co\"{i}ncides with the set of values $(\kappa/2) \ell^2$ for $\ell=\{1,\cdots,m\}$, the action of the L\'evy generator $\Lambda$ on $\vphi_\alpha\bar{\vphi}_\alpha$ in \eqref{lambdaalpha} co\"{i}ncides with that of the Brownian generator $-\frac{\kappa}{2}(z\partial-\bar z\overline{\partial})^2$.  In this case, $F(z,\bar z)$ \eqref{Fzz}, solution of the SLE$_\kappa$ equation \eqref{zz1}, is also a solution of the L\'evy--Loewner equation \eqref{LLee}, and Theorem \ref{main1} is also valid for a L\'evy--Loewner evolution with such symbols.\end{proof}
\begin{rkk}\label{theo3.3}
The cases $m=1$ and $m=2$ give respectively the condition $\eta_1=3$ with an equivalent SLE$_\kappa$ parameter  $\kappa=6$, with $ p=2, \alpha=1, \beta=3$, and the conditions $\eta_1=1, \eta_2=4$, corresponding to $\kappa=2, p=2, \alpha=2, \beta=4$, as in Corollary  \ref{F6-2}. Cases {\it (i)} and {\it (ii)} of Theorem \ref{theolle} thus follow.
\end{rkk}
\textcolor{black}{\begin{rkk}
 Theorem \ref{levy} is a generalization of Theorem \ref{main1}:  there  exist L\'evy processes satisfying the hypotheses of the Theorem which are not  Brownian motions. A simple example is as follows:  take the sum $L_t=\sqrt{\kappa} B_{t}+M_t$, where $B_t$ is standard Brownian motion,   and $M_t$ is an independent compound Poisson process with L\'evy symbol
$$\eta(\xi)=2\lambda \sin^2\pi \xi.$$
In other words, 
$$M_t=2\pi(Z_1+\cdots + Z_{N_t}),$$
where $N_t$ is a Poisson process of intensity $\lambda>0$ and $\{Z_j\}$ a collection of independent Bernoulli variables taking values $\pm 1$ with probability $1/2$. By additivity, the L\'evy symbol of $L_t$ then coincides with that of $\sqrt{\kappa}B_t$ for every integer.
\end{rkk}}

 \subsection{Odd whole-plane SLE}\label{subsecodd} 
In this section, we study the oddified whole-plane SLE$_\kappa$ map, $h_t(z)=z\sqrt{f_t(z^2)/z^2}$, and derive the analogue of the Beliaev--Smirnov equation for its derivative moments, $\mathbb E[|h_0'(z)|^p]$, before proceeding along lines similar to those in Section \ref{wpsle}, in order to find special solutions to that equation. \textcolor{black}{This will lead us to the proof of the second part of Theorem \ref{theolle}, which we recall here:}\\

\textcolor{black}{\noindent {\bf Theorem 1.2.} \textit{Case (iii). Let $(f_t)_{t\geq 0}$ be the Loewner whole-plane process driven by the L\'evy process $L_t$ with L\'evy symbol $\eta$. We write
 for the oddification of $f_t$
$$h_t(z)=z\sqrt{f_t(z^2)/z^2}=e^{t/2}\big(z+\sum_{n\geq 1}b_{2n+1}(t)z^{2n+1}\big).$$
The 
  {\it conjugate} oddified whole-plane L\'evy--Loewner evolution $e^{-(i/2)L_t} h_t\big(e^{(i/2)L_t}z\big)$ has the {\it same law} as $h_0(z)$, i.e., $e^{inL_t}b_n(t)\stackrel{\rm (law)}{=}b_n(0)=:b_n$, and $\mathbb E(|b_{2n+1}(t)|^2)=\mathbb E(|b_{2n+1}|^2)$.\\
 Then,  if $\eta_1=2$, we have 
$$\mathbb E(\I b_{2n+1}\I^2)=\frac{1}{2n+1},\,\forall n\geq 1;$$ this case covers the oddified SLE$_4$.}} 

 \textcolor{black}{For $p=2$, one has the derivative's quadratic  moment
\begin{equation}\label{expansionodd}
\mathbb E [|h'_0(z)|^2]=\sum_{n,m=0}^{\infty}\mathbb E[b_{2n+1}\overline{b}_{2m+1}] (2n+1)(2m+1)z^{2n}\bar z^{2m},\,\,\,z\in \mathbb D,
\end{equation}
so that its integral means, 
\begin{equation}\label{expansionoddbis}
\frac{1}{2\pi}\int_0^{2\pi}\mathbb E [|h'_0(|z|e^{i\theta})|^2]=1+\sum_{n\geq 1}\mathbb E(|b_{2n+1}|^2) (2n+1)^2 (z\bar z)^{2n},
\end{equation}
is a generating function for the  coefficients' quadratic moments. The proof of case \textit{(iii)} of Theorem \ref{theolle} will be obtained in Corollary \ref{F06-2} of Theorem \ref{main3} in the SLE$_{\kappa=4}$  case, and in Proposition \ref{oddeta1} in the LLE case.}
\subsubsection{Martingale argument} \label{martarg}The Loewner equation for $h_t(z)$ is easily derived from the one \eqref{loewner} governing $f_t(z)$, with a driving function $\lambda(t)$, as 
\begin{equation}\label{loewnerodd} \partial_t{h}_t(z)=\frac{z}{2}h_t'(z)\frac{\lambda(t)+z^2}{\lambda(t)-z^2}.\end{equation}
To avoid cumbersome factors of $2$, it is convenient in this section to 
 work with $(t,z)\mapsto h_{2t}(z)$, which is a normalized whole-plane  Loewner process with 
\begin{equation}\label{loewneroddbis}\partial_t{h}_{2t}(z)=zh_{2t}'(z)\frac{\lambda(2t)+z^2}{\lambda(2t)-z^2}.\end{equation}
Note that for this oddified whole-plane process, the underlying probability measure is no longer a single Dirac measure, but the barycenter of two Dirac masses at two diametrically opposite points, $\lambda(2t)$ and $-\lambda(2t)$.

In the case where $(f_t)$ is the SLE$_\kappa$ Loewner chain, we can write,  instead of $\lambda(2t)=e^{i\sqrt{\kappa}B_{2t}}$, $\xi_t:=e^{i\sqrt{2\kappa}B_t}$ which has the same law. We then follow the same method as in $\cite{BS}$, as recalled above, to find an equation satisfied by 
\begin{equation}\label{Fh0}F(z):=\mathbb E(\I h_0'(z)\I^p).\end{equation} 
To this aim, we consider the odd whole-plane map at time $0$, $h_0(z)$, as a particular large-time limit, $\lim_{t\to\infty}e^t\tilde{f}_t(z)$, where $\tilde{f}_t$ is now the (inner) radial Loewner process satisfying 
\begin{equation}\label{slezetat}\partial_t \tilde{f}_t(z)=z\tilde{f}_t'(z)\frac{z^2+\xi_t}{z^2-\xi_t}.\end{equation}
  It is easy to see that the Markov Lemma 2 in $\cite{BS}$ goes through for $\tilde f$ in this new setting; we can then argue as in Lemma 4 therein, namely by using a similar martingale argument. More precisely, we consider the martingale with respect to the Brownian filtration $\mathcal F_s, s\leq t$, together with the traduction of the Markov property,
\begin{equation}
\label{Nmart}\mathcal{N}_s:=\mathbb{E}(\I \tilde f_t'(z)\I^p \mid \mathcal{F}_s)\,=\I \tilde{f}_s'(z)\I^p \tilde{F}(z_s,t-s),\end{equation}
where
\begin{equation}\label{zs}z_s:=\tilde{f}_s(z)/\sqrt{\xi_s},\end{equation}
and
\begin{equation}\label{tildeh} \tilde{F}(z,t):=\mathbb{E}(\I \tilde f_t'(z)\I^p).\end{equation}
 Following step by step the argument therein, we write in our new setting
$$z_s=re^{i\theta},\;d\log{z_s}=d\log{r}+id\theta =d\log\tilde{f}_s-i\sqrt{\frac{\kappa}{2}}dB_s,$$
where
$$d\log \tilde{f}_s=\frac{d\tilde{f}_s}{\tilde{f}_s}=\frac{z_s^2+1}{z_s^2-1}ds.$$
Using here a somehow redundant notation in terms of $r,\,\theta$ and $\rho:=r^2$, $\alpha:=2\theta$, we get
\begin{eqnarray*}d\log{r}+id\theta&=&\frac{z_s^2+1}{z_s^2-1}ds-i\sqrt{\frac{\kappa}{2}}dB_s,\\
\partial_s\log{\I \tilde{f}_s'(z)\I}&=&\frac{\rho^4+8\rho^2-6\rho^3\cos{\alpha}-2\rho\cos{\alpha}-1}{(\rho^2-2\rho\cos{\alpha}+1)^2},\\
d\theta&=&-\frac{2\rho\sin{\alpha}}{(\rho^2-2\rho\cos{\alpha}+1)}ds-\sqrt{\kappa/2}dB_s,\\
\frac{dr}{r}&=&\frac{(\rho^2-1)}{(\rho^2-2\rho\cos{\alpha}+1)}ds.\end{eqnarray*}
Writing $\tilde F(z,t)$, as defined in \eqref{tildeh}, as $\tilde F(r,\theta,t)$, the vanishing of the $ds$ term in the It\^o derivative of $\mathcal N_s$ gives a PDE in $(r,\theta,t)$ satisfied by $\tilde{F}$, similar to Eq. \eqref{BSprime}. To finish, a large $t$ limit argument, entirely similar to \eqref{Fzt}-\eqref{Fdef} in Section \ref{BSderiv}, leads to the following PDE satisfied by  $F(z)$ \eqref{Fh0}, still using at this moment the mixed notation in $r$, $\theta$, and  $\rho:=r^2$, $\alpha:=2\theta$:
\begin{eqnarray}\nonumber&&p\left(\frac{\rho^4+8\rho^2-6\rho^3\cos{\alpha}-2\rho\cos{\alpha}-1}{(\rho^2-2\rho\cos{\alpha}+1)^2}+1\right)F+\frac{(\rho^2-1)}{\rho^2-2\rho\cos{\alpha}+1}rF_r\\ \label{oddBS}&&-\frac{2\rho\sin{\alpha}}{\rho^2-2\rho\cos{\alpha}+1}F_\theta+\frac{\kappa}{4} F_{\theta\theta}=0.\end{eqnarray}
Retaining $(\rho,\alpha)$ as the only variables, this finally gives:
\begin{eqnarray}\nonumber&&p\left(\frac{\rho^4+8\rho^2-6\rho^3\cos{\alpha}-2\rho\cos{\alpha}-1}{(\rho^2-2\rho\cos{\alpha}+1)^2}+1\right)F+\frac{2\rho(\rho^2-1)}{\rho^2-2\rho\cos{\alpha}+1}F_\rho \\ \label{oddBS1}&&-\frac{4\rho\sin{\alpha}}{\rho^2-2\rho\cos{\alpha}+1}F_\alpha+\kappa F_{\alpha\alpha}=0.\end{eqnarray}
In terms of the $z,\,\bar z$ variables, writing $F=F(z,\bar z)$, this equation becomes 
\begin{eqnarray}\nonumber&&-\frac{\kappa}{4}(z\partial_{z}-\bar z\partial_{\bar z})^2F+\frac{z^2+1}{z^2-1}z\partial_{z} F+\frac{\bar z^2+1}{\bar z^2-1}\bar z\partial_{\bar z}F\\ \label{zzodd}
&&+p\left[\frac{1}{1-z^2}+\frac{1}{1-\bar z^2}-\frac{2}{(1-z^2)^2}-\frac{2}{(1-\bar z^2)^2}+2\right]F=0.\end{eqnarray}
Naturally, defining $\zeta=z^2$, we can also rewrite this equation in the $\zeta,\bar{\zeta}$ variables, with now $F=F(\zeta,\bar\zeta)$, as 
\begin{eqnarray}\nonumber&&-\frac{\kappa}{2}\big(\zeta\partial_{\zeta}-\bar{\zeta}\partial_{\bar{\zeta}}\big)^2F+\frac{\zeta+1}{\zeta-1}\zeta\partial_{\zeta} F+\frac{\bar{\zeta}+1}{\bar{\zeta}-1}\bar{\zeta}\partial_{\bar{\zeta}}F\\
&&-\frac{p}{2}\left[-\frac{1}{1-\zeta}-\frac{1}{1-\bar{\zeta}}+\frac{2}{(1-\zeta)^2}+\frac{2}{(1-\bar{\zeta})^2}-2\right]F=0.\label{BSzeta}\end{eqnarray}
\begin{rkk}\label{BSBSo}In the $\zeta, \bar\zeta$ variables, the equation \eqref{BSzeta} for the oddified moment function has exactly the same differential part as the BS one \eqref{zz1}, the difference being only in the singular function $B(\zeta)+B(\bar\zeta)$ multiplying the $F$ term. Notice that this function also vanishes at $\zeta=\bar \zeta=0$, hence fom Lemma \ref{lem0}, the space of solutions which are double power series is one-dimensional. \end{rkk}
\subsubsection{Special solutions}
We can now argue as in Section \ref{wpsle}, and look for solutions of \eqref{BSzeta} of the form \begin{equation}\label{Fzeta}F(\zeta,\bar\zeta)=\vphi_\alpha(\zeta)\overline{\vphi}_\alpha(\bar\zeta) P(\zeta\bar\zeta),\end{equation} where $\vphi_\alpha(\zeta):=(1-\zeta)^\alpha$; $P$ is thus rotationally invariant. 

The restriction of Eq.  \eqref{BSzeta} to  $\bar\zeta=0$ gives the boundary operator: 
\begin{equation}\label{PBzeta}\mathcal{P}(\partial)\vphi:=-\frac{\kappa}{2}(\zeta\partial_\zeta)^2\vphi+\frac{\zeta+1}{\zeta-1}\zeta\partial_\zeta\vphi+\frac{p}{2}\left(\frac{1}{1-\zeta}-\frac{2}{(1-\zeta)^2}+1\right)\vphi,\end{equation}
resulting in the boundary equation $\mathcal{P}(\partial)(\vphi_\alpha)=0$. One easily finds  
$$\mathcal{P}(\partial)(\vphi_\alpha)=A_2(\alpha)\vphi_\alpha+B_2(\alpha)\vphi_{\alpha-1}+C_2(\alpha)\vphi_{\alpha-2},$$ with  
\begin{eqnarray}\label{A'}
A_2(\alpha)&:=&-\kappa \alpha^2/2+\alpha+p/2,\\\label{D'}
B_2(\alpha)&:=& \kappa\alpha^2-\kappa\alpha/2-3\alpha+p/2,\\\label{C'}
C(\alpha)&=&-\kappa\alpha^2/2+(\kappa/2+2)\alpha-p.
\end{eqnarray}
As before, we see that $A_2+B_2+C=0$, and solving for $A_2(\alpha)=B_2(\alpha)=0$ now gives the special set of values
\begin{equation}\label{p2k} \alpha=\alpha_2(\kappa):=\frac{8+\kappa}{3\kappa},\,\,\,p=p_2(\kappa):=\frac{(8+\kappa)(2+\kappa)}{9\kappa}.\end{equation}

For these values of $\alpha$, we look for a solution $P(\zeta\bar\zeta)$ of the singular equation analogous to Eqs. \eqref{Psing} and \eqref{PP'} in Section \ref{wpsle}. Because of Remark \ref{BSBSo}, when plugging the factorized form \eqref{Fzeta} into Eq. \eqref{BSzeta}, one obtains the same singular operator as in \eqref{zzb}:
$$ \mathcal P_{\textrm{sing}}(D)[P]=\left[-\frac{\kappa}{2}\overline{\vphi}(\zeta\partial)^2\vphi
-\frac{\kappa}{2}\vphi (\bar{\zeta}\bar{\partial})^2 \overline{\vphi}
+\overline{\vphi}\frac{\zeta+1}{\zeta-1}\zeta\partial\vphi
+\vphi\frac{\bar{\zeta}+1}{\bar{\zeta}-1}\bar{\zeta}\bar{\partial}\,\overline{\vphi}\right]P.$$
As a consequence, the singular equation $\mathcal P_{\textrm{sing}}(D)[P]=0$ is the same as in \eqref{Psing}, with the same solution \eqref{Pzzbar}, now in the $\zeta,\bar\zeta$ variables:
\begin{equation}\label{Pzetazetabar}P(\zeta\bar \zeta)=(1-\zeta\bar{\zeta})^{-\beta},\,\,\,\beta=\frac{\kappa\alpha^2}{2}.\end{equation}

 We thus can use Remark \ref{BSBSo} and Lemma  \ref{lem0} to conclude to the unicity of the solution $F$ to \eqref{BSzeta} with value $1$ at $0$. This yields, in the original $z$ variable:
\begin{theo} \label{main3}The oddified whole-plane SLE$_{\kappa}$ map $h_0(z)$ has derivative moments  
\begin{eqnarray*}\mathbb E\big[(h'_0(z))^{p/2}\big]&=&(1-z^2)^\alpha,\\
\mathbb E\big[|h'_0(z)|^p\big]&=&\frac{(1-z^2)^\alpha(1-{\bar z}^2)^\alpha}{(1-z^2{\bar z}^2)^\beta},
\end{eqnarray*}
for the special set of exponents $p=\kappa\alpha(\alpha+1)/4={(8+\kappa)(2+\kappa)}/{9\kappa}$, with $\alpha=(8+\kappa)/3\kappa$ and 
$\beta=\kappa \alpha^2/2=(8+\kappa)^2/18\kappa$.
\end{theo}
Notice that $p=2$ if and only if $\kappa=4$, in which case $\alpha=1\,,\beta=2.$ 
\begin{cor}\label{F06-2}
For $\kappa=4$, the oddified whole-plane SLE$_{\kappa}$ map $h_0(z)$ has first and second derivative moments:
\begin{eqnarray*}\mathbb E(h'_0(z))=1-z^2,\,\,\,\,
F(z,\bar z)=\mathbb E(|h'_0(z)|^2)=\frac{(1-z^2)(1-{\bar z}^2)}{(1-z^2{\bar z}^2)^2}.
\end{eqnarray*}
\end{cor}
\textcolor{black}{By considering the terms that are powers of $(z\bar z)^2$ in the double expansion \eqref{expansionodd} of $F(z,\bar z)$ in Corollary \ref{F06-2}, or by computing the latter's integral means \eqref{expansionoddbis},} one finally proves assertion {\it (iii)} in  Theorem \ref{theolle} for the oddified whole-plane SLE$_{\kappa=4}$, i.e., when the driving function of the whole-plane Loewner process is the exponential of a  Brownian motion. 

\subsubsection{Oddified L\'evy--Loewner Evolution} 
 In the case where the driving function in the oddified Loewner equation \eqref{loewnerodd}
is  the complex exponential of a L\'evy process $(L_t)$,
\begin{equation} \label{LLEzit}
\lambda(2t)=\xi_t:=e^{iL_{2t}},
\end{equation}
and in order to compute $F(z,\bar z)$ \eqref{Fh0}, we follow the same martingale argument as in Section \ref{martarg}  (see Eqs. \eqref{slezetat}, \eqref{Nmart}, \eqref{zs}, and \eqref{tildeh}). 

 The characteristic function of the exponential's argument,  $\frac 12L_{2t}$, associated with $\sqrt{\xi_t}$ is 
$$ \mathbb{E}\left(e^{\xi\frac{i}{2}L_{2t}}\right)=e^{-2t\eta({\xi}/{2})}.$$
The L\'evy generator $\Lambda$ is thus defined by its action on the characters 
\begin{equation} \label{Lambdaodd}\Lambda(e^{in\theta}):=-2\eta\left(\frac n2\right)e^{in\theta}, n\in\bb Z.\end{equation}
Eq.  \eqref{oddBS} in Section \ref{martarg} now becomes \begin{eqnarray*}&&p\left(\frac{\rho^4+8\rho^2-6\rho^3\cos{\alpha}-2\rho\cos{\alpha}-1}{(\rho^2-2\rho\cos{\alpha}+1)^2}+1\right)F+\frac{(\rho^2-1)}{\rho^2-2\rho\cos{\alpha}+1}rF_r\\
&&-\frac{2\rho\sin{\alpha}}{\rho^2-2\rho\cos{\alpha}+1}F_\theta+ \Lambda F=0.\end{eqnarray*}
By retaining, as in Eq. \ref{oddBS1},  $\rho=r^2$ and $\alpha=2\theta$ as the only variables, this finally gives:
\begin{eqnarray*}&&p\left(\frac{\rho^4+8\rho^2-6\rho^3\cos{\alpha}-2\rho\cos{\alpha}-1}{(\rho^2-2\rho\cos{\alpha}+1)^2}+1\right)F+\frac{2\rho(\rho^2-1)}{\rho^2-2\rho\cos{\alpha}+1}F_\rho\\
&&-\frac{4\rho\sin{\alpha}}{\rho^2-2\rho\cos{\alpha}+1}F_\alpha+\tilde{\Lambda}F=0,\label{BS2}\end{eqnarray*}
where the rescaled generator $\tilde \Lambda$ is defined  so that 
\begin{equation}\label{tildeLambdaodd}\tilde{\Lambda}(e^{in\alpha})=-2\eta(n)e^{in\alpha}, \,n\in\bb Z.
\end{equation}
In terms of the $z,\,\bar z$ variables, and writing $F=F(z, \bar z)$, this equation becomes 
\begin{eqnarray}\nonumber&&\Lambda F+\frac{z^2+1}{z^2-1}z\partial_{z} F+\frac{\bar z^2+1}{\bar z^2-1}\bar z\partial_{\bar z}F\\ \label{FLzzodd}
&&+p\left[\frac{1}{1-z^2}+\frac{1}{1-\bar z^2}-\frac{2}{(1-z^2)^2}-\frac{2}{(1-\bar z^2)^2}+2\right]F=0,\end{eqnarray}
where the original generator $\Lambda$ \eqref{Lambdaodd} now acts on monomials as $$\Lambda(z^k\bar z^{\ell})=-2\eta\left(\frac{k-l}{2}\right)z^k\bar z^{\ell}.$$
In the $p=2$ case, observe that $F(z,\bar z)$ in Corollary \ref{F06-2}
 involves only chiral terms of the form $e^{\pm2i\theta}$. For $n=\pm 2$, the L\'evy generator $\Lambda$ \eqref{Lambdaodd} acts on these terms by multiplying them by $-2\eta(1)$. If $\eta_1:=\eta(1)=2$, we see that this action is the same as that of the Brownian generator in the SLE$_{\kappa=4}$ case. Therefore, we obtain the following proposition, generalized in the following section: \begin{prop}\label{oddeta1}
 Corollary \ref{F06-2} for the derivative second moment goes through, in the oddified L\'evy setting, under the sole condition that  $\eta_1=2$. \textcolor{black}{This completes the proof of case {\it (iii)} of Theorem \ref{theolle}.} 
 \end{prop}
 \subsubsection{Oddified L\'evy--Loewner moments}\label{oddllesle}
\begin{theo}\label{levyodd}
If a L\'evy process has its first $m\, (\geq 1)$ symbols given by $\eta_j=\kappa j^2/2,  1\leq j\leq m,$ with 
$\kappa=8/(3m-1)$,  then the associated odd L\'evy--Loewner map $h_0(z)=z\sqrt{f_0(z^2)/z^2}$, where $f_0(z)$ is the whole-plane LLE, has the same derivative moments of order $p$ as for SLE$_\kappa$, for the particular value of the exponent, $p=2m(m+1)/(3m-1)$, as given in Theorem \ref{main3} with $\alpha=m, \beta=4m^2/(3m-1).$   
\end{theo}
\begin{rkk}
The case $m=1$ gives the condition $\eta_1=2$ with an equivalent SLE$_\kappa$ parameter  $\kappa=4$, with $ p=2, \alpha=1, \beta=2$. \end{rkk}
\begin{proof}
For the odd case of the whole-plane {L\'evy--Loewner} evolution, the BS-like equation \eqref{FLzzodd} becomes,  when  setting $F=F(\zeta,\bar{\zeta})$ in the $\zeta=z^2$ variable,

\begin{equation}\label{LLeeodd} \tilde \Lambda F+2\frac{\zeta+1}{\zeta-1}\zeta\partial_{\zeta} F+2\frac{\bar{\zeta}+1}{\bar{\zeta}-1}\bar{\zeta}\partial_{\bar{\zeta}}F
+p\left[\frac{1}{1-\zeta}+\frac{1}{1-\bar{\zeta}}-\frac{2}{(1-\zeta)^2}-\frac{2}{(1-\bar{\zeta})^2}+2\right]F =0.\end{equation}
The action of the L\'evy infinitesimal generator $\tilde \Lambda$ \eqref{tildeLambdaodd} on a term 
$\zeta^k {\bar \zeta}^\ell$ is $$\tilde \Lambda\, [\zeta^k {\bar \zeta}^\ell]=-2\eta(k-\ell)\zeta^k {\bar \zeta}^\ell,$$ where $\eta(\cdot)$ here is real and even. It is such that $\eta(0)=0$, 
therefore for any $n\in \mathbb Z$,$$\tilde \Lambda\, [(\zeta\bar \zeta)^n]=0.$$
For the set of solutions $F(\zeta,\bar \zeta)=(1-\zeta\bar \zeta)^{-\beta}\varphi_\alpha(\zeta)\varphi_\alpha(\bar\zeta)$,  as given in Theorem \ref{main3}, we thus have 
\begin{equation} \label{lambdaalphaodd}\tilde \Lambda F(\zeta,\bar \zeta)=\frac{1}{(1-\zeta\bar \zeta)^\beta}\tilde \Lambda [(1-\zeta)^\alpha(1-\bar \zeta)^{\alpha}].\end{equation}
If the exponent $\alpha={(8+\kappa)}/{3\kappa}$ equals an {\it integer} $m\geq 1$,  $\vphi_\alpha=(1-\zeta)^\alpha$ is polynomial of order $m$, and $\tilde \Lambda [\vphi_\alpha{\bar \vphi}_\alpha]$ contains 
only the finite set of L\'evy symbols $\{\eta_1,\ldots,\eta_m\}$. If this set co\"{i}ncides with the set of values $(\kappa/2) \ell^2$ for $\ell=\{1,\cdots,m\}$, the action of the L\'evy generator 
$\tilde \Lambda$ on $\vphi_\alpha\bar{\vphi}_\alpha$  co\"{i}ncides with that of the Brownian generator. In this case, $F(\zeta,\bar \zeta)$, solution to the SLE$_\kappa$ equation \eqref{BSzeta}
 is also solution to the L\'evy--Loewner differential equation \eqref{LLeeodd}, and Theorem \ref{main3} is also valid for an oddified L\'evy--Loewner evolution with such symbols.\end{proof}
 \subsection{Generalization to processes with m-fold symmetry}\label{subsecm-fold}
The preceding results may be generalized to the case of functions with $m$-fold symmetry. These are  functions of the form
$[f(z^m)]^{1/m}$ with $f\in \mathcal{S}$ and $m\in\mathbb N,\,m\geq 1$. The case of odd functions corresponds to $m=2$; equivalently, the functions with $m$-fold symmetry are functions in $\mathcal{S}$ whose Taylor series has the form 
$f(z)=\sum_{k=0}^\infty a_{mk+1}z^{mk+1}.$ 
As for the oddification case, we can associate to $f_0$, where $(f_t)$ is a whole-plane SLE$_\kappa$, its $m$-folded version ${h}^{(m)}_0(z):=[f_0(z^m)]^{1/m}.$ 
By setting $\zeta:=z^m$, we obtain the following Beliaev--Smirnov-like equation for 
$F(\zeta,\bar \zeta):=\mathbb{E}\left(|({h}^{(m)}_0)'(z)|^p\right)$
\textcolor{black}{\begin{eqnarray}\label{BSmm} &&\mathcal P_m(D) [F(\zeta,\bar\zeta)]=0\\ \nonumber&&\mathcal P_m(D):=-\frac{\kappa}{2}(\zeta\partial_{\zeta}-\bar{\zeta}\partial_{\bar{\zeta}})^2+\frac{\zeta+1}{\zeta-1}\zeta\partial_{\zeta} +\frac{\bar{\zeta}+1}{\bar{\zeta}-1}\bar{\zeta}\partial_{\bar{\zeta}}\\ \label{BSm}&&+\frac{p}{m}\left[\frac{m-1}{1-\zeta}+\frac{m-1}{1-\bar{\zeta}}-\frac{m}{(1-\zeta)^2}-\frac{m}{(1-\bar{\zeta})^2}+2\right],\end{eqnarray}}
We then look for special solutions of the form $\vphi_\alpha(\zeta)\overline{\vphi}_\alpha(\bar \zeta) P(\zeta\bar \zeta)$ where $\vphi_\alpha(\zeta):=(1-\zeta)^\alpha=(1-z^m)^\alpha$ and \textcolor{black}{$\overline{\vphi}_\alpha(\bar \zeta):=\overline{\varphi_\alpha (\zeta)}=(1-\bar\zeta)^\alpha=(1-\bar z^m)^\alpha$}. 
For $\vphi_\alpha$, we look for solutions of the boundary equation for $\bar \zeta=0$
$$\mathcal P_m(\partial)[\varphi_\alpha]=-\frac{\kappa}{2}(\zeta\partial_\zeta)^2\vphi_\alpha+\frac{\zeta+1}{\zeta-1}\zeta\partial_\zeta\vphi_\alpha+\frac{p}{m}\left(\frac{m-1}{1-\zeta}-\frac{m}{(1-\zeta)^2}+1\right)\vphi_\alpha=0.$$ We identically have
\begin{eqnarray}\label{Pphialpham}
\mathcal P_m(\partial)[\vphi_\alpha]&=&A_m(\alpha) \vphi_\alpha+B_m(\alpha)\vphi_{\alpha-1}+ C(\alpha) \vphi_{\alpha-2},\\
\label{Aalpham}A_m(\alpha)&:=&- \frac{\kappa}{2}\alpha^2+\alpha+\frac{p}{m},\\ \label{Balpham}
B_m(\alpha)&:=&\frac{\kappa}{2}\alpha(2\alpha-1)-3\alpha+\left(1-\frac{1}{m}\right)p,\\ \label{Calpham}
C(\alpha)&=& -\frac{\kappa}{2}\alpha(\alpha-1)+2\alpha-p.  
\end{eqnarray} Notice that we again have the identity $A_m+B_m+C=0.$ Setting the conditions $A_m(\alpha)=0, C(\alpha)=0$ so that \eqref{Pphialpham} vanishes, gives the special set of values:
\begin{equation}\label{alpham}\alpha=\alpha_m(\kappa):=\frac{2m+4+\kappa}{(m+1)\kappa},\,\,\,p=p_m(\kappa):=\frac{m(2m+4+\kappa)(2+\kappa)}{2(m+1)^2\kappa}.\end{equation}
For the rotationally invariant pre-factor $P(\zeta,\bar\zeta)$, Eq. \eqref{BSm} shows that the resulting singular equation \eqref{Psing}, $\mathcal P_{\textrm{sing}}(D)[P]=0$, does not depend on $m$, so that 
$P(\zeta,\bar\zeta)=(1-\zeta\bar \zeta)^{-\beta}$, with $\beta={\kappa\alpha^2}/{2}.$ This leads to 
\begin{theo} \label{main4}The $m$-fold whole-plane SLE$_{\kappa}$ map $h^{(m)}_0(z)$ has derivative moments  
\begin{eqnarray*}\mathbb E\big[\big((h^{(m)}_0)'(z)\big)^{p/2}\big]&=&(1-z^m)^\alpha,\\
\mathbb E\big[|(h^{(m)}_0)'(z)|^p\big]&=&\frac{(1-z^m)^\alpha(1-{\bar z}^m)^\alpha}{(1-z^m{{\bar z}^m})^\beta},
\end{eqnarray*}
for the special set of exponents $p=p_m(\kappa)={m(2m+4+\kappa)(2+\kappa)}/{2(m+1)^2\kappa}$, with $\alpha=\alpha_m(\kappa)=(2m+4+\kappa)/(m+1)\kappa$ and 
$\beta=\kappa \alpha^2_m(\kappa)/2=(2m+4+\kappa)^2/2(m+1)^2\kappa$.
\end{theo}
The case $p=2$ is of special interest, since it allows one to find the moments $\mathbb{E}\left(\I a_{mk+1}\I^2\right)$ from Plancherel formula. Setting $p=2$ in the above, and solving for $\kappa$ yields
\begin{equation}\label{kappam}\kappa=2m,\,\,\, \textrm{or}\,\,\,\kappa=\frac{2(m+2)}{m}.\end{equation}
In the first case, we have $\alpha=2/m,\,\,\beta=4/m$,  thus
\begin{equation}\label{F1m}F(\zeta,\bar \zeta)=\left(\frac{(1-\zeta)(1-\bar\zeta)}{(1-\zeta\bar{\zeta})^2}\right)^{2/m};\,\,\,\kappa=2m.\end{equation}
In the second case, we have $\alpha=1,\,\,\beta=\frac{m+2}{m}$, and
\begin{equation}\label{F2m} F(\zeta,\bar \zeta)=\frac{(1-\zeta)(1-\bar{\zeta})}{(1-\zeta\bar{\zeta})^{\frac{m+2}{m}}};\,\,\,\kappa=\frac{2(m+2)}{m}.\end{equation}
Let us detail all possibilities with $m\leq 4$:
\begin{enumerate}[(i)]
\item $m=1$ yields the two cases $\kappa=6,\,\kappa=2$, corresponding to the two first cases of Theorem \ref{theolle}.
\item $m=2$ gives rise to the single value $\kappa=4$, and to the third case in Theorem \ref{theolle}.
\item $m=3$ corresponds to $\kappa=6$ and $\kappa=10/3$, with respective $F$-functions
$$F(\zeta,\bar \zeta)=\frac{(1-\zeta)^{2/3}(1-\bar{\zeta})^{2/3}}{(1-\zeta\bar{\zeta})^{4/3}};\,\,\,\,F(\zeta,\bar\zeta)=\frac{(1-\zeta)(1-\bar{\zeta})}{(1-\zeta\bar{\zeta})^{5/3}}.$$
\item For $m=4$, one gets $\kappa=8$ or $3$, with respective $F$-functions:$$F(\zeta,\bar \zeta)=\frac{(1-\zeta)^{1/2}(1-\bar{\zeta})^{1/2}}{1-\zeta\bar{\zeta}};\,\,\,\,F(\zeta,\bar\zeta)=\frac{(1-\zeta)(1-\bar{\zeta})}{(1-\zeta\bar{\zeta})^{3/2}}.$$
\end{enumerate}
Let us return for $p=2$ to general values of $m$, with $\kappa$ given by \eqref{kappam}, and compute $\mathbb{E}\left(\I a_{mk+1}\I^2\right)$. Write
$$ (1-x)^{\alpha}=\sum_{k=0}^{\infty}\lambda_k(\alpha) x^k,$$
with coefficients $\lambda_k(\alpha):=((-1)^k/k!)\,\alpha (\alpha-1)\cdots(\alpha-k+1)$. In the first case, $\kappa=2m$, we have from \eqref{F1m}
 \begin{equation}\mathbb{E}\left(\I a_{mk+1}\I^2\right)=\frac{\sum_{j=0}^{k}\lambda_j^2(2/m)\I \lambda_{k-j}(-4/m)\I}{(mk+1)^2}.\label{2m}\end{equation}
In the second case, $\kappa=\frac{2(m+2)}{m}$, Eq. \eqref{F2m} gives
\begin{equation}\mathbb{E}\left(\I a_{mk+1}\I^2\right)=\frac{\lambda_k(\frac{m+2}{m})+\lambda_{k-1}(\frac{m+2}{m})}{(mk+1)^2}=\frac{\prod_{j=0}^{k-1}(jm+2)}{(mk+1)m^kk!};\label{autre}\end{equation}
 for $m=1,2$ one recovers the values already computed.

In conclusion, we have found infinitely many cases where one may exactly compute the variances of the coefficients of whole-plane SLE. The following cases correspond to some physically \textcolor{black}{significant} situations:
\begin{enumerate}[(a)]
\item $m=1,\kappa=6$ with formula \eqref{autre} (percolation \cite{Schr,MR1851632});
\item $m=1,\kappa=2$ with formula \eqref{2m} (loop-erased random walk  \cite{MR2044671,Schr});
\item $m=2,\kappa=4$ with formula \eqref{2m} or \eqref{autre} (Gaussian free field contour lines \cite{MR2486487});
\item $m=3,\kappa=6$ with formula \eqref{2m} (percolation \cite{MR1851632});
\item $m=4,\kappa=8$ with formula \eqref{2m} (spanning trees \cite{Schr});
\item $m=4,\kappa=3$ with formula \eqref{autre} (critical Ising model \cite{2009arXiv0910.2045C,MR2680496});
\item $m=6,\kappa=8/3$ with formula \eqref{autre} (self-avoiding walk \cite{MR2112128,MR2112128bis}).
\end{enumerate}
 \section{Multifractal spectra for infinite whole-plane SLE}\label{multifractal} 
\textcolor{black}{The aim of this section is to give compelling arguments that support Statement \ref{theoMF} concerning the  explicit averaged  integral means spectra, as defined in \eqref{betavdef}, for the interior whole-plane SLE map $f_0(z)$, its oddified version $h_0(z)$,  or  its  $m$-fold transforms $h_0^{(m)}(z)$ (which generalize $f_0(z)=h_0^{(1)}(z)$ and $h_0(z)=h_0^{(2)}(z)$). On the rigorous side, we establish Theorem \ref{theoMFc} and Theorem \ref{theoMFcm}:  that these  spectra do have a phase transition for $p$ large enough (respectively at \eqref{pstar}, and at \eqref{pmstar} or \eqref{pdoublestar0}). Above this phase transition, they are bounded below by the multifractal spectrum $B_m(p,\kappa)$ \eqref{Bm} that appears on  the right hand-side of formulae \eqref{beta} and \eqref{betam}, up to the special points $p_m(\kappa)$ \eqref{pkappa1} and \eqref{pkappam}. Thereafter, they  are bounded above by the same expression. 
 We strongly believe these bounds to be exact as in Statement \ref{theoMF}.} Let us begin with some relevant results for integral means spectra and related multifractal spectra. 
\subsection{Integral means spectrum}  \label{introductionI}
\subsubsection{SLE's harmonic measure spectra} 
The general theory of the integral means spectra, and of the associated multifractal properties of the harmonic measure, have been the subject of  important  pioneering works, among which stand out those of L. Carleson, P. Jones  and N. Makarov 
\cite{1994ArM....32...33C,JonesMak,Makadist,Makanaliz}. The search for {\it universal} spectra, which provide universal functions as upper-bounds, have lead to well-known results and conjectures which we briefly recall below (Section \ref{universal}).

In the case of conformally invariant critical curves, i.e., SLEs, the multifractal spectrum associated with the harmonic measure near those curves  was first obtained from quantum gravity methods by the first author \cite{1999PhRvL..82..880D,1999PhRvL..82.3940D,2000PhRvL..84.1363D,MR1964687,MR2112128,2006math.ph...8053D}. 
  This was extended to the {\it mixed} multifractal spectrum  describing both the singularities and the winding of  equipotentials near a conformally invariant curve \cite{PhysRevLett.89.264101}. Another heuristic derivation of the harmonic measure spectrum was obtained from a Laplacian growth equation, similar to Eq. \eqref{BS} here, but for  chordal SLE \cite{PhysRevLett.88.055506}. 
  The corresponding SLE integral means spectrum was later rigorously established, in an expectation sense, by Beliaev and Smirnov \cite{BS}, starting from Beliaev's thesis \cite{BKTH}. These authors used the very same equation as Eq. \eqref{BS} here, that they derived precisely for that purpose, in their case for the exterior whole-plane SLE.
  
    Another method, the so-called ``Coulomb gas'' approach of conformal field theory, is also applicable \cite{PhysRevLett.95.170602,2007JPhA...40.2165R}, and was extended to the mixed multifractal spectrum \cite{2008JPhA...41B5006B,2008NuPhB.802..494D}.  Amazingly, these predictions for the fine structure of the harmonic measure were tested numerically, and successfully, for percolation and Ising clusters \cite{2008PhRvL.101n4102A}. 
   
 Let us mention that Chen and Rohde obtained in Ref. \cite{2009CMaPh.285..799C}  derivative estimates  for the (chordal) Loewner evolution driven by a {\it symmetric $\alpha$-stable process}, and showed that its hull has Hausdorff dimension $1$, thereby presenting a {\it non-multifractal} behavior. Similar results was found by Johansson and Sola \cite{Johansson2009238} for a random growth model obtained by driving the Loewner equation by a compound Poisson process. We therefore restrict our study here to the {\it interior} whole-plane SLE curve.

    In the radial setting, the integral means spectrum \eqref{betadef}-\eqref{betavdef} is associated with the divergent behavior of the moments of order $p$ of the map derivative's modulus $|f'(z)|$ near the unit circle $\partial \mathbb D$, possibly augmented by the extra singular behavior of the map at  point $z=1$, where the SLE driving function originates at time $t=0$. When the latter singular behavior starts to dominate, in fact when $p$ becomes negative enough \cite{BS}, the integral means spectrum undergoes a phase transition, after which the harmonic measure behavior is dominated by the {\it tip} of the SLE curve. This was observed in Ref. \cite{PhysRevLett.88.055506}, while the tip spectrum was later obtained rigorously, in the sense of {\it expectations}  in Ref. \cite{BS}, and  in an {\it almost sure} sense in Ref. \cite{0911.3983}.
  
\subsubsection{Universal spectra}\label{universal}
Given $f$ holomorphic and injective in the unit disk, we define, for $p\in\R$,
$$\beta_f(p)=\limsup_{r\to 1^{-}}\frac{\ln \int_0^{2\pi}\I f'(re^{i\theta})\I^p d\theta}{\ln \frac{1}{1-r}},$$
so that $\beta_f(p)$ is the smallest number $q$ such that there exists a $C>0$ 
$$ \int_0^{2\pi}\I f'(re^{i\theta})\I^p d\theta\leq \frac{C}{(1-r)^{q+\veps}}$$
as $r\to 1$, and for every $\veps>0$.

\begin{theo}\label{FMcG} If $f$ is holomorphic and injective in the unit disk, then
$$\beta_f(p)\leq 3p-1,\,\,\,2/5\leq p<\infty.$$
If moreover $f$ is bounded, then 
$$\beta_f(p)\leq p-1,\,\,\,p\geq 2.$$
Both exponents are sharp, the first one being attained for the Koebe function.
\end{theo}
This theorem is due to Feng and McGregor \cite{FengMcG}. For a proof, consult, e.g., \cite{Pommerenke}.  We will need below the following variant of this theorem: 
\begin{theo}\label{FMcGo} Let $f$ be a function injective and holomorphic in the unit disk, such that $f(0)=0$, and let us denote by $h(z)$ its oddification:
$$h(z):=z\sqrt{f(z^2)/z^2}.$$ Then we have $\beta_h(p)\leq 2p-1$ for $p\geq 2/3$. This bound is attained for the oddification of the Koebe function. 
\end{theo}
 This variant is originally due to Makarov \cite{Makanaliz}; a proof that parallels that of Feng and McGregor  is given below in Appendix C \ref{McGo}. 

In the sequel we will need some facts about {\it three different universal spectra}, respectively for the schlicht class $\mathcal{S}$, the subclass $\mathcal{S}_b$ of bounded functions, and the subclass $\mathcal{S}_o$ of odd functions. For $p\in\bb R$,  define:
\begin{eqnarray*}B_{\mathcal{S}}(p)&:=&\sup\{\beta_f(p),\,f\in\mathcal{S}\};\,\,\,B_{\mathcal{S}_b}(p):=\sup\{\beta_f(p),\,f\in\mathcal{S}_b\};\\
B_{\mathcal{S}_o}(p)&:=&\sup\{\beta_f(p),\,f\in\mathcal{S}_o\}.
\end{eqnarray*}
The $B_{\mathcal{S}_b}$ spectrum is the most studied one in the literature. 
By Theorems \ref{FMcG} and  \ref{FMcGo}, respectively:
\begin{eqnarray*}B_{\mathcal{S}}(p)&=&3p-1 \,\,\mathrm{for}\,\,p\geq 2/5;\,\,\,\,B_{\mathcal{S}_b}(p)= p-1\,\,\mathrm{for}\,\,p\geq 2;\\
B_{\mathcal{S}_o}(p)&=&2p-1\,\,\mathrm{for}\,\,p\geq 2/3.\end{eqnarray*}
For the sake of completeness, let us briefly recall some known or conjectured results about these universal spectra. The main one is {\it Brennan conjecture}, which reads: $B(-2)=1.$ 
If true, this conjecture would imply that
$ B_{\mathcal{S}_b}(p)=\I p\I-1,$
for $p\leq -2$. This is not known, but Carleson and Makarov \cite{1994ArM....32...33C} have shown that there exists $p_0\leq -2$ such that $B_{\mathcal{S}_b}(p)=\I p\I-1$ for $p\leq p_0$. Notice that $B_{\mathcal{S}_b}(p)=\I p\I-1$ for $p\geq 2$. A rather speculative conjecture, named after Kraetzer, asserts that
$$B_{\mathcal{S}_b}(p)=\frac{p^2}{4},\,\,\,-2\leq p\leq 2.$$
Makarov \cite{Makanaliz} has proven that 
$$ B_{\mathcal{S}}(p)=\max(B_{\mathcal{S}_b}(p),3p-1),$$
so that if both Kraetzer and Brennan conjectures are true, then for $\I p\I\leq -2,\,B_{\mathcal{S}}(p)=\I p\I-1$, for $-2\leq p\leq 6-4\sqrt{2},\,B_{\mathcal{S}}(p)=\frac{p^2}{4}$, and for
 $p\geq 6-4\sqrt{2},\,B_{\mathcal{S}}(p)=3p-1.$
 
In our study, the {\it unbounded} character of the whole-plane maps under consideration plays a crucial role for the spectrum. This can already be seen  in the limit  $\kappa \to 0$, where the spectra should converge to that of the Koebe function \eqref{koebe}, ${\mathcal K}(z)=z/(1+z)^2$, hence to $\beta_{\mathcal K}(p)=3p-1$, or to $\beta_{{\mathcal K}_o}(p)=2p-1$ for its oddified version ${\mathcal K_o}(z)=z/(1+z^2)^2$.

The Theorem \ref{FMcGo} can be generalized to the class $\mathcal{S}_m$ defined for nonzero $m\in\mathbb{N}$ as the set of functions of the form
$$ h^{(m)}(z):=z[f(z^m)/z^m]^{1/m}\,,f\in\mathcal{S}.$$
The case of odd functions above corresponds to $m=2$. We may define as above
$$B_{\mathcal{S}_m}(p):=\sup\{\beta_h(p)\,,\,h\in\mathcal{S}_m\},$$
and a straightforward adaptation of the proof of Theorem \ref{FMcGo} (as given in Appendix C \ref{McGo}) gives \cite{Makanaliz} the
\begin{theo} \label{mtheo}For $p\geq \frac{2m}{m+4}$ we have 
$$  B_{\mathcal{S}_m}(p)=\frac{m+2}{m}p-1.$$
\end{theo}

\subsection{Integral means spectrum for unbounded whole-plane SLE} \label{IMS1}
  \subsubsection{Restriction to the unit circle} \label{restriction} The method introduced in \cite{BS} consists in finding approximate solutions to equation 
  \eqref{BS} in the vicinity of the unit circle for $|z|>1$, and near $z=1$.  We generalize it here to the {\it inner} whole-plane equation \eqref{BS}, for which $\sigma =-1$ and $|z|< 1$, 
  \begin{eqnarray}\nonumber&&p\left(\frac{r^4+4r^2(1-r\cos \theta)-1}{(r^2-2r\cos \theta+1)^2}-\sigma\right)F
+\frac{r(r^2-1)}{r^2-2r\cos \theta+1}F_r\\ \label{BSsig}&-&\frac{2r\sin \theta}{r^2-2r\cos \theta+1}F_\theta+\frac{\kappa}{2} F_{\theta\theta}=0.
\end{eqnarray}  
It is convenient to write this equation as 
 \begin{eqnarray}\label{BSsigD} p\left(\frac{N(r,\theta)}{D^2(r,\theta)}-\sigma\right)F
+\frac{r(r^2-1)}{D(r,\theta)}F_r-\frac{2r\sin \theta}{D(r,\theta)}F_\theta+\frac{\kappa}{2} F_{\theta\theta}=0,
\end{eqnarray} 
with 
\begin{eqnarray} 
\label{D}D(r,\theta)&:=&r^2-2r\cos\theta+1=|1-z|^2,\\ 
 \label{N}
N(r,\theta)&:=&r^4+4r^2(1-r\cos \theta)-1\\\nonumber
&=&2r^2D(r,\theta)+(r-1)(r^3-r^2+3r+1).
\end{eqnarray} 

We then look for approximate (but possibly exact) solutions  of the form
  \begin{eqnarray}\label{psi}
 \psi(r,\theta)
 =[-\sigma(1-r^2)]^{-\beta}g(r^2-2r\cos \theta+1)=[-\sigma(1-z\bar z)]^{-\beta} g(|1-z|^2).
 \end{eqnarray}
 Let us first remark that for any given value of $\kappa$, and for the special values  $p=p(\kappa)$ and  $\alpha=\alpha(\kappa)$ given in \eqref{kpgen},   the {\it exact solution found in Theorem \ref{main1} is precisely of the form} \eqref{psi}, with $\beta=(\kappa/2)\alpha^2$ and $g(x)=x^\alpha$. 
 It is thus necessarily a solution, for $p=p(\kappa)$, to the following explicit equation, obtained from \eqref{BSsigD},  $[-\sigma(1-r^2)]^{-\beta}$ being  further factored out:
 \begin{eqnarray} \nonumber p\left(\frac{N(r,\theta)}{D^2(r,\theta)}-\sigma\right) g 
-\frac{2r^2}{D(r,\theta)}\beta\,g--\frac{r(1-r^2)}{D(r,\theta)}(2r-2\cos\theta)g'\\ \label{BSsigDg}
-\frac{4r^2\sin^2 \theta}{D(r,\theta)}g'
+\frac{\kappa}{2}\big(2r\cos\theta \,g'+4r^2\sin^2\theta \,g''\big)=0,
\end{eqnarray} 
where $g= g(r^2-2r\cos\theta+1)$.

When $p$ is not equal to the special value $p(\kappa)$ of Eq. \eqref{kpgen}, the trial exponent $\beta$ and the function $g$ are determined from the {\it restriction} of Eq. \eqref{BSsigDg} to the unit circle $\partial \mathbb D$.  
One observes that on the unit circle \begin{eqnarray}\label{Dg}D(r=1,\theta)=2-2\cos\theta;\,\,\,\,N(1,\theta)=2D(1,\theta);\,\,\,\,g=g(2-2\cos\theta).\end{eqnarray}
Setting $r=1$ in Eq. \eqref{BSsigDg} and factoring out $[D(1,\theta)]^{-1}$, we therefore arrive at 
 \begin{eqnarray} \label{BSsigDg1} p\left[2-\sigma D(1,\theta)\right] g 
-2\beta\,g
-{4\sin^2 \theta}\,g'
+\frac{\kappa}{2}D(1,\theta)\big(2\cos\theta \,g'+4\sin^2\theta \,g''\big)=0.
\end{eqnarray} 
Define now $x:=2-2\cos\theta=D(1,\theta)$, such that $0\leq x\leq 4$; the equation on the unit circle simply becomes
 \begin{eqnarray} \label{BSsigDgx} 
 [p(2-\sigma x) -2\beta]\,g(x) +\left[\frac{\kappa}{2}(2-x)-(4-x)\right]x\,g'(x)+\frac{\kappa}{2} (4-x) x^2\,g''(x)=0.
\end{eqnarray} 
By homogeneity, for a function of the \emph{power law} form $g(x)=x^\gamma$, the left-hand side of \eqref{BSsigDgx} becomes 
$(c+dx)g(x)$, with
\begin{eqnarray*}
c&=&2p-2\beta-(\kappa+4)\gamma+2\kappa \gamma^2,\\
d&=&-\sigma p +\gamma-\frac{\kappa}{2}\gamma^2.
\end{eqnarray*}
We thus get a power law solution to \eqref{BSsigDgx} if and only if $c=0$ and $d=0$, i.e.,
\begin{eqnarray}\label{betapgammasig}
&&\beta=p-(\kappa+4)\frac{\gamma}{2} +\kappa\gamma^2,\\ \label{gammasig}
&&\gamma^2 \frac{\kappa}{2}-\gamma +\sigma p=0.
\end{eqnarray}
\textcolor{black}{Upon substituting $g(x)=x^\gamma g_0(x)$ into  \eqref{BSsigDgx}, we obtain
\begin{eqnarray} \nonumber
&& \big[2\beta(\gamma)-2\beta+x A^{\sigma}(\gamma)\big]\,x^\gamma g_0(x) \\ \label{BSsigDg0x}  &&+\left[\frac{\kappa}{2}(2-x)+(\kappa\gamma-1)(4-x)\right]\,x^{\gamma+1}g_0'(x)+\frac{\kappa}{2} (4-x)\, x^{\gamma+2}g_0''(x)=0,\\ \label{Bsigma}&& \beta(\gamma):=\kappa\gamma^2/2-C(\gamma),\\ \label{Asigma}
&&A^{\sigma}(\gamma):=A(\gamma)-(1+\sigma)p,\end{eqnarray}}\textcolor{black}{
where we recall definitions \eqref{Aalpha} and \eqref{Calpha} for $A$ and $C$:
\begin{eqnarray}\label{AC}
A(\gamma)=- \frac{\kappa}{2}\gamma^2+\gamma+p,\,\,\,\,\,
C(\gamma)= -\frac{\kappa}{2}\gamma^2+\big(\frac{\kappa}{2}+2\big)\gamma-p.
\end{eqnarray}
 We thus have
\begin{eqnarray}\label{betaCgamma}
&&\beta(\gamma)=\kappa\gamma^2/2-C(\gamma)=\kappa\gamma^2-(\kappa/2+2)\gamma -p,\\ \label{Asig}
&&A^{\sigma}(\gamma)=A(\gamma)-(1+\sigma)p=- \frac{\kappa}{2}\gamma^2+\gamma -\sigma p.
\end{eqnarray}
A power law solution, $g(x)=x^\gamma$ to Eq. \eqref{BSsigDgx}, i.e., $g_0$ constant, is obtained if the first line of Eq. \eqref{BSsigDg0x} vanishes, so  that 
\begin{eqnarray}\label{betapgamma}
&&\beta=\beta(\gamma)=\kappa\gamma^2/2-C(\gamma),\\ \label{gamma}
&&A^{\sigma}(\gamma)=0,
\end{eqnarray}
which is equivalent to equations \eqref{betapgammasig} and \eqref{gammasig}. 
Recall that for the interior whole-plane SLE considered here, we have $\sigma=-1$, hence $A^{(-1)}(\gamma)=A(\gamma)$, while in the exterior case considered by BS  \cite{BS} one has $\sigma=+1$, hence $A^{(+1)}(\gamma)=A(\gamma)-2p$.}  

\textcolor{black}{The solutions to Eqs. \eqref{betapgamma} and \eqref{gamma} are
\begin{eqnarray}\label{gammasigma}
&&\gamma^\sigma_{\pm}(p)=\frac{1}{\kappa}\big(1\pm \sqrt{1-2\sigma\kappa p}\big),\\
\label{betapgammasigma}
&&\beta^\sigma_{\pm}(p)=(1-2\sigma)p-\frac{\kappa}{2}\gamma^\sigma_{\pm}(p)=(1-2\sigma)p -\frac{1}{2}\big(1\pm \sqrt{1-2\sigma\kappa p}\big).
\end{eqnarray}
We have thus obtained a pair of power law solutions, 
\begin{eqnarray}\label{psisigma}
\psi^\sigma_{\pm}(z,\bar z):=[-\sigma(1-z\bar z)]^{-\beta^\sigma_{\pm}}g(|1-z|^2),\,\, g(x)=x^{\gamma^\sigma_{\pm}},
\end{eqnarray} to the boundary equation \eqref{BSsigDgx}.} 

\textcolor{black}{For the interior case $\sigma=-1$, we shall use hereafter the simplified notation
\begin{eqnarray}\label{gammafin}&&\gamma_{\pm}(p):=\gamma_{\pm}^{(-1)}(p)=\frac{1}{\kappa}\big(1\pm \sqrt{1+2\kappa p}\big),\\ \label{betapgammafin}
&&\beta_{\pm}(p):=\beta_{\pm}^{(-1)}(p)=3p-\frac{\kappa}{2}\gamma_{\pm}(p)=3p -\frac{1}{2}\big(1\pm \sqrt{1+2\kappa p}\big). 
\end{eqnarray}}
\textcolor{black}{Besides the power law solutions obtained here for the particular values \eqref{gammasigma} of $\gamma$, the second order differential equation \eqref{BSsigDg0x} for $g_0$ has a general class of  solutions which depends on the continuous parameter $\gamma$. Observe in particular that, for a given $\gamma$, the choice of parameter $\beta=\beta(\gamma)$ reduces the equation \eqref{BSsigDg0x} to the following  hypergeometric equation, which will be studied in Sections \ref{BSapproach} and \ref{secondsolrig}:
\begin{eqnarray} \label{hypergeom}  A^{\sigma}(\gamma)\, g_0(x)   +\left[\frac{\kappa}{2}(2-x)+(\kappa\gamma-1)(4-x)\right]\,g_0'(x)+\frac{\kappa}{2} (4-x)\, x g_0''(x)=0.\end{eqnarray}}  
\subsubsection{Action of the differential operator}\label{action}
\textcolor{black}{Let us consider a general function of the type
\begin{equation}\label{psior}
\psi(r,\theta)=[-\sigma(1-r^2)]^{-\beta}(r^2-2r\cos \theta+1)^\gamma.
\end{equation}
This function is of the form 
\begin{eqnarray}\label{psiphi}
\psi(z,\bar z)&=&[-\sigma(1-z\bar z)]^{-\beta} x^\gamma=[-\sigma(1-z\bar z)]^{-\beta} |1-z|^{2\gamma}\\ \nonumber
&=&[-\sigma(1-z\bar z)]^{-\beta} \varphi_\gamma(z)\varphi_\gamma(\bar z),
\end{eqnarray}
\textcolor{black}{where}
\begin{equation}\label{x}x=r^2-2r\cos \theta+1=|1-z|^2=1-(z+\bar z)+z\bar z.
\end{equation}}

It will prove useful to evaluate the action of the differential operator $\mathcal P(D)$ \eqref{PD} on the function \eqref{psiphi} for \textit{general values} of $\beta$ and $\gamma$ by using \eqref{zzb}, \eqref{Psing} and \eqref{Pphialpha}. 
\textcolor{black}{The general result, using the identity $A+B+C=0$ in \eqref{Pphialpha}, \eqref{Aalpha}, \eqref{Balpha2}, and  \eqref{Calpha}, is} 
\begin{eqnarray} \nonumber\frac{\mathcal P(D)[\psi(z,\bar z)]}{\psi(z,\bar z)}= (\kappa \gamma^2-2\beta)\frac{z\bar z}{x}+C(\gamma)\left[\frac{1-z\bar z}{x}\left(\frac{1-z\bar z}{x}+1\right)-\frac{2}{x}\right]\\\label{genres}-A(\gamma) \left(\frac{1-z\bar z}{x}-1\right)-(1+\sigma)p,\end{eqnarray} 
\textcolor{black}{where $A$ and $C$ are given by \eqref{AC}.} 
\textcolor{black}{Using \eqref{Bsigma} and \eqref{Asigma}, we can recast the above equation as 
\begin{eqnarray} \nonumber\frac{\mathcal P(D)[\psi(z,\bar z)]}{\psi(z,\bar z)}&=& \big(\beta(\gamma)-\beta\big)\frac{2}{x}+C(\gamma)\left(\frac{1-z\bar z}{x}\right)^2\\ \nonumber&&+\big(2\beta-2\beta(\gamma) -A(\gamma)-C(\gamma)\big)\frac{1-z\bar z}{x}\\ 
&&+A^\sigma(\gamma).\label{genresbis}\end{eqnarray} 
Note that 
\begin{eqnarray}\label{zzx0}
&&1-z\bar z=2\Re (1-z)-|1-z|^2=2x^{1/2}\cos\varphi -x,\\ \nonumber&&\varphi:=\arg (1-z).
\end{eqnarray}
Hence in the  $z\to 1$, $x\to 0$ limit, one has 
\begin{equation}\label{zbarzx}
1-z\bar z\sim 2 x^{1/2}\cos \varphi
\end{equation}
 along any ray passing through $1$, except if $\varphi=\pm\pi/2$, which corresponds to $z$ reaching $1$ tangentially to the unit circle $z\bar z=1$.}

\textcolor{black}{From the equivalence \eqref{zbarzx} for $x\to 0$, we conclude that the most singular terms in the action \eqref{genresbis} of the differential operator $\mathcal P(D)$ are the two terms on the r.h.s. of the first line, scaling like $x^{-1}$ and $(1-z\bar z)^2 x^{-2}$. They are furthermore independent of each other because the second one is parameterized by the angle $\varphi$.}
\subsubsection{\textcolor{black}{General action of the operator $\mathcal P(D)$}}
\textcolor{black}{Consider in this section the function \begin{eqnarray}\label{generalpsi}
&&\psi_0(z,\bar z):=P(z\bar z) g(x),\\ \nonumber
&&P(z\bar z) =[-\sigma(1-z\bar z)]^{-\beta},\\ \nonumber
&&g(x)=x^{\gamma}g_0(x),\,\,\, x=(1-z)(1-\bar z),\end{eqnarray}
where $g$ satisfies the boundary equation \eqref{BSsigDgx}, or, equivalently, $g_0$ satisfies \eqref{BSsigDg0x}, $\beta$ and $\gamma$ being considered here as parameters. After some calculation, we obtain:
\begin{eqnarray} \nonumber &&\frac{\mathcal P(D)[\psi_0(z,\bar z)]}{\psi_0(z,\bar z)}=(1-z\bar z)\left[\frac{1}{x}(\beta-p-\gamma)+\frac{1}{4-x}\big(-\beta +p(1-2\sigma)-\frac{\kappa}{2}\gamma\big)\right]\\ \nonumber &+&(1-z\bar z)\left[\left(\frac{\kappa}{2}-1-\frac{2\kappa}{4-x}\right)\frac{g'_0}{g_0}\right] \\ 
\nonumber &+&\frac{(1-z\bar z)^2}{x^2} \left[\frac{1}{4-x}\left(2p(1-2\sigma)-2\beta-\kappa x\frac{g'}{g}\right)+(\sigma-1)p+\big(\frac{\kappa}{2}+1\big)x\frac{g'}{g}\right],\\ \label{genrester}&&
\end{eqnarray} 
where 
$
x{g'}/{g}=\gamma+x{g_0'}/{g_0}.$}

\textcolor{black}{Substituting the particular value \eqref{betapgamma} $\beta=\beta(\gamma)=\kappa\gamma^2/2-C(\gamma)$ gives
\begin{eqnarray} \nonumber&&\frac{\mathcal P(D)[\psi_0(z,\bar z)]}{\psi_0(z,\bar z)}=(1-z\bar z)\left[-\frac{1}{x}\big[C(\gamma)+A(\gamma)\big]+\frac{1}{4-x}2A^\sigma(\gamma)+\left(\frac{\kappa}{2}-1-\frac{2\kappa}{4-x}\right)\frac{g'_0}{g_0}\right] \\ \nonumber &+&\frac{(1-z\bar z)^2}{x^2} \left\{\frac{1}{4-x}\left[2p(1-2\sigma)-2\beta(\gamma)-\kappa \big(\gamma+x\frac{g_0'}{g_0}\big)\right]+(\sigma-1)p+\big(\frac{\kappa}{2}+1\big)\big(\gamma+x\frac{g_0'}{g_0}\big)\right\}.\\ 
&& \label{genresquater0}
\end{eqnarray}
Using the identity:
$$
4A^\sigma(\gamma)=2p(1-2\sigma)-2\beta(\gamma)-\kappa\gamma, 
$$
we obtain 
\begin{eqnarray} \nonumber&&\frac{\mathcal P(D)[\psi_0(z,\bar z)]}{\psi_0(z,\bar z)}=(1-z\bar z)\left[-\frac{1}{x}\big[C(\gamma)+A(\gamma)\big]+\frac{1}{4-x}2A^\sigma(\gamma)+\left(\frac{\kappa}{2}-1-\frac{2\kappa}{4-x}\right)\frac{g'_0}{g_0}\right] \\ \nonumber &+&\frac{(1-z\bar z)^2}{x^2} \left\{\frac{1}{4-x}\left[4A^\sigma(\gamma)-\kappa x\frac{g_0'}{g_0}\right]+(\sigma-1)p+\big(\frac{\kappa}{2}+1\big)\big(\gamma+x\frac{g_0'}{g_0}\big)\right\}.\\ 
&& \label{genresquater}
\end{eqnarray}In the $x \to 0$ limit, assuming that $xg'_0(x)/g_0(x)=o(1)$, the second line is equivalent to 
\begin{eqnarray}\label{Cxto0}&&\frac{(1-z\bar z)^2}{x^2} \left\{A^\sigma(\gamma)+(\sigma-1)p+\left(\frac{\kappa}{2}+1\right)\gamma\right\}\\ \nonumber&&=\frac{(1-z\bar z)^2}{x^2} \left\{A(\gamma)-2p+\left(\frac{\kappa}{2}+1\right)\gamma\right\}  = \frac{(1-z\bar z)^2}{x^2} C(\gamma).\end{eqnarray}}
\subsubsection{\textcolor{black}{The Beliaev-Smirnov  approach}}\label{BSapproach}
\textcolor{black}{In this section we discuss the Beliaev-Smirnov approach of Ref. \cite{BS} to the standard BS spectrum \eqref{beta000}, and compare it to the formulation here.  
\begin{rkk}\label{BSCA0} $\bullet$ The study carried out in Ref. \cite{BS}  by Beliaev and Smirnov consists first in selecting  a particular function $\psi_0$ of the form \eqref{generalpsi}, with the choice $\beta=\beta(\gamma)$ \eqref{Bsigma}.   The corresponding  solution to the differential equation \eqref{BSsigDgx} for $g$, or, equivalently, to the hypergeometric equation \eqref{BSsigDg0x}, \eqref{hypergeom} for $g_0$, then involves a combination of two hypergeometric functions. While obtained in Ref. \cite{BS} for the exterior case $\sigma=+1$, it can be readily generalized to the interior whole-plane case with $\sigma=-1$, and is written for general $\sigma$ as:
\begin{eqnarray}\label{ghyper}
g(x)=\big(\frac{x}{4}\big)^{\gamma}g_0(x),\,\,\,g_0(x)= {}_{2}F_{1}\big(a,b,c,\frac{x}{4}\big) -C_0 \big(\frac{x}{4}\big)^{1/2-a-b} {}_2F_{1}\big(a',b',c',\frac{x}{4}\big)\end{eqnarray}
with
\begin{eqnarray}
&&a=a(\gamma):=\gamma-\gamma^\sigma_+,\,\,
b=b(\gamma):=\gamma-\gamma^\sigma_-,\,\, c=\frac{1}{2}+a+b, \label{abcBS}\\ 
&&a'=\frac{1}{2}-b,\,\,
b'=\frac{1}{2}-a,\,c'=\frac{1}{2}+a'+b'=\frac{3}{2}-a-b,\end{eqnarray}
where $\gamma^\sigma_{\pm}$ is defined in \eqref{gammasigma}. The constant
\begin{eqnarray}\label{C0} 
C_0=\frac{\Gamma(c)}{\Gamma(a)\Gamma(b)}\frac{\Gamma(a')\Gamma(b')}{\Gamma(c')}
\end{eqnarray}
is chosen such that $g_0(x)$ is singularity-free at $x=4$, i.e., at the point $z=-1$ on the unit circle (see pp. 590-591 in \cite{BS}). 
\end{rkk}
\begin{rkk}\label{singularity4}
In the action \eqref{genresquater} of the differential operator, one notices the existence of apparently singular terms involving $(4-x)^{-1}$. In fact, the choice of the constant $C_0$ \eqref{C0} in \eqref{ghyper}, made to insure that $g_0(x)$ is regular at $x=4$,  yields in turn the particular identity 
\begin{equation}\label{Asigmag0}
 \frac{g'_0(4)}{g_0(4)}=-\frac{ab}{2}=-\frac{1}{2}(\gamma-\gamma^\sigma_+)(\gamma-\gamma^\sigma_-)=\frac{1}{\kappa}A^\sigma(\gamma).
\end{equation}
This resolves the apparent singularities at $x=4$ in \eqref{ghyper}.
\end{rkk}
\begin{rkk} \label{remark46}$\bullet$ The  
 parameter $\gamma=\gamma_0$, hence $\beta=\beta_0:=\beta(\gamma_0)$, (corresponding to Eqs. (11)  and (12) in Ref. \cite{BS}) is chosen such that the leading singularities in the action  \eqref{genresbis} of the differential operator $\mathcal P(D)$ on the truncated $\psi$ function \eqref{psiphi} vanish: 
\begin{eqnarray}\label{BScondition1}
C(\gamma_0)=0,;\,\,\,\,\beta_0=\beta(\gamma_0)=\kappa\gamma^2_0/2. \end{eqnarray} 
When considering the action \eqref{genresquater} of the operator $\mathcal P(D)$ onto the full $\psi_0$ function \eqref{generalpsi} including $g_0(x)$, the leading singularity \eqref{Cxto0} vanishes.    
 Because  the functions  $C(\gamma)$  \eqref{AC} and $\beta(\gamma)$ \eqref{betaCgamma} are independent of $\sigma$,  the BS exponents $\gamma_0$ and $\beta_0$ stay the same for the interior problem.   
The solutions to \eqref{BScondition1} are 
\begin{eqnarray}
\label{gammazero}
\gamma_0^{\pm}(p)&:=&\frac{1}{2\kappa}\left(4+\kappa\pm\sqrt{(4+\kappa)^2-8\kappa p}\right),\\ \label{BS0pm}
\beta_0^{\pm}(p)&:=&\frac{1}{2}\kappa{\gamma_0^{\pm}}(p)^2,\\ \label{BS0}
\beta_0(p)&:=&\beta_0^-(p)=\frac{1}{2}\kappa\gamma_0^2,\,\,\, \gamma_0:=\gamma_0^-,\end{eqnarray} 
where the lower branch $\gamma_0:=\gamma_0^{-}$ is the one selected among the two solutions $\gamma_0^{\pm}$ (see Eq. (11) in Ref. \cite{BS} and Eqs. \eqref{gamma00} and \eqref{beta000} here). \\ 
\end{rkk}}
\textcolor{black}{The method of proof in Ref. \cite{BS}, that \eqref{BS0} yields the integral means spectrum of the whole-plane SLE in the exterior case, requires the BS solution \eqref{ghyper}, \eqref{gammazero}, \eqref{BS0} to the boundary equation \eqref{hypergeom} to be bounded and positive. When checking these conditions, one finds the following results for the two cases $\sigma=\pm 1$.}  
\begin{prop}\label{BScease}
\textcolor{black}{For $\sigma=+1$, the BS average integral means spectrum \eqref{gammazero}, \eqref{BS0},  while analytic up to $p=(4+\kappa)^2/8\kappa$,  holds only up to $p_0(\kappa):=3(4+\kappa)^2/32\kappa$, after which it stays linear \cite{BS}. 
For $\sigma=-1$, the BS spectrum only holds for $p\leq p^*(\kappa)$, where $p^*(\kappa)$ is given by \eqref{pstar}:
\begin{eqnarray}\label{pstarbis}p^*(\kappa)&=&
\frac{1}{16\kappa}\left((4+\kappa)^2-4-2\sqrt{4+2(4+\kappa)^2}\right),\end{eqnarray} and corresponds to the intersection of spectra \eqref{BS0} and \eqref{betapgammafin} $\beta_0(p^*)=\beta_+(p^*)$. Note that $\forall \kappa\geq 0, p^*(\kappa)<p_0(\kappa)$.} 
\end{prop}
\begin{proof} {Following \bf Lemma 5} in \cite{BS}, let us recall the values of the \textcolor{black}{parameters \eqref{abcBS} of the hypergeometric functions, in the case $\gamma=\gamma_0$:} 
\begin{eqnarray}\label{aprime}
a_0(p)&:=&a(\gamma_0)=\gamma_0(p)-\gamma^\sigma_{+}(p)=\gamma_0(p) -\frac{1}{\kappa}-\frac{1}{\kappa}\sqrt{1-2\sigma\kappa p}\\
\label{bprime}
b_0(p)&:=&b(\gamma_0)=\gamma_0(p)-\gamma^\sigma_{-}(p)=\gamma_0(p) -\frac{1}{\kappa}+\frac{1}{\kappa}\sqrt{1-2\sigma\kappa p},
\end{eqnarray} 
where $\gamma_0=\gamma_0^{-}$ is the lower BS parameter in  \eqref{gammazero}, and where $\gamma^\sigma_{\pm}$ is defined in \eqref{gammasigma}.

A first condition \cite{BS} for the existence of a \textcolor{black}{bounded BS solution $g_0(x)$ when $x\to 0$}, i.e., $z\to 1$ on the circle $\partial \mathbb D$, is the condition $1/2-a_0-b_0\geq 0$, which insures that the second term in \eqref{ghyper} is non-diverging, and gives $p\leq 3(4+\kappa)^2/32\kappa$, independently of $\sigma$. 

\textcolor{black}{Then a second condition \cite{BS} concerns the positivity of $g_0(x)$ on the interval $x\in [0,4]$, which is shown to amount to $g_0(4) >0$}, or explicitly  $$\Gamma(1/2-a_0)\Gamma(1/2-b_0)>0.$$ For $\sigma =+1$ and for $0\leq p\leq 1/2\kappa$, BS show that   $1/2-a_0 \geq 1/2-b_0> 0$, whereas for $p\geq 1/2\kappa$, the inequality is fulfilled since $a_0$ and $b_0$ are complex conjugate, thus $\Gamma(1/2-b_0)=\overline{\Gamma(1/2-a_0)}.$

For $\sigma =-1$,  the situation turns out to be different. The parameters $a_0$ \eqref{aprime} and $b_0$ \eqref{bprime} are real for $p\geq 0$, but the inequality $1/2-b_0> 0$ is no longer necessarily satisfied. One has indeed from \eqref{gammafin}
\begin{eqnarray}\nonumber \frac{1}{2}-b_0(p)&=&\frac{1}{2}+\frac{1}{\kappa} -\frac{1}{\kappa}\sqrt{1+2\kappa p}-\gamma_0(p)\\ \label{1/2-b}
&=&\frac{1}{2}+\frac{2}{\kappa} -\gamma_+(p)-\gamma_0(p).\end{eqnarray} 
Since $\gamma_+$ and $\gamma_0$ are both increasing functions of $p$, there may be a point where $1/2-b_0=0$, after which it becomes negative and the  \textcolor{black}{positive BS solution $g_0$} \eqref{ghyper}, \eqref{C0}  ceases to exist for $\sigma =-1$. Recall that the four pairs $({\gamma_0}^{\pm},{\beta_0}^{\pm})$ \eqref{gammazero}-\eqref{BS0pm}, and $(\gamma_{\pm},\beta_{\pm})$ \eqref{gammasigma}-\eqref{betapgammasigma} all belong to the curve $(\gamma,\beta(\gamma))$ \eqref{betaCgamma} (Fig. \ref{BSfig1}).  
\textcolor{black}{In these notations,  the transition point $p^*(\kappa)$ \eqref{pstar} is defined by the intersection of the two spectra $\beta_0(p^*)=\beta_+(p^*)$. Thus the corresponding parameters $\gamma_0=\gamma_0^-(p^*)$ and $\gamma_+=\gamma_+(p^*)$ are such that $\beta(\gamma_0)=\beta_0(p^*)=\beta_+(p^*)=\beta(\gamma_+)$ (see Fig. \ref{BSfig1}, top figure). Because  $\beta(\gamma)$ is the quadratic form \eqref{betaCgamma},   
    $\gamma_0+\gamma_+=2/\kappa+1/2$.  
Thus $p^*$ is precisely the point where $1/2-b_0$ \eqref{1/2-b} vanishes. For $p> p^*$, $1/2-b_0<0$ and a positive BS solution $g_0$ to the boundary equation \eqref{hypergeom} no longer exists.}  
\end{proof}
\textcolor{black}{\begin{rkk}
In the original Beliaev-Smirnov case $\sigma=+1$, one has 
\begin{eqnarray}\nonumber \frac{1}{2}-b_0(p)&=&\frac{1}{2}+\frac{1}{\kappa} -\frac{1}{\kappa}\sqrt{1-2\kappa p}-\gamma_0(p)\\ \label{1/2-b+}
&=&\frac{1}{2}+\frac{2}{\kappa} -\gamma_+^{(+1)}(p)-\gamma_0(p).
\end{eqnarray} In the negative range of moments, there exists a value of $p$ where $1/2-b_0(p)$ \eqref{1/2-b+} vanishes,  $p^{**}(\kappa)=-(4+\kappa)^2(8+\kappa)/128$, and below which $1/2-b_0(p)$ is negative. This signals the possible onset of a phase transition, similar to the one  studied here in the $\sigma=-1$ case and occurring at $p^*(\kappa)$. This will be further studied in a separate publication with Dmitry Beliaev \cite{BDZ}. It has also been noticed in Ref. \cite{2013JSMTE..04..007L}.  
\end{rkk}}
\begin{figure}
\begin{center}
\includegraphics[angle=0,width=.53029\linewidth]{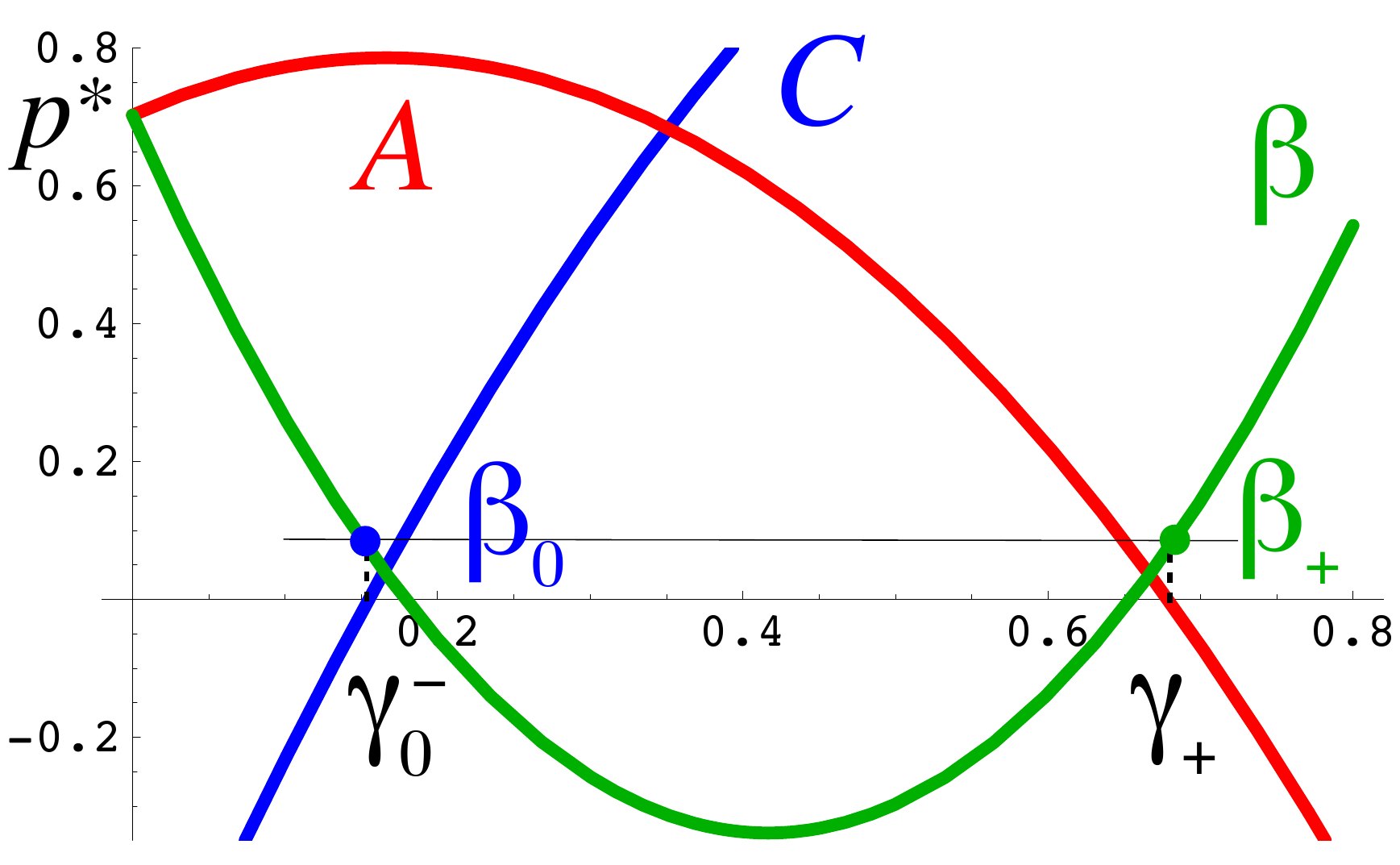}
\vskip.7cm
\includegraphics[angle=0,width=.53029\linewidth]{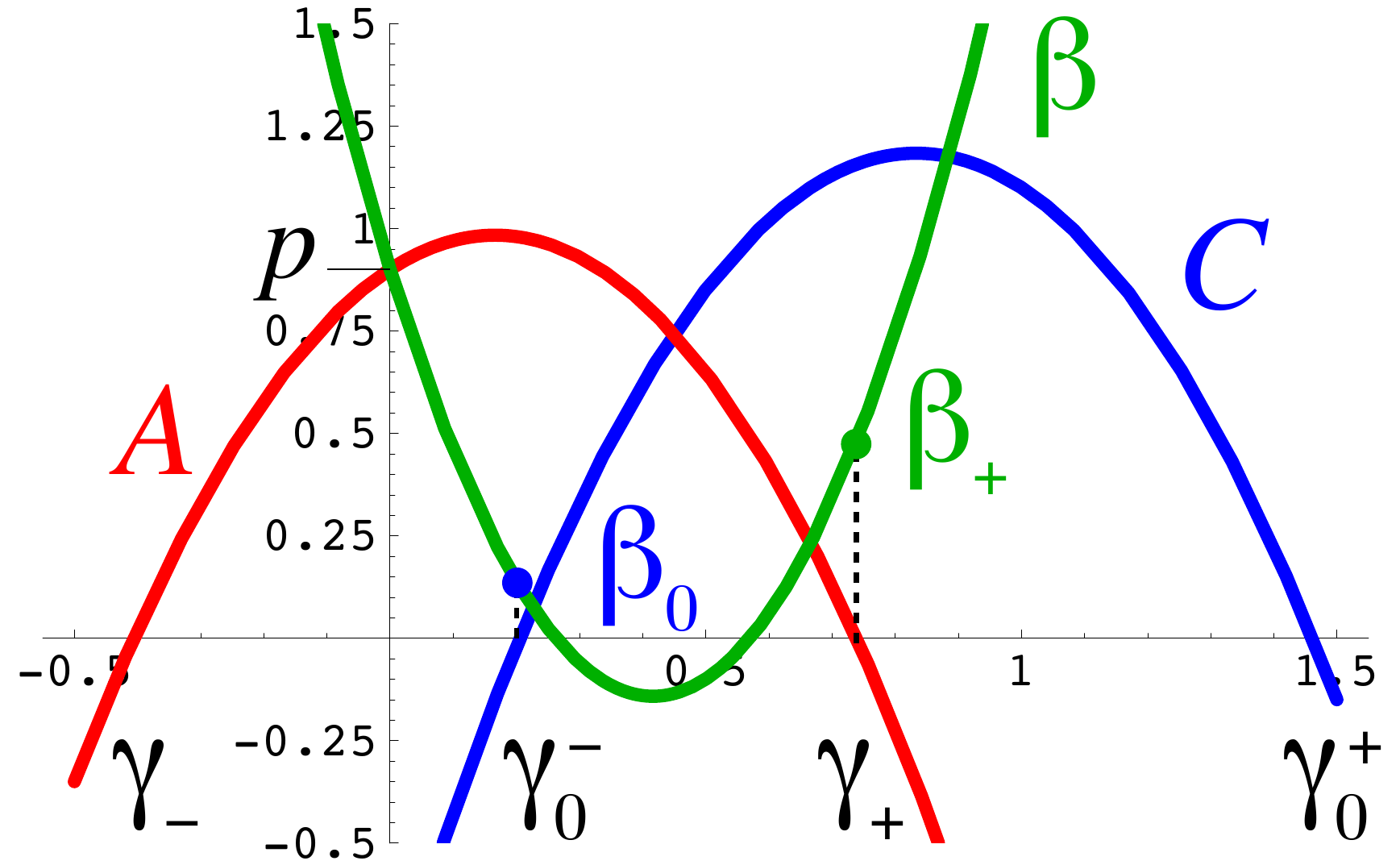}
\vskip.7cm
\includegraphics[angle=0,width=.53029\linewidth]{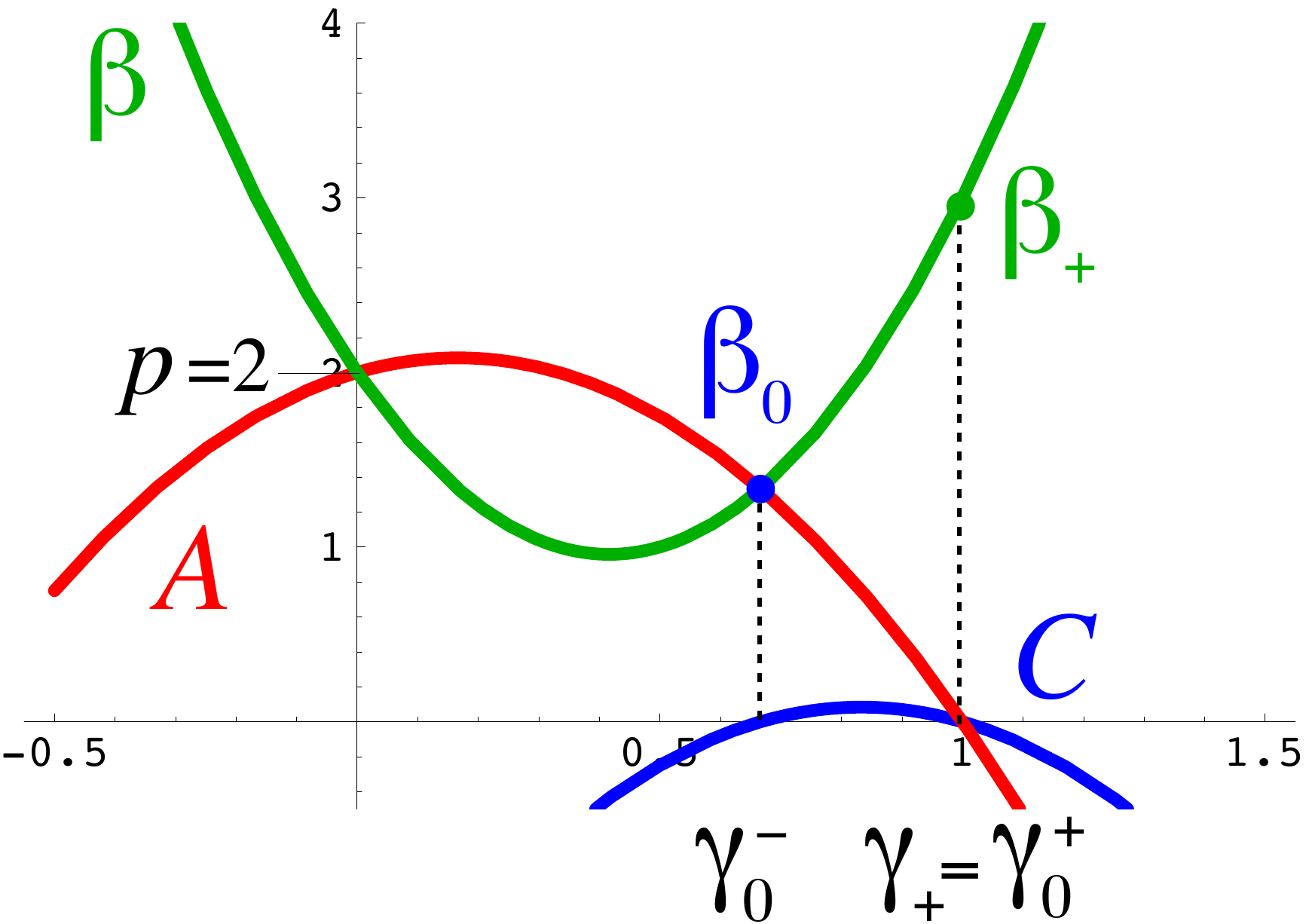}
\caption{\textcolor{black}{Curves $A(\gamma)$, $\beta(\gamma)=\kappa\gamma^2/2-C(\gamma)$, and  $C(\gamma)$, for the value $\kappa=6$ of the SLE parameter. They are displayed  for three different values of the moment order $p$: respectively for the critical value $p^*(\kappa=6)$ where the two spectra $\beta_0:=\beta(\gamma_0^-)$ and $\beta_+:=\beta(\gamma_+)$ co\"{\i}ncide; for a generic $p=0.9\in \big(p^*(6),p(6)\big)$, for which $\beta(\gamma_0^-)<\beta(\gamma_+)$; and  for the special point $p(\kappa=6)=2$, where $\gamma_+=\gamma_0^+=1$ and $\beta(\gamma_+)=3$.}}
\label{BSfig1}
\end{center}
\end{figure}

 Recall now that the special point \eqref{kpgen} $(p(\kappa), \alpha(\kappa)>0)$ of Theorem \ref{main1} was obtained as obeying both conditions \eqref{Aalpha} $A(\alpha)=0$ and \eqref{Calpha} $C(\alpha)=0$, together with $\beta=\kappa \alpha^2/2$ (see Fig. \ref{BSfig1}, bottom figure, where $\alpha=\gamma_+=\gamma_0^+$). This leads to the following remark. \begin{rkk}The special point $(p(\kappa), \alpha(\kappa))$ \eqref{kpgen} is such that  
\begin{eqnarray}\label{+++}\beta_+=\beta_0^+=\kappa \alpha^2/2,\,\,\, \alpha=\gamma_+=\gamma_0^+,
\end{eqnarray}
and lies  at the intersection of the curve $\gamma_+(p)$ \eqref{gammafin} and of the complementary BS curve $\gamma_0^+(p)$ \eqref{gammazero}.This is illustrated in Fig. \ref{BSfig} further below.
 \end{rkk} 

\textcolor{black}{$\bullet$ We therefore conclude that Beliaev and Smirnov's method of proof  works up to $p^*(\kappa)$ in the interior case $\sigma=-1$, in such a way that the integral means spectrum is given by $\beta_0=\beta(\gamma_0)$, with $\gamma_0=\gamma_0^-$ such that $C(\gamma_0^-)=0$ (Fig. \ref{BSfig1}). Above the transition point $p^*(\kappa)$, we will argue in the next sections that the integral means spectrum is  given by  $ \beta_+=\beta(\gamma_+)$ \eqref{betapgammafin}, where $\gamma_+$ \eqref{gammafin} satisfies $A(\gamma_+)=0$ (see Fig. \ref{BSfig1}, middle figure).}  

\textcolor{black}{To study the integral means spectrum of the inner whole-plane SLE  above the transition point $p^*(\kappa)$, we shall use as a first step in the next Section \ref{firstsol} the truncated function \eqref{psiphi}, $\psi(z,\bar z)=(1-z\bar z)^{-\beta}x^\gamma$, where  $\gamma$  and $\beta$ belong to the curve $(\gamma,\beta(\gamma))$, as given by the relation \eqref{betapgamma} (see Fig. \ref{BSfig1}, middle figure). The action of the differential operator $\mathcal P(D)$ on this function $\psi$ is given by  equation \eqref{genresbis}, which can be written as
\begin{equation}\label{PpsiA}\mathcal P(D)\psi(z,\bar z)=\psi(z,\bar z) x^{-1}\left[C(\gamma)\frac{1-z\bar z}{x}-A(\gamma)\right] (1-z\bar z-x).\end{equation}
For the BS parameter $\gamma=\gamma_0^-$,  the first term inside the brackets  (the most singular term for $x\to 0$, i.e., $z\to 1$)  vanishes, while in our case, $\gamma=\gamma_+$,  the second term $A$ vanishes.} 

\subsubsection{\textcolor{black}{Beyond the transition point: $p\geq p^*(\kappa)$.}} \label{firstsol} 
In that range, we now look for the interior case $\sigma=-1$ at the properties of the function $\psi(z,\bar z)=(1-z\bar z)^{-\beta}x^\gamma$ \eqref{psiphi}, where  $\gamma$  and $\beta$ are assumed throughout this section to satisfy the relation \eqref{betapgamma}  $\beta=\beta(\gamma)$, and $\gamma$ is such that $A(\gamma)=A^{(-1)}(\gamma)=0$ \eqref{gamma}; they are given by the pair of solutions $\gamma_{\pm}=\gamma_{\pm}(p)$ \eqref{gammafin} and $\beta_{\pm}=\beta_{\pm}(p)$ \eqref{betapgammafin}. Eq. \eqref{PpsiA} then yields the explicit result:
\begin{eqnarray}
\mathcal P(D)[\psi(z,\bar z)]
\label{Ppsi}
=\psi(z,\bar z)\left(\frac{\kappa}{2}\gamma^2-\beta\right) (1-z\bar z)(1-z\bar z-x)x^{-2},\end{eqnarray}
where we recall that $\frac{\kappa}{2}\gamma^2-\beta(\gamma)=C(\gamma)$. 

The quantity in factor of $\psi(z,\bar z)$ in \eqref{Ppsi} vanishes both on the unit circle $\partial \mathbb D$ and on the circle $\partial \mathbb D_{1/2}:=\{z: 1-z\bar z -x=0\}$, centered at $(1/2,0)$ and of  radius $1/2$, which passes through $z=0$ and $z=1$, and is tangent to the unit circle $\partial \mathbb D$ at $z=1$ (see Fig. \ref{figlune}). The overall sign of \eqref{Ppsi} crucially depends on the position of $z$ with respect to the circle $\partial \mathbb D_{1/2}$:  For $z$ inside the disk \begin{equation}\label{D1/2}\mathbb D_{1/2}:=\{z: 1-z\bar z -x>0\},\end{equation} \eqref{Ppsi}  has the same sign as the coefficient $\kappa\gamma^2/2-\beta$, and the opposite sign when $z$ lies outside of that disk.  
\begin{figure}[htb]\begin{center}
\includegraphics[angle=0,width=.53290\linewidth]{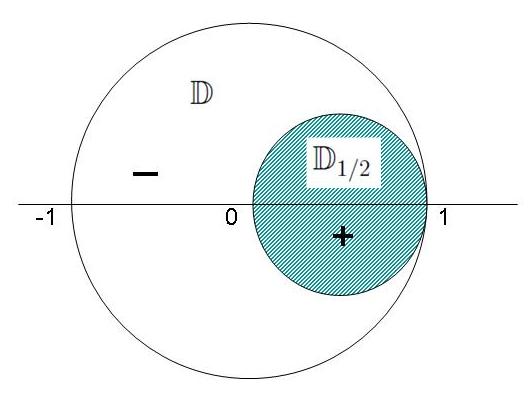}
\caption{{\it The unit disk $\mathbb D$ and the disk $\mathbb D_{1/2}$ \eqref{D1/2}. The signs indicated are those \eqref{signpsi+_} of $\mathcal P(D)[\psi_+]$ for $p\leq p(\kappa)$, which vanishes on $\partial\mathbb D$ and $\partial\mathbb D_{1/2}$.}}
\label{figlune}\end{center}
\end{figure}
The sign of the coefficient $\kappa\gamma^2/2-\beta$ itself depends on which branch is chosen in \eqref{gammafin} and \eqref{betapgammafin}. One easily finds that 
$$C(\gamma_{\pm}(p))=\frac{\kappa}{2}\gamma^2_{\pm}(p)-\beta_{\pm}(p)=\left(\frac{1}{\kappa}+\frac{1}{2}\right)\left(1\pm \sqrt{1+2\kappa p}\right)-2p.$$
For the negative branch, and for $p\geq 0$, it is clear that 
\begin{equation}\kappa\gamma_{-}^2/2-\beta_{-}\leq 0.\end{equation} The positive branch $\frac{\kappa}{2}\gamma_+^2-\beta_+$, on the other hand, has a zero for   $p=p(\kappa)=(6+\kappa)(2+\kappa)/8\kappa$ with  $\gamma_+=\alpha=(6+\kappa)/2\kappa$, which naturally corresponds to the \textit{special set of values} \eqref{kpgen} where there exits the exact solution \eqref{Fzz} with  $\beta_+=\beta=\kappa \alpha^2/2$. One therefore  has 
\begin{eqnarray}\nonumber\kappa\gamma_{+}^2/2-\beta_{+} > 0,&& p < p(\kappa),\\ \nonumber \kappa\gamma_{+}^2/2-\beta_{+}= 0,&& p = p(\kappa),\\ \nonumber \kappa\gamma_{+}^2/2-\beta_{+}< 0,&& p> p(\kappa).
\end{eqnarray}
We therefore arrive at the various domain inequalities for $p\neq p(\kappa)$, with the obvious notation $\psi_{\pm}(z,\bar z):=(1-z\bar z)^{-\beta_{\pm}}x^{\gamma_{\pm}}$, 
\begin{eqnarray}\label{signpsi-}
\mathcal P(D)[\psi_-(z,\bar z)] < 0, z\in \mathbb D_{1/2}&,&
\mathcal P(D)[\psi_-] > 0, z\in \mathbb D\setminus  \overline{\mathbb D}_{1/2} \label{Ppsi-};\\ \label{signpsi+_}
p\leq p(\kappa):\mathcal P(D)[\psi_+(z,\bar z)]> 0, z\in \mathbb D_{1/2}&,&
\mathcal P(D)[\psi_+]< 0, z\in \mathbb D\setminus\overline{\mathbb D}_{1/2} \label{Ppsi+};\\ \label{signpsi++}
p\geq p(\kappa):\mathcal P(D)[\psi_+(z,\bar z)]< 0, z\in \mathbb D_{1/2}&,&
\mathcal P(D)[\psi_+] > 0, z\in \mathbb D\setminus\overline{\mathbb D}_{1/2};\label{Ppsi++}\end{eqnarray}
 the resulting ratio $\mathcal P(D)[\psi_{\pm}]/\psi_{\pm}$ vanishes at the boundaries of the above domains, i.e., on $\partial \mathbb D$ and $\partial \mathbb D_{1/2}$. For $p=p(\kappa)$, $\psi_+$ is an exact solution such that  $\mathcal P(D)[\psi_+] =0$.\\ 

\textit{Logarithmic modification.} Following Ref. \cite{BS}, let us consider now the action of the differential operator on the modified function $\psi(z,\bar z)\ell_\delta(z\bar z)$, where here $\psi:=\psi_{\pm}$ and where the factor \begin{equation}\label{ell}\ell_\delta(z\bar z):= [-\log (1-z\bar z)]^\delta,\,\, \delta \in \mathbb R ,\end{equation}  brings in a (soft) \emph{logarithmic singularity}. From Eq. \eqref{PD} one finds the simple result
\begin{eqnarray}\nonumber \mathcal P(D)[\psi(z,\bar z)\ell_\delta(z\bar z)]&=&\ell_\delta(z\bar z)\mathcal P(D)[\psi(z,\bar z)]-\psi(z,\bar z)2z\bar z(1-z\bar z)x^{-1}\ell_\delta^{\,\prime}(z\bar z) \\ \label{Ppsiell}
&=&\ell_\delta(z\bar z)\left\{\mathcal P(D)[\psi(z,\bar z)]-\psi(z,\bar z) \frac{2\delta z\bar z x^{-1}}{[-\log(1-z\bar z)]}\right\},\end{eqnarray}
where the derivative $\ell_\delta^{\,\prime}(z\bar z)$ is taken with respect to $z\bar z$. Using Eq. \eqref{Ppsi}  yields (here $\gamma:=\gamma_{\pm}$ and $\beta:=\beta_{\pm}$): 
\begin{eqnarray}\label{Ppsi0} \mathcal P(D)[\psi(z,\bar z)\ell_\delta(z\bar z)]&&=\ell_\delta(z\bar z)\psi(z,\bar z)x^{-1}\\ \nonumber &&\times \left[\left(\frac{\kappa}{2}\gamma^2-\beta\right) (1-z\bar z)(1-z\bar z-x)x^{-1} -\frac{2z\bar z \delta}{[-\log(1-z\bar z)]}\right].\end{eqnarray}\\
$\bullet$ Consider first the domain $\mathbb D_{1/2}$ (Fig. \ref{figlune}). 
  The sign of $\mathcal P(D)[\psi_{\pm}(z,\bar z)]$ for $z\in \mathbb D_{1/2}$ is given in the three different cases by the first column of Eqs. \eqref{Ppsi-}, \eqref{Ppsi+} and \eqref{Ppsi++}. In each case, this sign is also that of the first term of \eqref{Ppsi0} in the same domain and for the same case. For each case, \emph{choose the sign of $\delta$} so that the second term in \eqref{Ppsi0} has the \emph{same uniform sign} as the first term in $\mathbb D_{1/2}$. Then $\mathcal P(D)[\psi_{\pm}(z,\bar z)]$ and $\mathcal P(D)[\psi_{\pm}(z,\bar z)\ell_\delta(z\bar z)]$ have the same sign in $\mathbb D_{1/2}$. \\ 
  
\noindent  $\bullet$ Consider now the complementary domain $\mathbb D\setminus \mathbb D_{1/2}$. Take $z\in \mathbb D\setminus \mathbb D_{1/2}$ on the circle $\partial \mathbb D(r)$ of radius $r<1$ centered at the  origin (with $\partial \mathbb D(1)=\partial \mathbb D$), so that $1-z\bar z=1-r^2$. The quantity $(1-z\bar z-x) x^{-1}$ in \eqref{Ppsi0} is negative in $\mathbb D\setminus \mathbb D_{1/2}$ and equals  $(1-r^2)x^{-1}-1$   on $\partial \mathbb D(r)$,  
 for which $-1< (1-r^2)x^{-1}-1 \leq 0.$ 
 On the circle $\partial \mathbb D(r)$ of radius $r<1$ and \emph{outside of} $\mathbb D_{1/2}$, one thus has the following bounds for the first term of  \eqref{Ppsi0}:
 $$ -(1-r^2) <(1-z\bar z)(1-z\bar z-x) x^{-1} \leq 0.$$
Therefore, for $r$ close enough to $1$, say $r_1 <r<1$, the second term in \eqref{Ppsi0}, which vanishes  logarithmically when $r\to 1^-$, dominates the first term, which is of order $O(1-r^2)$, hence determines the overall sign of  $\mathcal P(D)[\psi_{\pm}(z,\bar z)\ell_\delta(z\bar z)]$ for $z\in \mathbb D\setminus \mathbb D_{1/2}$.\\

\noindent $\bullet$  Recall now that $\delta$  has been chosen precisely such that the sign of the second term in \eqref{Ppsi0} is that of $\mathcal P(D)[\psi_{\pm}\ell_\delta]$ and $\mathcal P(D)[\psi_{\pm}]$  in $\mathbb D_{1/2}$. We thus conclude that 
in the \emph{whole annulus} $r_1<r<1$, the sign of $\mathcal P(D)[\psi_{\pm}(z,\bar z)\ell_\delta(z\bar z)]$ is uniform  and given by that of $\mathcal P(D)[\psi_{\pm}(z,\bar z)]$ for $z\in \mathbb D_{1/2}$, as given for the three canonical cases  $[\psi_-], [\psi_+, p\leq p(\kappa)], [\psi_+, p\geq p(\kappa)]$ by the first column in Eqs. \eqref{Ppsi-}, \eqref{Ppsi+} and \eqref{Ppsi++}.\\

We therefore  conclude that there exist for these three cases (denoted by $i=1,2,3$) three open annuli  $\mathbb A(r_i):=\{z: r_i<|z|<1\}=\mathbb D\setminus \overline{\mathbb D}(r_i)$  whose boundary includes $\partial \mathbb D$, where one has respectively (for a specific sign of $\delta$ chosen in each case as described above): \\
 $\bullet$ $\mathcal P(D)[\psi_-(z,\bar z)\ell_\delta(z\bar z)] < 0, z\in \mathbb A(r_1)$, so that $\psi_-(z,\bar z)\ell_\delta(z\bar z)$ is locally a \textit{subsolution} to the equation $\mathcal P(D)[F(z,\bar z)]=0$;\\
 $\bullet$ for $p < p(\kappa)$,  $\psi_+(z,\bar z)\ell_\delta(z\bar z)$ is a \textit{supersolution} with $\mathcal P(D)[\psi_+(z,\bar z)\ell_\delta(z\bar z)] > 0$ for $z\in \mathbb A(r_2)$; \\$\bullet$  for  $p> p(\kappa)$, $\psi_+(z,\bar z)\ell_\delta(z\bar z)$ is a \textit{subsolution} with $\mathcal P(D)[\psi_+(z,\bar z)\ell_\delta(z\bar z)] < 0$ for $z\in \mathbb A(r_3)$;\\ 
  $\bullet$ for  $p=p(\kappa)$,  $\mathcal P(D)[\psi_+(z,\bar z)] = 0, z\in \mathbb D$, so that $\psi_+(z,\bar z)=F(z,\bar z)=(1-z\bar z)^{-\beta_+}|1-z|^{2\gamma_+}$ is the exact solution in Theorem \ref{main1} with parameters \eqref{kpgen}: $\gamma_+=\alpha(\kappa)$ and $\beta_+=\kappa \alpha^2/2$.\\
  
  We then follow the same method as in Refs. \cite{BKTH,BS}. The operator $\mathcal P(D)$, when written in polar coordinates as in \eqref{BSsig},  is \emph{parabolic}, where $\theta$ corresponds to the spatial variable, and $r$ to the time variable \cite{evans}.  In the above, the functions $\psi_{\pm}(z,\bar z)\ell_\delta(z\bar z)$ are positive functions bounded on the respective circles of radius $r_i$, as $F(z,\bar z)=\mathbb E\big[|f'_0(z)|^p\big]$ is. One can thus find positive constants $c_i$ such that $$F<c_1\,\psi_-\,\ell_\delta,\, r=r_1;\,\,\,\, c_2\,\psi_+\,\ell_\delta < F,\,r=r_2,\, p<p(\kappa);\,\,\,\, F< c_3\,\psi_+\,\ell_\delta,\,r=r_3,\,p>p(\kappa).$$ Using then in each of the corresponding annuli where $\mathcal P(D)[\psi_{\pm}\ell_\delta]$ has a definite sign, respectively, the maximum principle, the minimum principle, and  the maximum principle  (\cite{evans}, Th. 7.1.9), yields the
  \begin{prop} \label{annineq}
  \begin{eqnarray}\label{F-}&&F<c_1\,\psi_-\,\ell_\delta,\,\,\, z\in \mathbb A(r_1),\,\,\, \forall p,\\ \label{F+} &&c_2\,\psi_+\,\ell_\delta < F,\,\,\,z\in \mathbb A(r_2),\,\,\, p<p(\kappa),\\ \label{F++}&&F< c_3\,\psi_+\,\ell_\delta,\,\,\,z\in \mathbb A(r_3),\,\,\,p>p(\kappa).\end{eqnarray} 
  \end{prop} 
These inequalities will be used in the following section  to establish  the existence at $p^*(\kappa)$ of a phase  transition in the integral means spectrum of the inner whole-plane SLE  and to prove Theorem \ref{theoMFc}. 

 \subsubsection{Proof of Theorem \ref{theoMFc}}\label{singan}
\begin{proof}The average  integral means spectrum of the whole-plane SLE is given  by the asymptotic behavior  for $r\to 1^-$ of the $F$ integral: 
\begin{equation}\label{intmeanF} \int_0^{2\pi} F(r,\theta)d\theta = \int_0^{2\pi}  \mathbb E (|f'_0(r,\theta)|^p)d\theta \stackrel{(r\to 1^-)}{\asymp} (1-r)^{-\beta(p)}.
\end{equation}
For the $\psi$ function defined in Eqs. \eqref{psior}-\eqref{psiphi}, the integral means are: 
\begin{equation}\label{intmean} \int_0^{2\pi} \psi(r,\theta)d\theta =(1-r^2)^{-\beta} \int_0^{2\pi} g(r,\theta)d\theta
,\end{equation} 
where we write $g(r,\theta)=g[D(r,\theta)]=D^\gamma(r,\theta)=|1-z|^{2\gamma}$. This function $g$ has a singularity at $z=1$, i.e., for $r=1,\theta=0$. Near that point, its argument $D(r,\theta)=r^2-2r\cos \theta+1$ 
 is equivalent to $D(r,\theta)\sim (1-r)^2+\theta^2$. 
 
The exponents $\beta$ and $\gamma$ in the above are given by Eqs. \eqref{betapgammafin} and \eqref{gammafin}. For the $(+)$ branch, $\gamma_+> 0$ and the  integral 
$\int_0^{2\pi} g_{+}(1,\theta) d\theta$ is integrable at $\theta=0$, which is a zero of $g$. 
 For the other branch, $\gamma_{-}$ is negative for $p\geq 0$, and  the singularity along the unit circle at $\theta =0$ is no longer integrable when $2\gamma_{-}+1\leq 0$. This  corresponds to a cross-over value $p=\tilde p(\kappa):=(4+\kappa)/8$.
One therefore has: 
 \begin{eqnarray} \label{g+}
 \int_0^{2\pi} \psi_{+}(r,\theta)d\theta&\stackrel{(r\to 1^-)}{\asymp}& (1-r)^{-\beta_{+}},\\ \label{g-} 
  \int_0^{2\pi} \psi_{-}(r,\theta)d\theta&\stackrel{(r\to 1^-)}{\asymp}& \begin{cases}(1-r)^{-\beta_{-}},\,\,\,\,p\leq \frac{4+\kappa}{8},\\ \label{+g--} 
  (1-r)^{-\beta_{-}+2\gamma_{-}+1},\,\,\,\, p\geq\frac{4+\kappa}{8}.\end{cases}\end{eqnarray}
 
Consider now the modified functions 
 $\psi_{\pm}\,\ell_\delta$, where $\ell_\delta$ is the weakly diverging or vanishing logarithmic function \eqref{ell}. The asymptotic power law behaviors of their integral means  near the unit circle are obviously  the same as for $\psi_{\pm}$:
\begin{eqnarray} \label{g+ell}
 \int_0^{2\pi} \psi_{+}(r,\theta)\ell_\delta(r^2)d\theta&\stackrel{(r\to 1^-)}{\asymp}& (1-r)^{-\beta_{+}},\\ \label{g-ell} 
  \int_0^{2\pi} \psi_{-}(r,\theta)\ell_\delta(r^2)d\theta&\stackrel{(r\to 1^-)}{\asymp}& \begin{cases}(1-r)^{-\beta_{-}},\,\,\,\,p\leq \frac{4+\kappa}{8},\\  
  (1-r)^{-\beta_{-}+2\gamma_{-}+1},\,\,\,\, p\geq\frac{4+\kappa}{8}.\end{cases}\end{eqnarray}
 By plugging   the asymptotic behaviors \eqref{intmeanF} and  \eqref{g+ell} into inequalities  \eqref{F+} and \eqref{F++} of Proposition \ref{annineq}, we obtain  \begin{eqnarray}
 \label{+}
&&{\beta_{+}}(p)\leq \beta(p),\,\,\, p<p(\kappa)=\frac{(6+\kappa)(2+\kappa)}{8\kappa},\\ \label{+0} &&{\beta(p)}= \beta_+(p),\,\,\, p=p(\kappa),\\ \label{++}
&&{\beta(p)}\leq \beta_+(p),\,\,\, p>p(\kappa)=\frac{(6+\kappa)(2+\kappa)}{8\kappa}.
 \end{eqnarray}
This ends the proof of Theorem \ref{theoMFc} for the $m=1$ case. \end{proof}

By plugging  the asymptotic behaviors \eqref{intmeanF} and \eqref{g-ell} into inequality \eqref{F-} of Proposition \ref{annineq}, we obtain
\begin{eqnarray}
\label{-} 
\beta(p)\leq \begin{cases}{\beta_{-}}(p),\,\,\,p\leq \frac{4+\kappa}{8},\\ 
 {\beta_{-}(p)-2\gamma_{-}(p)-1},\,\,\,p\geq\frac{4+\kappa}{8}.\end{cases}
\end{eqnarray}
At this point, we can invoke the bound implied by the existence of an universal spectrum $B_{\mathcal S}$ for the $\mathcal S$  schlicht class of univalent functions $f$ in the unit disk.   As we have seen in Section \ref{universal},  the integral means spectrum $\beta_f(p)$ of such an $f\in \mathcal S$ is bounded,  for $p\geq 2/5$, as:
$$\beta_f(p)\leq B_{\mathcal S}(p)=3p-1,\,\,\,  p\geq 2/5,$$
the maximum being attained for the Koebe function \eqref{koebe}. For $p> 0$, we have from the definitions \eqref{betapgammafin} and \eqref{gammafin}
\begin{eqnarray}
\beta_{+}< 3p-1 < \beta_{-},\,\,\, p\leq \frac{4+\kappa}{8},\\
\beta_{+}< 3p-1 < \beta_{-}\leq \beta_{-}-2\gamma_{-}-1,\,\,\, p\geq \frac{4+\kappa}{8}.
\end{eqnarray}
This therefore excludes the $(-)$ branch as a possible integral means spectrum.

This strongly suggests that the average integral means spectrum of the unbounded inner whole-plane SLE$_\kappa$ is simply given, for $p\geq 0$, by 
\begin{equation}\label{betawhopl} \beta(p,\kappa)=\max\left\{\beta_{0}(p,\kappa),3p-\frac{1}{2}-\frac{1}{2}\sqrt{1+2\kappa p}\right \},
\end{equation}
where $\beta_0(p,\kappa)$ is given by Eq. \eqref{beta000}.  If one extends the range of moment orders to  $p\leq 0$, $\beta_0$ is  replaced by $\beta_{tip}$ \eqref{beta0tip} for $p$ sufficiently negative. The phase transition in Eq. \eqref{betawhopl} occurs at the critical point $p^*(\kappa)$, as given by Eq. \eqref{pstar}, where the \textcolor{black}{BS spectrum ceases to hold} (see Proposition \eqref{BScease}). These conclusions are illustrated in Fig. \ref{BSfig}.
\begin{figure}
\begin{center}
\includegraphics[angle=0,width=.8329\linewidth]{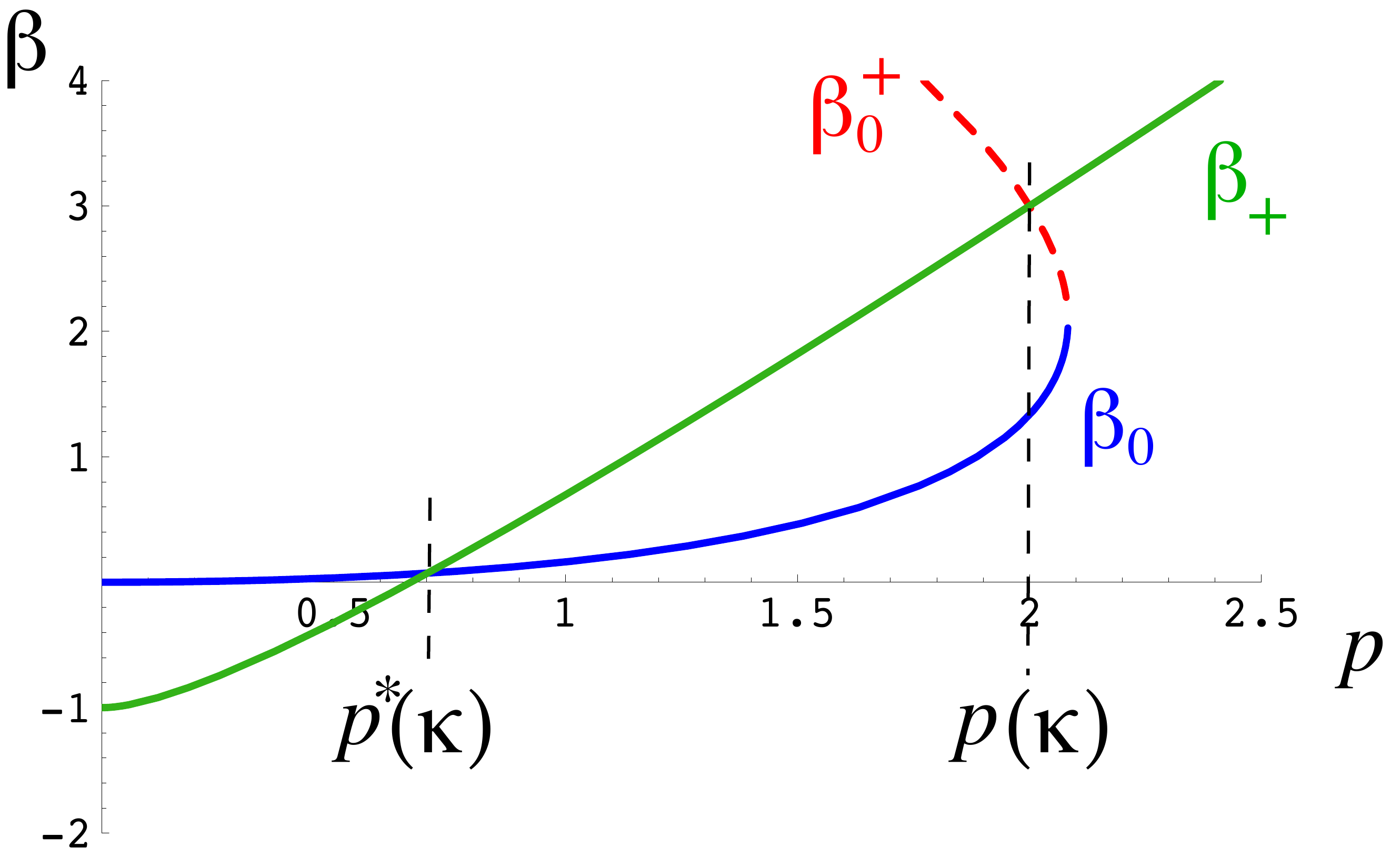}
\caption{{\it Various integral means spectra for the interior whole-plane SLE$_\kappa$ (here for $\kappa=6$): Standard BS branch of the SLE$_{\kappa=6}$ bulk average means spectrum $\beta_0(p,\kappa)=\beta_0^-(p,\kappa)$ in Eqs. \eqref{beta000} and \eqref{BS0} \textcolor{blue}{(in blue)};  second ``non-physical'' BS branch $\beta_0^+(p,\kappa)$ in Eq. \eqref{BS0pm} \textcolor{red}{(in red)};  whole-plane multifractal function $\beta_+(p,\kappa)=3p-1/2-(1/2)\sqrt{1+2\kappa p}$ in Eq. \ref{betapgammafin} \textcolor{green}{(in green)}. The resulting average integral means spectrum $ \beta(p,\kappa)$ \eqref{betawhopl} of whole-plane SLE$_\kappa$ undergoes a phase transition at $p=p^*(\kappa)$ \eqref{pstar}, where the first two curves intersect, $\beta_0(p^*,\kappa)=\beta_+(p^*,\kappa)$, such that $\beta(p,\kappa)=\beta_0(p,\kappa), \forall p\in [0,p^*(\kappa)]$, and $\beta(p,\kappa)=\beta_+(p,\kappa), \forall p\in [p^*(\kappa),\infty)$. Note that the special point $p=p(\kappa)$ \eqref{kpgen} (here $p(6)=2$) lies at the intersection \eqref{+++} of the curve $\beta_+(p,\kappa)$ with the ``non-physical'' BS branch $\beta_0^+(p,\kappa)$.
}}
\label{BSfig}
\end{center}
\end{figure}
 \newpage  
\subsubsection{Proof of Theorem \ref{theoMFrigorous}}\label{secondsolrig}
\textcolor{black}{In this section, we provide a proof of Theorem  \ref{theoMFrigorous}. It is based on an extended use of a \textit{duality} property of the boundary solution  \eqref{ghyper}, which we specialize here to the interior whole-plane case $\sigma=-1$. A similar study can be made for the exterior BS case $\sigma=+1$ \cite{BDZ}.    Define the dual parameters $(\gamma,\gamma')$ such that 
  \begin{eqnarray}\label{bb'}
  \beta(\gamma)&=&\beta(\gamma'),\\ \label{gg'}
   \gamma+\gamma'&=&\frac{2}{\kappa}+\frac{1}{2},  
   \end{eqnarray}
  where $\beta(\gamma)$ is defined in \eqref{betapgamma} (see Fig. \ref{BSfigter}). By using this duality, we make the following observation:  \begin{rkk}\label{BSCA} $\bullet$  The solution to the differential equation \eqref{BSsigDgx} for $g$, or \eqref{hypergeom} for $g_0$, which involves the combination of two hypergeometric functions \eqref{ghyper}, can be written in a dual manner as:
\begin{eqnarray}\label{ghyperbis}
g(x)&=&\big(\frac{x}{4}\big)^{\gamma} {}_{2}F_{1}\big(a,b,c,\frac{x}{4}\big) -C_0 \big(\frac{x}{4}\big)^{\gamma'} {}_2F_{1}\big(a',b',c',\frac{x}{4}\big)\\ \label{g0ter}
&=:&\big(\frac{x}{4}\big)^{\gamma}g_0(x)\\ \label{g0tilde} 
&=:&\big(\frac{x}{4}\big)^{\gamma'}\tilde g_0(x),
\end{eqnarray}
with
\begin{eqnarray}
&&a=a(\gamma):=\gamma-\gamma_+,\,\,
b=b(\gamma):=\gamma-\gamma_-,\,\, c=\frac{1}{2}+a+b, \label{abcBSbis}\\ \label{a'b'c'}
&&a'=a(\gamma')=\frac{1}{2}-b(\gamma),\,\,
b'=b(\gamma')=\frac{1}{2}-a(\gamma),\,c'=\frac{1}{2}+a'+b',\end{eqnarray}
where $\gamma_{\pm}$ is defined in \eqref{gammafin}:
\begin{eqnarray}\label{g+-dual}
\gamma_{\pm}=\frac{1}{\kappa}\left(1\pm\sqrt{1+2\kappa p}\right),\,\,\,
\beta(\gamma_{\pm})=3p-\frac{1}{2}-\frac{1}{2}\left(1\pm\sqrt{1+2\kappa p}\right).
\end{eqnarray}
 Recall that the constant
\begin{eqnarray}\label{C0bis} 
C_0(a,b):=\frac{\Gamma(c)}{\Gamma(a)\Gamma(b)}\frac{\Gamma(a')\Gamma(b')}{\Gamma(c')}=\frac{\Gamma(1/2+a+b)}{\Gamma(a)\Gamma(b)}\frac{\Gamma(1/2-a)\Gamma(1/2-b)}{\Gamma(3/2-a-b)}
\end{eqnarray}
is chosen such that $g(x)$ is singularity-free at $x=4$, i.e., at the point $z=-1$ on the unit circle. 
\end{rkk}}
\begin{figure}
\begin{center}
\label{BSfigbis}
\includegraphics[angle=0,width=.63290\linewidth]{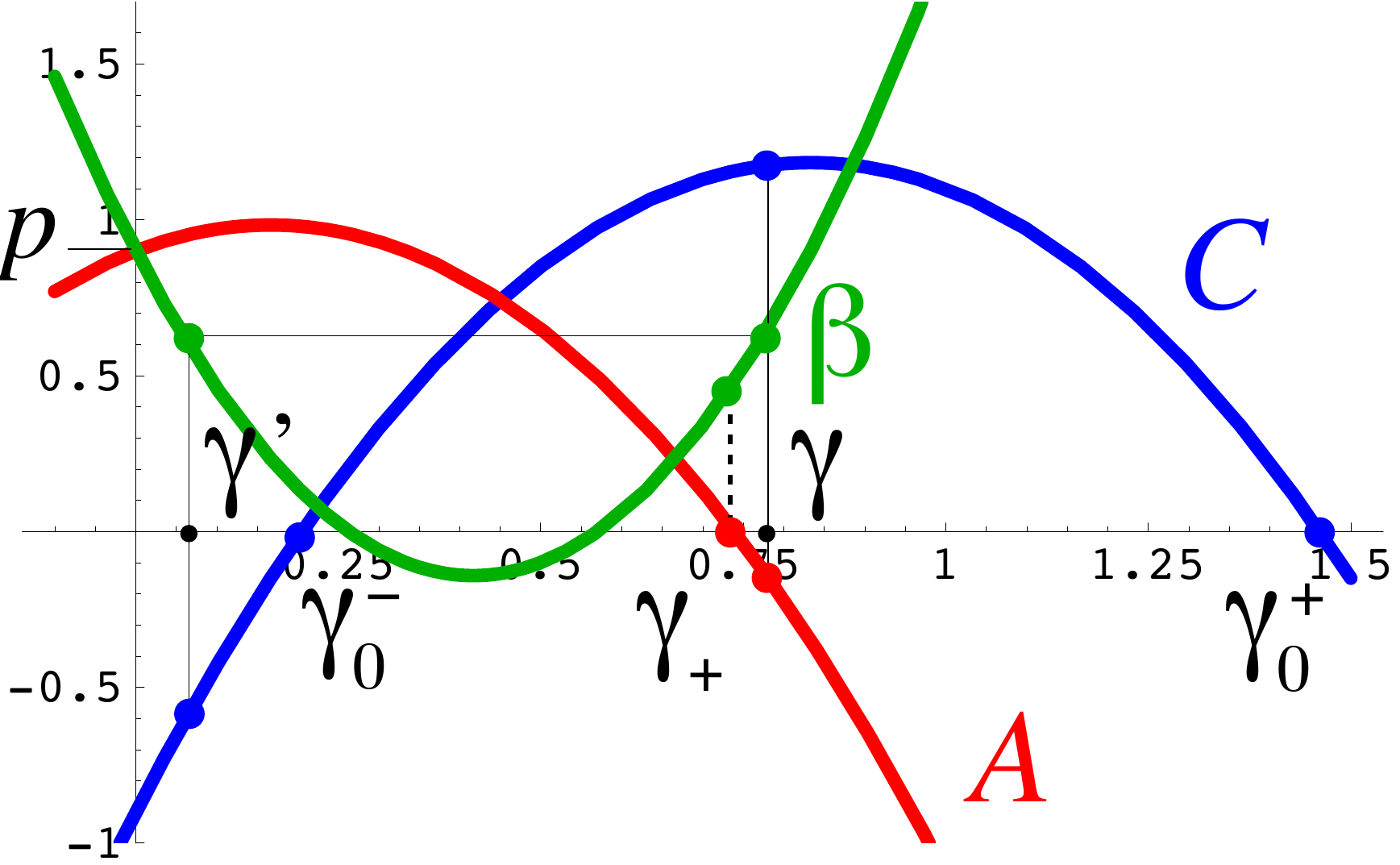}
\caption{\textcolor{black}{Curves $A(\gamma)$, $\beta(\gamma)=\kappa\gamma^2/2-C(\gamma)$, and  $C(\gamma)$, for the value $\kappa=6$ of the SLE parameter, and for a generic $p=0.9\in \big(p^*(6),p(6)\big)$. The dual parameters $(\gamma,\gamma'
)$ obey $\gamma+\gamma'=2/\kappa+1/2$, so that $\beta(\gamma)=\beta(\gamma')$. The choice of parameter $\gamma\in\big[\gamma_+,\gamma_0^+\big)$ yields values $\beta(\gamma) \geq \beta(\gamma_+)$, $A(\gamma)\leq 0$, $C(\gamma)> 0$, and $A(\gamma')>0$, $C(\gamma')< 0$.}}
\label{BSfigter}
\end{center}
\end{figure}\textcolor{black}{In Sections \ref{firstsol} and \ref{singan}, we studied the properties of the truncated function $\psi_+(z,\bar z)=x^{\gamma_+} (1-z\bar z)^{-\beta(\gamma_+)}$, which, once modified by an adequate logarithmic factor, provided a \textit{supersolution} to the BS differential equation for $p\in [p^*(\kappa),p(\kappa)]$, and a \textit{subsolution} for $p\geq p(\kappa)$. Here we propose to use the complete solution \eqref{ghyperbis} for a value of the parameter $\gamma$ slightly above $\gamma_+$, in order to obtain a \textit{subsolution}  for $p\geq p^*(\kappa)$, such that $\beta(\gamma)>\beta(\gamma_+)$ (Fig. \ref{BSfigter}). We therefore set $\gamma=\gamma_+ +a$, with $a>0$, and will ultimately let $a\to 0^+$ to prove that the limit $\beta(\gamma_+)$ yields the rigorous average integral means spectrum. 
\begin{rkk}\label{gamma'gamma0}
We have seen above that the phase transition at $p^*(\kappa)$ occurs when $1/2-b(\gamma_0)=0$. Because of the duality equation \eqref{a'b'c'}, this corresponds to the equality  $a(\gamma'_0)=0$, hence $\gamma'_0=\gamma_+$,  i.e., $\gamma_0=\gamma'_+$.  One has furthermore $\gamma_+'\geq \gamma_0$ for $p\leq p^*(\kappa)$ and $\gamma_+'\leq \gamma_0$ for $p\geq p^*(\kappa)$. Since $\gamma_0=\gamma_0^-$ is the lower value such that $C(\gamma_0)=0$, one has $C(\gamma_+') \leq 0$ after the transition, hence $C(\gamma')<0$ for $\gamma'<\gamma_+'$ (see Fig. \ref{BSfigter}). This last property is the key to obtain a subsolution  to the partial differential equation for $\gamma>\gamma_+$.
\end{rkk}
$\bullet$ \textit{The $x\to 0$, $z\to 1$ limit}. In this limit, the leading term in the function $g$ \eqref{ghyperbis} is the second one, since $\gamma'<\gamma$; it gives the equivalent 
$ g(x)\sim -C_0 ({x}/{4})^{\gamma'}$, 
so that $g_0$ \eqref{g0ter} diverges as $g_0(x)\sim  -C_0 ({x}/{4})^{\gamma'-\gamma}$, while for \eqref{g0tilde}, $\tilde g_0(0)=-C_0$. It is easy to see from \eqref{C0bis} that for $a$ small and positive
$$C_0(a,b) \sim \frac{a}{1/2-b} \frac{\Gamma(1/2+b)\Gamma(1/2)}{\Gamma(b)}.$$
We furthermore have $1/2-b(\gamma)=a(\gamma')=\gamma'-\gamma_+=\gamma_+'-\gamma_+-a<0$, so that $C_0<0$ and $g(x) >0$ for $x$ small, and $\tilde g_0(0)>0$.}

\textcolor{black}{$\bullet$ \textit{Sign of $g$.} To show that $g$ \eqref{ghyperbis} is positive on the interval $[0,4]$, we  argue as in Ref. \cite{BS}.  Since $\gamma_+$ is the rightmost zero of $A$,  we have $A(\gamma)<0$ for $\gamma >\gamma_+$ (Fig. \ref{BSfigter}).
 In the hypergeometric equation \eqref{hypergeom}, the signs of the $g_0$ and $g''_0$ terms are thus opposite.  We already know that $g(0^+)>0$, and it suffices to prove that at the endpoint $x=4$, $g(4)=g_0(4)=\tilde g_0(4)>0$, to show the positivity of $g$ on the whole interval. From \eqref{ghyperbis} and \eqref{C0bis}, we have \cite{BS}
 $$g(4)=\frac{\Gamma(1/2)\Gamma(1/2+a+b)}{\Gamma(1/2+a)\Gamma(1/2+b)}\left(1-\tan \pi a\tan \pi b\right).$$
 By continuity, it is sufficient to study the sign of this quantity for $a\to 0^+$. Recall that $b=a+1/2+\gamma_+-\gamma_+'$, where the constant added  to $a$ is positive. If that constant is different from $1/2+n$, $n\in \mathbb N$, then $g(4)$ tends to $1$ for $a\to 0^+$. If the constant happens to be  an half-integer, $1-\tan \pi a\tan \pi b=2$, independently of $a$, and $g(4)$ is positive for $a\geq 0$, and tends to $2$ for $a\to 0^+$.} 
 
 \textcolor{black}{$\bullet$ \textit{Action of the operator $\mathcal P(D).$}
 The action of the partial differential operator $\mathcal P(D)$ on the function in the unit disk, $\psi_0(z,\bar z)=(1-z\bar z)^{-\beta(\gamma)} g(x)$, where $g(x)$ is the boundary solution \eqref{ghyperbis}, \eqref{g0ter}, is given by  \eqref{genresquater}, here specified for $\sigma=-1$,
 \begin{eqnarray} \nonumber\frac{\mathcal P(D)[\psi_0(z,\bar z)]}{\psi_0(z,\bar z)}&=&(1-z\bar z)\left[-\frac{1}{x}\big[C(\gamma)+A(\gamma)\big]+\frac{1}{4-x}2A(\gamma)+\left(\frac{\kappa}{2}-1-\frac{2\kappa}{4-x}\right)\frac{g'_0}{g_0}\right] \\ \nonumber &+&\frac{(1-z\bar z)^2}{x^2} \left\{\frac{1}{4-x}\left[4A(\gamma)-\kappa x\frac{g_0'}{g_0}\right]-2p+\big(\frac{\kappa}{2}+1\big)\big(\gamma+x\frac{g_0'}{g_0}\big)\right\}.\\ 
&& \label{genresquinter}
\end{eqnarray}
We now use the $(\gamma,\gamma')$ duality \eqref{gg'} and the associated dual function $\tilde g_0$ \eqref{g0tilde}, to rewrite the operator's action as
  \begin{eqnarray} \nonumber\frac{\mathcal P(D)[\psi_0(z,\bar z)]}{\psi_0(z,\bar z)}&=&(1-z\bar z)\left[-\frac{1}{x}\big[C(\gamma')+A(\gamma')\big]+\frac{1}{4-x}2A(\gamma')+\left(\frac{\kappa}{2}-1-\frac{2\kappa}{4-x}\right)\frac{\tilde g'_0}{\tilde g_0}\right] \\ \nonumber &+&\frac{(1-z\bar z)^2}{x^2} \left\{\frac{1}{4-x}\left[4A(\gamma')-\kappa x\frac{\tilde g_0'}{\tilde g_0}\right]-2p+\big(\frac{\kappa}{2}+1\big)\big(\gamma'+x\frac{\tilde g_0'}{\tilde g_0}\big)\right\}.\\ 
&& \label{genrestilde}
\end{eqnarray}
\begin{rkk}
The Remark \eqref{singularity4}  about the absence of singularity in the operator's action at $x=4$ is also valid after using duality in \eqref{genrestilde}. More specifically, the identity \eqref{Asigmag0} there can be recast as 
$$ \frac{\tilde g'_0(4)}{\tilde g_0(4)}=-\frac{1}{2}(\gamma'-\gamma_+)(\gamma'-\gamma_-)=\frac{1}{\kappa}A(\gamma').$$ 
\end{rkk}
\begin{rkk}In the $x \to 0$ limit, because $x\tilde g'_0(x)/\tilde g_0(x)=O(x^{\gamma-\gamma'})$, the second line of \eqref{genrestilde} is equivalent to 
\begin{eqnarray}\label{Cxto0tilde}\frac{(1-z\bar z)^2}{x^2} \left\{A(\gamma')-2p+\left(\frac{\kappa}{2}+1\right)\gamma'\right\}  = \frac{(1-z\bar z)^2}{x^2} C(\gamma').\end{eqnarray}
Because of Remark \ref{gamma'gamma0}, we know that $C(\gamma')<0$ for $p\geq p^*(\kappa)$.
\end{rkk}
We can therefore write \eqref{genrestilde} under the form
  \begin{eqnarray} \label{h0th0}\frac{\mathcal P(D)[\psi_0(z,\bar z)]}{\psi_0(z,\bar z)}&=&\frac{1-z\bar z}{x} h_0(x) +\frac{(1-z\bar z)^2}{x^2} \tilde h_0(x) \label{genresdual}\\ \nonumber \label{h0}
  h_0(x)&:=&-\big[C(\gamma')+A(\gamma')\big]+\frac{x}{4-x}2A(\gamma')+\left(\frac{\kappa}{2}-1-\frac{2\kappa}{4-x}\right)x\frac{\tilde g'_0}{\tilde g_0}\\ \nonumber \label{tildeh0}
  \tilde h_0(x)&:=&\frac{1}{4-x}\left[4A(\gamma')-\kappa x\frac{\tilde g_0'}{\tilde g_0}\right]-2p+\big(\frac{\kappa}{2}+1\big)\big(\gamma'+x\frac{\tilde g_0'}{\tilde g_0}\big),
\end{eqnarray}
where $h_0(x)$ and $\tilde h_0(x)$ are two bounded functions on the interval $[0,4]$. Owing to \eqref{Cxto0tilde}, $\tilde h_0(0)=C(\gamma')<0.$
}

\textcolor{black}{$\bullet$ \textit{Logarithmic modification.} As in Section \ref{firstsol}, let us consider now the action of the differential operator on the modified function $\psi_0(z,\bar z)\ell_\delta(z\bar z)$, with the logarithmic factor \begin{equation}\nonumber \ell_\delta(z\bar z):= [-\log (1-z\bar z)]^\delta,\,\, \delta \in \mathbb R .\end{equation}  Eq. \eqref{Ppsiell} yields
\begin{eqnarray}\nonumber \mathcal P(D)[\psi_0(z,\bar z)\ell_\delta(z\bar z)]
=\ell_\delta(z\bar z)\left\{\mathcal P(D)[\psi_0(z,\bar z)]-\psi_0(z,\bar z) \frac{2\delta z\bar z x^{-1}}{[-\log(1-z\bar z)]}\right\};\end{eqnarray}
 Eq. \eqref{h0th0} then gives: 
\begin{eqnarray}\label{Ppsi0tilde} \frac{x\mathcal P(D)[\psi_0(z,\bar z)\ell_\delta(z\bar z)]}{\ell_\delta(z\bar z)\psi_0(z,\bar z)}= (1-z\bar z) h_0(x)+\frac{(1-z\bar z)^2}{x} \tilde h_0(x)-\frac{2\delta z\bar z }{[-\log(1-z\bar z)]}.\end{eqnarray}
\noindent  $\bullet$ Consider now the annulus $\mathbb A(r)=\{z: r<|z|<1\}$, and its intersection with the  domain $\mathbb D\setminus \mathbb D_{1/2}$, where $1-z\bar z\leq x$ (Fig. \ref{figlune}). In this domain, $$(1-z\bar z) h_0(x)+(1-z\bar z)^2 x^{-1} \tilde h_0(x)=O(1-z\bar z),$$  so that this term is dominated by the logarithmic term in \eqref{Ppsi0tilde} for $r$ close enough to $1$. The sign of $\mathcal P(D)[\psi_0\ell_\delta]$ is therefore given by that of $-\delta$ in this domain.\\
$\bullet$ Consider next the domain $\mathbb A(r)\cap \mathbb D_{1/2}$. For $r \to 1^-$, the first term  in \eqref{Ppsi0tilde} is dominated as before by the logarithmic one. The second term is not $O(1-z\bar z)$, but its sign, for $x$ small enough, hence for $r$ near $1$, is that of $\tilde h_0(0)=C(\gamma')<0$. Therefore, if we choose $\delta >0$, the sign of the r.h.s. of \eqref{Ppsi0tilde} will be negative in $\mathbb A(r)\cap \mathbb D_{1/2}$,  for $r$ close enough to $1$.\\ As seen just above, this  also holds in $\mathbb A(r)\cap(\mathbb D\setminus \mathbb D_{1/2})$, so that we conclude that for $\delta >0$, $\mathcal P(D)[\psi_0\ell_\delta]<0$ in the whole annulus $\mathbb A(r)$, i.e., $\psi_0\ell_\delta$ is  a \textit{positive subsolution} there.}\\ 
  
 \textcolor{black}{We then follow the same method as in Proposition \ref{annineq} and Section \ref{singan}. 
 The subsolution $\psi_{0}(z,\bar z)\ell_\delta(z\bar z)$ is a positive function on $\mathbb A(r)$, as  the true solution  $F(z,\bar z)=\mathbb E\big[|f'_0(z)|^p\big]$ is. The maximum principle then yields 
   \begin{prop} \label{annineqdual} There exists a positive constant $c$ such that 
  \begin{eqnarray}\label{F+dual}&&F<c\,\psi_0\,\ell_\delta,\,\,\, z\in \mathbb A(r),\,\,\,  p \geq p^*(\kappa).
 \end{eqnarray} 
  \end{prop}
Recall now that $\psi_0$ here involves a parameter $\gamma>\gamma_+$, whereas  $\psi_+=\psi_+(z,\bar z)=x^{\gamma_+} (1-z\bar z)^{-\beta(\gamma_+)}$, which appears in  Proposition \ref{annineq} (about the existence of a positive constant $c_2$, such that $c_2\,\psi_+\,\ell_\delta < F$ for $z\in \mathbb A(r)$ and $p<p(\kappa)$), involves  $\gamma=\gamma_+$ exactly.}\\

\textcolor{black}{$\bullet$ \textit{Proof of Theorem \ref{theoMFrigorous}}.\\
Recall that the average  integral means spectrum $\beta(p):=\beta(p,\kappa)$ of the whole-plane SLE$_\kappa$ is defined  by \eqref{intmeanF} (a notation which should not be confused with $\beta(\gamma)$, as defined in \eqref{betapgamma}). For the function    $\psi_0(z,\bar z)=\psi_0(r,\theta)$,  the integral means are: 
\begin{equation}\label{intmeandual} \int_0^{2\pi} \psi_0(r,\theta)d\theta =(1-r^2)^{-\beta(\gamma')} \int_0^{2\pi} g(r,\theta)d\theta,
\end{equation} 
where we write \eqref{ghyperbis} as $g(x)=g(|1-re^{i\theta}|^2)=g(r,\theta)$. 
 Recall that $g(x)\sim x^{\gamma'}$ for $x\to 0$, where $\gamma'=\gamma'_+-a\leq \gamma'_+$ can be negative, 
 and  the singularity along the unit circle at $\theta =0$ is no longer integrable when $2\gamma'+1\leq 0$. For the upper limit of $\gamma'$,  $\gamma'_+=1/2+(1/\kappa)(1-\sqrt{1+2\kappa p})$, this  corresponds to a cross-over value $p=\hat p(\kappa):=1+\kappa/2$, after which $2\gamma'_+ +1<0$.}
 
 \textcolor{black}{Consider now the logarithmically modified functions 
 $\psi_{+}\,\ell_\delta$ and $\psi_0\,\ell_\delta$, whose integral means asymptotic power law behaviors are obviously  the same as for $\psi_{+}$ and $\psi_0$:
 \begin{eqnarray} \label{g+elldual}
 \int_0^{2\pi} \psi_{+}(r,\theta)\ell_\delta(r^2)d\theta&\stackrel{(r\to 1^-)}{\asymp}& (1-r)^{-\beta(\gamma_+)},\\ \label{g-elldual} 
  \int_0^{2\pi} \psi_{0}(r,\theta)\ell_\delta(r^2)d\theta&\stackrel{(r\to 1^-)}{\asymp}& \begin{cases}(1-r)^{-\beta(\gamma')},\,\,\,\,2\gamma'+1 \geq 0,\\  
  (1-r)^{-\beta(\gamma')+2\gamma'+1},\,\,\,\, 2\gamma'+1 < 0.\end{cases}\end{eqnarray}
  }
\textcolor{black}{By using  the asymptotic behaviors \eqref{intmeanF}, \eqref{g+elldual}, and \eqref{g-elldual} in Propositions \ref{annineq} and \ref{annineqdual}, we obtain  \begin{eqnarray}
 \label{+dual}
&&\beta(\gamma_+)\leq \beta(p,\kappa),\,\,\, p^*(\kappa)\leq p\leq p(\kappa)=\frac{(6+\kappa)(2+\kappa)}{8\kappa},\\ \label{+0dual} &&{\beta(p,\kappa)}\leq \beta(\gamma')=\beta(\gamma),\,\,\,0\leq 2\gamma'+1, \,\,\, p^*(\kappa)\leq p, \\  \label{++dual}
&&{\beta(p,\kappa)}\leq \beta(\gamma')-2\gamma'-1,\,\,\, 2\gamma'+1<0, \,\,\, p^*(\kappa)\leq p. \end{eqnarray} 
 Suppose first that $2\gamma_+'+1>0$, i.e., $p<\hat p(\kappa)=1+\kappa/2$, then $2\gamma'+1=2\gamma'_++1-2a$ is non-negative for $a>0$ small enough. Eq. \eqref{+0dual} then gives by duality $\beta(p,\kappa)\leq \beta(\gamma')=\beta(\gamma'_+-a)=\beta(\gamma_++a)=\beta(\gamma)$.  Similarly, if $2\gamma_+'+1=0$, i.e., $p=\hat p(\kappa)$, then $2\gamma'+1=-2a$, and by \eqref{++dual}, $\beta(p,\kappa)\leq \beta(\gamma')+2a=\beta(\gamma'_+-a)+2a=\beta(\gamma_++a)+2a$.  In both cases, by letting  $a\to 0+$, we obtain $\beta(p,\kappa)\leq \beta(\gamma_+)$ for $p^*(\kappa)\leq p \leq \hat p(\kappa)$.  
 By combining this with Eq. \eqref{+dual}, we obtain the expected identity $\beta(p,\kappa)=\beta(\gamma_+)$ for $p^*(\kappa)\leq p \leq  \min\{\hat p(\kappa),p(\kappa)\}$. By recalling that $\beta(\gamma_+)=\beta(\gamma_+(p))=\beta_+(p)$ (Eqs. \eqref{gammafin},\eqref{betapgammafin}), we finally obtain $\beta(p,\kappa)=3p-1/2-(1/2)\sqrt{1+2\kappa p}$, i.e.,Theorem \ref{theoMFrigorous}.}
\textcolor{black}{\begin{rkk} For $p>p(\kappa)$, in contrast to \eqref{+}, \eqref{+dual}, we have  from \eqref{++}  the (subsolution) inequality $\beta(p,\kappa)\leq \beta_+(p)=\beta(\gamma_+)$; this simply co\"{\i}ncides with the inequality obtained here by the duality method. Hence, we cannot  prove that $\beta(p,\kappa)=\beta_+(p)$ for $p> p(\kappa)$ by this method.
\end{rkk}} 
 
\subsection{Spectrum of the $m$-fold whole-plane SLE ($m\geq 1$)}\label{subsecSpecm}
\textcolor{black}{In this section, we address the derivation of Statement \ref{theoMF} for general $m$. 
For the $m$-fold inner whole-plane SLE$_\kappa$ map, defined as $h_0^{(m)}(z)=z\big[f_0(z^m)/z^m\big]^{1/m}$, $m\geq 1$, the average integral means spectrum is, for $p\geq 0$,
\begin{eqnarray}\label{betamm}
\beta_m(p,\kappa)&=& \max \left \{\bar \beta_0(p,\kappa),B_m(p,\kappa)\right\},\\ \nonumber
B_m(p,\kappa)&=&\left(1+\frac{2}{m}\right)p-\frac{1}{2}-\frac{1}{2}\sqrt{1+\frac{2\kappa p}{m}},
\end{eqnarray}
where $\bar\beta_0(p,\kappa)$ is the BS expected integral mean spectrum \eqref{beta00bar}. 
The phase transition takes place when the second term $B_m$ on the r.h.s.  of \eqref{betamm} equals, then exceeds, the first one. For $1\leq m\leq 3$, this takes place  at the critical point \eqref{pmstar} $p_m^*(\kappa)\leq p_0^*(\kappa)$, hence before the transition point $p_0^*(\kappa)$ \eqref{p00star} of the BS spectrum $\beta_0$ in \eqref{beta00bar} to the linear behavior $\hat \beta_0$ \eqref{tildebeta00}. 
For $m\geq 4$, the order of the two critical points $p^*_0(\kappa)$ and $p^*_m(\kappa)$ depends on $\kappa$, and is given by \eqref{km0}
\begin{eqnarray*}\label{km}
p_m^*(\kappa)\lesseqqgtr p_0^*(\kappa),\,\,\,\,\, \kappa\lesseqqgtr \kappa_m,\,\,\,\,\kappa_m:=4\frac{m+3}{m-3},\,\,\,\,\, m\geq 4,
\end{eqnarray*}
such that for $\kappa\leq \kappa_m$, 
\begin{eqnarray*}\label{betamm0}
\beta_m(p,\kappa)=\begin{cases}\beta_0(p,\kappa),\,\,\,\,0\leq p\leq p_m^*(\kappa),\,\,\,
\\ B_m(p,\kappa),\,\,\,\, p_m^*(\kappa)\leq p,\end{cases}
\end{eqnarray*}
whereas  for $\kappa\geq \kappa_m$,
\begin{eqnarray*}\label{in0bis}\beta_m(p,\kappa)=\begin{cases}\beta_0(p,\kappa),\,\,\,\,0\leq p\leq  p_0^*(\kappa),\,\,\,
\\ 
\hat \beta_0(p,\kappa),\,\,\,\,p_0^*(\kappa)\leq p\leq p^{**}_m(\kappa),
\\ 
 B_m(p,\kappa),\,\,\,\,p_m^{**}(\kappa)\leq p,\end{cases}
\end{eqnarray*}}
where $p_m^{**}(\kappa)$ is the second critical point \eqref{pdoublestar0} 
$p_m^{**}(\kappa):=m(\kappa^2-16)/{32\kappa},$ where the spectrum $B_m(p,\kappa)$ intersects the linear spectrum $ \hat \beta_0(p,\kappa)$ \eqref{tildebeta00}.

These results for $\beta_m(p,\kappa)$  obviously  satisfy Makarov's Theorem \ref{mtheo} \cite{Makanaliz}  for $m$-fold symmetric functions: \begin{equation}\label{mak2}\beta_m(p,\kappa)\leq \left(1+\frac{2}{m}\right)p-1\,\, \textrm{for}\,\, p\geq \frac{2m}{4+m}.\end{equation}

\subsubsection{Derivation of Statement \ref{theoMF}} Rather than providing here in full detail the calculation of the spectrum for the $m$-fold whole plane SLE map, which is quite similar to those of Section \ref{IMS1} above, 
  we shall take the following shortcut, as suggested by Remark \ref{remark46} and by the comment after Eq. \eqref{PpsiA}. 

We now use the identities \eqref{Aalpham}, \eqref{Balpham} and \eqref{Calpham}  giving $A$, $B$ and $C$ in the $m$-fold case.
 We first remark that the expression \eqref{Calpham} for $C$ does not depend on $m$, hence  the standard spectrum, as given by  $\beta=\kappa\alpha^2/2$  and $C= -\frac{\kappa}{2}\alpha(\alpha-1)+2\alpha-p=0$, co\"{\i}ncides with the BS spectrum $\beta_0(p)$ (for the choice of the negative branch solution $\beta_0^-(p)$ to $C=0$). This is expected, since this part of the spectrum should correspond to the \emph{fine multifractal structure of the bulk of the SLE curve}, which should  stay \textit{invariant}  under any $m$-fold transform. 
 
 The spectrum corresponding to the \textit{unbounded} nature of the $m$-fold whole-plane SLE can now be obtained  by setting $A_m=0$ in \eqref{Aalpham}, and using again $\beta=\kappa \alpha^2/2-C$ together with $C$  \eqref{Calpham}. This gives the two solutions 
 \begin{eqnarray}\label{gammam}
 \alpha={\gamma}_{m}^{\pm}(p,\kappa)&:=&\kappa^{-1}\big(1\pm \sqrt{1+2\kappa p/m}\big),\\ \nonumber
\beta= B_{m}^{\pm}(p,\kappa)&:=&\left(\frac{2}{m}+1\right)p-\frac{\kappa}{2}{\gamma}_m^{\pm}(p,\kappa)\\ \label{betampm} &=&\left(\frac{2}{m}+1\right)p-\frac{1}{2}\left(1\pm \sqrt{1+2\kappa \frac{p}{m}}\right). \end{eqnarray}
\textcolor{black}{Note that in the case $m=1$, the two functions $B_1^{\pm}(p,\kappa)$ coincide with the functions $\beta_{\pm}(p,\kappa)$ used in Section \ref{IMS1}, and  defined in \eqref{betapgammafin}.} Thanks to the universal spectrum for $m$-fold symmetric analytic functions, as given by Makarov's Theorem \ref{mtheo}, the negative branch $B_{m}^{-}$ is clearly excluded, while the positive one $B_{m}^{+}(\equiv B_m \eqref{Bm})$ satisfies the universal bound \eqref{mak2}. The transition point where $B_{m}(p,\kappa)=\beta_0(p,\kappa)$ is given by $p=p^*_m(\kappa)$ in Eq. \eqref{pmstar}. 
 
\textcolor{black}{Consider first the case $1\leq m\leq 3$ for which, for all $\kappa$, this transition point $p_m^*(\kappa)\leq p_0^*(\kappa)$ \eqref{p00star}, where  the BS spectrum \eqref{beta000} in \eqref{beta00bar} changes to the linear spectrum \eqref{tildebeta00}. For $p\leq p^*_m(\kappa)$ one has $\beta_0(p,\kappa)\geq B_{m}(p,\kappa)$, hence $\beta_m(p,\kappa)=\beta_0(p,\kappa)$. For $p\geq p^*_m(\kappa)$, the unbounded integral means spectrum $B_m(p,\kappa)$ dominates, hence $\beta_m(p,\kappa)=B_m(p,\kappa)$.}
 
 The alternative inequality $p_0^*(\kappa)\leq p_m^*(\kappa)$ arises only for $m\geq 4$ and for $\kappa\geq \kappa_m$ \eqref{km0}. In this case, the phase transition at $p^*_0(\kappa)$ to the linear piece  \eqref{tildebeta00} of the BS spectrum appears \textit{before} the phase transition from $\beta_0(p,\kappa)$ to the $m$-fold unbounded spectrum $B_{m}(p,\kappa)$  happens at $p_m^*(\kappa)$. The latter transition  therefore takes place at the second phase transition point $p_m^{**}(\kappa)$ \eqref{pdoublestar0}, where $\hat \beta_0(p,\kappa)$ intersects $B_{m}(p,\kappa)$. We thus expect the sequence of spectra \eqref{in0} to take place for $m\geq 4, \kappa\geq \kappa_m$.   This concludes the (non-rigorous) derivation of Statement \ref{theoMF}. 
\subsubsection{Proof of Theorem \ref{theoMFcm}} 
\begin{proof} A rigorous proof of Theorem \ref{theoMFcm} for general $m$ can be achieved in the same manner as in Section \ref{IMS1} above for the $m=1$ case. \textcolor{black}{Using the differential operator $\mathcal P_m(D)$ \eqref{BSm} 
instead of \eqref{zz1} (or the cylindrical coordinate version thereof instead of \eqref{BSsig}) in the analysis of Section \ref{IMS1}, we search for approximate solutions to $\mathcal P_m(D)[\psi(\zeta,\bar\zeta)]=0$ \eqref{BSmm} of the form} 
\begin{equation}
\psi(\zeta,\bar\zeta):=(1-\zeta\bar \zeta)^{-\beta} x^\gamma,
\end{equation} with here $$\zeta:=z^m,\,\,\,x:=(1-\zeta)(1-\bar \zeta)=(1-z^m)(1-\bar z^m).$$
Since $\psi(\zeta,\bar\zeta)=(1-\zeta\bar\zeta)^{-\beta}\varphi_\gamma(\zeta)\varphi_\gamma(\bar\zeta)$, with $\varphi_\gamma(\zeta)=(1-\zeta)^\gamma$,  the same algebra as in \eqref{Pphialpham} in Section \ref{action} yields, for arbitrary values of $\beta,\gamma$, the analogue of the action \eqref{genres}
\begin{eqnarray}\label{genresm}
\mathcal P_m(D)[\psi(\zeta,\bar \zeta)]&=&\psi(\zeta,\bar \zeta) x^{-1} \left\{(\kappa \gamma^2-2\beta)\zeta\bar \zeta-A_m(\gamma) (1-\zeta\bar \zeta-x)\right. \\ \nonumber &&\left.+C(\gamma)\left[(1-\zeta\bar \zeta)\big(\frac{1-\zeta\bar \zeta}{x}+1\big)-2\right]\right\},\end{eqnarray} where $A_m$ is given by \eqref{Aalpham} and $C$ by \eqref{Calpham} (for $\alpha=\gamma$), in conjunction with \eqref{Pphialpham}.
We write  $\psi=\psi_+$ for the choice $\beta:=B_{m}^{+}(p,\kappa)$ and $\gamma:=\gamma_{m}^+(p,\kappa)$, as defined in Eqs. \eqref{gammam} and \eqref{betampm}, such that $A_m(\gamma)=0$ and $C(\gamma)=\kappa\gamma^2/2-\beta$. The action \eqref{genresm} then simply becomes, as in \eqref{Ppsi}, 
\begin{eqnarray}\label{Ppsim}
\mathcal P_m(D)[\psi_+(\zeta,\bar \zeta)]= 
&=&\psi_+(\zeta,\bar \zeta)\left(\frac{\kappa}{2}\gamma^2-\beta\right) (1-\zeta\bar \zeta)(1-\zeta\bar \zeta-x)x^{-2}.\end{eqnarray}
In complete analogy to Eqs. \eqref{signpsi+_} and \eqref{signpsi++}, we then have, with the special point $p_m(\kappa)$ defined as in \eqref{alpham}, 
\begin{eqnarray} \nonumber 
&&p< p_m(\kappa):\mathcal P_m(D)[\psi_+(\zeta,\bar \zeta)]> 0, \zeta\in \mathbb D_{1/2},\,\,\,
\mathcal P_m(D)[\psi_+] < 0, \zeta\in \mathbb D\setminus\overline{\mathbb D}_{1/2};\\ \nonumber
&&p= p_m(\kappa): \mathcal P_m(D)[\psi_+(\zeta,\bar \zeta)]=0,\zeta\in \mathbb D;\\ \nonumber 
&&p> p_m(\kappa):\mathcal P_m(D)[\psi_+(\zeta,\bar \zeta)]< 0, \zeta\in \mathbb D_{1/2}\,\,\,
\mathcal P_m(D)[\psi_+] > 0, \zeta\in \mathbb D\setminus\overline{\mathbb D}_{1/2}.\end{eqnarray}
 A modification by the logarithmic factor \eqref{ell} of $\psi_+$ into $\psi_+(\zeta,\bar\zeta)\ell_\delta(\zeta\bar\zeta)$,   yields the same  conclusions as in Section \ref{firstsol}: 
 there exist  open annuli  $\mathbb A_m(r_i):=\{\zeta: r_i<|\zeta|<1\}=\mathbb D\setminus \overline{\mathbb D}(r_i)$, $i=1,2$,   whose boundary includes $\partial \mathbb D$, and where one has respectively (for a specific sign of $\delta$ chosen appropriately to each case): \\
 $\bullet$ for $p < p_m(\kappa)$,  $\psi_+\ell_\delta$ is a \textit{supersolution} with $\mathcal P_m(D)[\psi_+(\zeta,\bar \zeta)\ell_\delta(\zeta\bar \zeta)] > 0$ for $\zeta\in \mathbb A_m(r_1)$; \\
 $\bullet$  for  $p> p_m(\kappa)$, $\psi_+\ell_\delta$ is a \textit{subsolution} with $\mathcal P_m(D)[\psi_+(\zeta,\bar \zeta)\ell_\delta(\zeta\bar \zeta)] < 0$ for $\zeta\in \mathbb A_m(r_2)$;\\ 
  $\bullet$ for  $p=p_m(\kappa)$,  $\mathcal P_m(D)[\psi_+(\zeta,\bar \zeta)] = 0, \zeta\in \mathbb D$, so that $\psi_+(\zeta,\bar \zeta)=F(\zeta,\bar \zeta)=(1-\zeta\bar \zeta)^{-\beta}|1-\zeta|^{2\gamma}$ is the exact solution  of Theorem \ref{main4} with parameters \eqref{alpham}: $\gamma=\gamma^+_{m}(p_m(\kappa),\kappa)=\alpha_m(\kappa)$ and $\beta=B^+_{m}(p_m(\kappa),\kappa)=\kappa \alpha_m(\kappa)^2/2$.
  
    We then follow the same method as above \cite{BKTH,BS}. The operator $\mathcal P_m(D)$, when written in polar coordinates,  is \emph{parabolic}. 
  Using in each of the two annuli where $\mathcal P_m(D)[\psi_{+}\ell_\delta]$ has a definite sign, respectively, the minimum principle, and  the maximum principle  (\cite{evans}, Th. 7.1.9), yields 
  \begin{prop} \label{annineqm} There exist two positive constants $c_i$, $i=1,2$, such that 
  \begin{eqnarray}
  \label{F+m} &&c_1\,\psi_+\,\ell_\delta < F,\,\,\,\zeta\in \mathbb A_m(r_1),\,\,\, p<p_m(\kappa),\\ \nonumber
 && \psi_+=F,\,\,\,\zeta\in \mathbb D,\,\,\,p=p_m(\kappa),\\ \label{F++m}&&F< c_2\,\psi_+\,\ell_\delta,\,\,\,\zeta\in \mathbb A_m(r_2),\,\,\,p>p_m(\kappa).
 \end{eqnarray} 
where $\zeta=z^m$ and $F=F(\zeta,\bar \zeta):=\mathbb{E}\left(|({h}^{(m)}_0)'(z)|^p\right)$, with  ${h}^{(m)}_0(z):=[f_0(z^m)]^{1/m}$.  
 \end{prop} 
From the inequality  \eqref{F+m} (resp. \eqref{F++m}), we therefore conclude  that  the spectrum  associated with $\psi_+$ or $\psi_+\ell_\delta$, 
\begin{equation}\nonumber B_{m}^{+}(p,\kappa)\equiv B_m(p,\kappa)=\left(1+\frac{2}{m}\right)p-\frac{1}{2}-\frac{1}{2}\sqrt{1+\frac{2\kappa p}{m}},\end{equation}
 is, for $p\leq p_m(\kappa)$, (resp. for $p\geq p_m(\kappa)$), a \textit{lower bound} $B_{m}\leq \beta_m$ (resp. \textit{upper bound}  $B_{m}\geq \beta_m$) to the \textit{exact} average integral means spectrum  $\beta_m(p,\kappa)$ of the $m$-fold inner whole-plane SLE$_\kappa$. 

Recall then that  the BS average integral means spectrum $\bar\beta_0(p,\kappa)$ \eqref{beta0}  becomes \textit{smaller} than the spectrum $B_{m}(p,\kappa)$  at the transition point $p=p_m^*(\kappa)$ \eqref{pmstar}  [for $1\leq m\leq 3, \forall \kappa$, or for $m\geq 4, \kappa\leq \kappa_m$ \eqref{km0}], or at the transition point $p=p_m^{**}(\kappa)$ \eqref{pdoublestar0} [for $m\geq 4,\kappa\geq \kappa_m$]. Observe that these transition values are both smaller than the special point $p_m(\kappa)$ \eqref{pkappam}.  We thus conclude that $\beta_m(p,\kappa)=\bar\beta_0(p,\kappa)$ before these transition points, whereas necessarily $\beta_m(p,\kappa)\geq B_{m}(p,\kappa)>\bar\beta_0(p,\kappa)$ after them. At the higher special value $p=p_m(\kappa)$, we know that the two spectra $\beta_m(p,\kappa)$ and $ B_{m}(p,\kappa)$ \textit{co\"incide}. For $p>p_m(\kappa)$, the inequality is reversed: $\beta_m(p,\kappa)\leq B_{m}(p,\kappa)$. This concludes the proof of Theorem \ref{theoMFcm}. 
 \end{proof}
 \subsection{Integral means spectrum and derivative exponents} \label{derivative}
 \subsubsection{Motivation}\label{motiv}
  In this section, we (heuristically) explain the striking relationship between the packing spectrum \eqref{packing} and the whole-plane average integral means spectrum \eqref{betaunb} after the phase transition that takes place at $p=p^*(\kappa)$ \eqref{pstar}. (See  also Remark \ref{tipexponents}.)  
 
 The average integral means spectrum  \eqref{betavdef} involves evaluating, for the whole-plane SLE map $f_0(z)$, the integral 
 \begin{equation}\label{intJ}
 \mathbb I_p(r):= \int_{\partial \bb D} \mathbb E\, \left[\I f_0'(rz)\I^p\right]\, \I dz\I,
 \end{equation} on a circle of radius $r<1$ concentric to $\partial \mathbb D$, and taking the limit for $r\to 1^-$
 \begin{equation}\label{betaI} 
 \beta(p)=\lim \sup_{r\to 1}\,\,\frac{\log \mathbb I_p(r)}{-\log(1-r)}.
 \end{equation}
 If $\mathbb I_p(r)$ has a power law behavior, an alternative definition of $\beta(p)$ would be such that
\begin{equation} \label{betapow}
(1-r)^{\beta(p)}\, \mathbb I_p(r)\stackrel{r\to 1}{\asymp} 1.
 \end{equation}
 To understand why the average integral means spectrum, for $p\geq p^*(\kappa)$ \eqref{pstar}, crosses over to the special whole-plane form \eqref{betaunb}, one should consider that the integrand in \eqref{intJ} behaves more like a distribution for $p$ large enough. Then, the circle integral \eqref{intJ} concentrates in the vicinity of the pre-image point $z_{0}:= f^{-1}_0(\infty)\in\partial \mathbb D$, which is sent to infinity by the unbounded whole-plane SLE map $f_0$ (Fig. \ref{whpl}). To see this and the relation to the packing spectrum, we first need to recall the relation of our inner whole-plane SLE to standard radial SLE.
 
\subsubsection{Radial and whole-plane SLE} Let us consider the standard {\it inner radial} SLE$_\kappa$ process $g_t$ in the unit disk $\mathbb D$ \cite{Schr}, satisfying the stochastic differential equation 
 $$
 \partial_t g_t(w)=g_t(w)\frac{\lambda(t)+g_t(w)}{\lambda(t)-g_t(w)},\,\,\,\lambda(t)=e^{i\sqrt{\kappa}B_t},\,\,\,t\geq 0. $$
It is defined for $w\in \mathbb D\setminus K_t$, where $(K_t, t\geq 0)$ is a random increasing family of subsets (hulls) of the unit disk that grows towards the origin $0$ (Fig. \ref{whsle}). The map $g_t$ is the unique conformal map from $\mathbb D\setminus K_t$ onto $\mathbb D$, such that $g_t(0)=0$ and $g_t'(0)=e^t$. 

 Denote by $g_t^{-1}(z),\,t\geq 0,\,z\in\mathbb D$, the inverse map  of $g_t$, which maps $\mathbb D$ to $\mathbb D\setminus K_t$ (Fig. \ref{whsle}).  It is such that $g_t^{-1}(0)=0$ and $(g_t^{-1})'(0)=1/g_t'(0)=e^{-t}$. It also has the same law as the continuation to \textit{negative times}, $g_{-t}$, of the forward radial map $g_t$. Consider now the inner whole-plane  map $f_{t}(z)$ as defined in \eqref{loewner}  for $t\geq 0,\,z\in \mathbb D$, its inverse map $f_{t}^{-1}$, and the whole-plane map at $t=0$, $f_0(z)$. Define
  \begin{equation}\label{phit}
 \varphi_{t}(z):=f_{t}^{-1}\circ f_0(z),\,\,\,z\in\mathbb D.
 \end{equation} 
We then have the identities in law \cite{MR2129588,MR2153402}
 \begin{equation}\label{phitinlaw}
 \varphi_{t}(z)\stackrel{\rm (law)}{=} g_{-t}(z)\stackrel{\rm (law)}{=}g_t^{-1}(z).
 \end{equation}  As  already mentioned in Section \ref{BSderiv} (see also \cite{MR2129588}), \textit{the limit for $t\to +\infty$ of $e^t \varphi_t(z)\stackrel{\rm (law)}{=} e^t g_{-t}(z)\stackrel{\rm (law)}{=} e^t g_t^{-1}(z)$  exists, and has  the same law as the inner whole-plane process $f_{0}(z)$:}
 \begin{eqnarray}\label{ftilde}
  \tilde f_t(z)&:=&e^t g_t^{-1}(z),\,\,\,z\in \mathbb D,\\
 \label{f0g}
 f_0(z)&\stackrel{\rm (law)}{=}&\lim_{t\to +\infty}\tilde  f_t(z)=\lim_{t\to +\infty}e^t g_t^{-1}(z).
 \end{eqnarray} 
 In the limit $t\to +\infty$ of the radial inverse SLE map $e^tg_t^{-1}$, the boundary circle $e^t\partial \mathbb D$ is pushed back to infinity, while the limit of hulls $(e^tK_t)_{t\to +\infty}$ becomes the whole-plane SLE hull (e.g., for $\kappa\leq 4$ the single slit $\gamma([0,\infty))$ in Fig. \ref{whpl}). 
 Since 
 the tip of $g_t^{-1}(\partial \mathbb D)$ is at distance of order $e^{-t}$ from $0$, the limit of the tip of $e^t g_t^{-1}(\partial \mathbb D)$ for $t\to+\infty$ stays at a finite distance from $0$, as does the tip $f_0(1)=\gamma(0)$  (Fig. \ref{whpl} and Fig.  \ref{whsle}).
  \begin{figure}[tb]\
\begin{center}
\includegraphics[angle=90,width=.93290\linewidth]{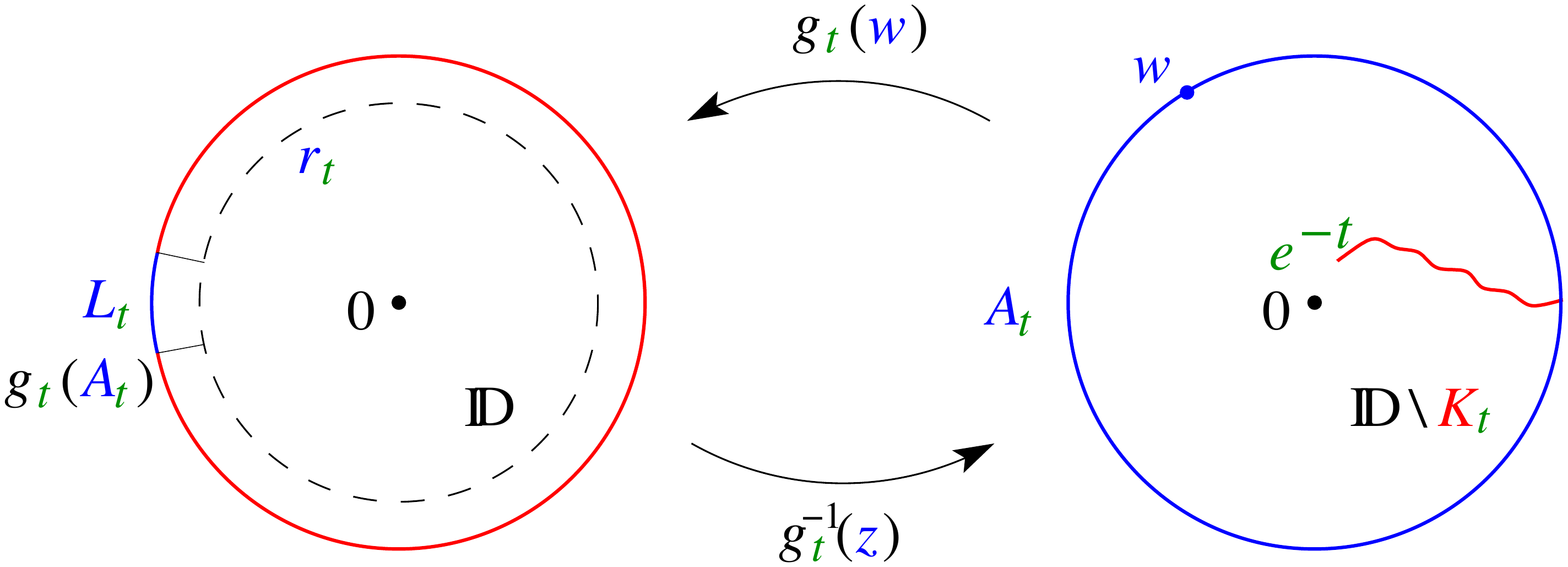}
  \vskip-2.039cm
 \caption{{\it Inverse Schramm--Loewner map $z\mapsto g^{-1}_t(z)$ from $\mathbb D$ to the slit domain $\mathbb D\setminus K_t$, where $K_t$ is the SLE$_\kappa$ hull (here a single curve for $\kappa\leq 4$). The distance from the SLE tip to the origin is of order $e^{-t}$ when $t\to+\infty$. The length $L_t:= |g_t(A_t)|$ of the image of the boundary set  $A_t:=\partial\mathbb D\setminus \overline{K_t}$   gives the harmonic  measure $L_t/2\pi$ of $A_t$ as seen from $0$ in $\mathbb D\setminus K_t$. The inner circle radius $r_t$  is chosen so that $1-r_t=L_t$. }}
 \label{whsle}
\end{center}
\end{figure}
 \subsubsection{Packing and derivative exponents}
 Let us briefly recall Lemma  {\bf 3.2} in Ref. \cite{MR2002m:60159b}: 
 \begin{lemma}\label{3.2}
 Let 
 $$
 A_t:=\partial\mathbb D\setminus\overline{K_t},
 $$
 which is either an arc on $\partial \mathbb D$ or $A_t=\emptyset$. Let $s\geq 0$, and set 
 \begin{equation}\label{nu}
\nu= \nu(s,\kappa):=\frac{s}{2} +\frac{1}{16}\left(\kappa-4+\sqrt{(4-\kappa)^2+16 \kappa s}\right).
  \end{equation}
  Assume $\kappa>0$ and $s>0$. Let $\mathcal H(\theta,t)$ denote the event $\{w=exp(i\theta)\in A_t\}$, and set
  \begin{eqnarray}\label{gprimes}
  \mathcal F(\theta,t)&:=&\mathbb E\left[\big|g_t'\big(\exp(i\theta)\big)\big|^s\,1_{\mathcal H(\theta,t)}\right],\\ \label{q}
  q=q(s,\kappa)&:=&\mathcal U^{-1}_\kappa(s)=\frac{\kappa-4+\sqrt{(4-\kappa)^2+16 \kappa s}}{2\kappa},\\ \nonumber
  h^*(\theta,t)&:=&\exp(-\nu\,t)\big(\sin(\theta/2)\big)^q.
  \end{eqnarray}
  Then there is a constant $c>1$ such that
  \begin{equation}\label{asymp} \forall t\geq 1,\,\,\,\forall \theta\in (0,2\pi),\,\,\,\,\,h^*(\theta,t)\leq \mathcal F(\theta,t)\leq c\,h^*(\theta,t),
 \end{equation} which we denote by $h^*(\theta,t)\asymp \mathcal F(\theta,t).$     \end{lemma}
By conformal invariance, the harmonic measure from $0$ of the boundary arc $A_t$ in the slit domain $D_t:=\mathbb D\setminus K_t$ is $L_t/2\pi$, where $L_t$ is the length of the arc $g_t(A_t)$.  Let us then also recall Theorem {\bf 3.3} in \cite{MR2002m:60159b}: 
 \begin{theo}\label{3.3}
 Suppose the $\kappa>0$ and $s\geq 1$. Then, when $t\to +\infty$,
 $$
 \mathbb E\left[(L_t)^s\right]\asymp\exp(-\nu\, t).
 $$
 \end{theo}
 Notice that Lemma \ref{3.2} and Theorem \ref{3.3} taken together strongly suggest that $|g_t'(w)|\asymp L_t$ for $w\in A_t$ and $t\to +\infty$.

 Let us now use \eqref{f0g}, and replace the whole-plane map $f_0(z)$ by its large time equivalent in law, $\tilde f_t(z)=e^t g_t^{-1}(z)$ \eqref{ftilde}, taken for some large time $t$. The domain which is sent far away from the origin by this map $\tilde f_t$ is the subset $g_t(A_t)  \subset \partial \mathbb D$ of the unit circle, as well as its immediate vicinity in $\mathbb D$ (see Fig. \ref{whsle}). This corresponds in the image domain $ \mathbb D\setminus  K_t$ to $w\in A_t$.  Define then the restricted boundary integral:
 \begin{eqnarray}\label{intg}
 \mathcal I_p(t) &:=& \int_{A_t} e^{pt}\I (g'_t(w)\I^{s}\, \I dw\I\\ \nonumber &=& \int_{0}^{2\pi} e^{pt}\I (g'_t\big(e^{i\theta}\big)\I^{s}\,1_{\mathcal H(\theta,t)}\, d\theta \\ \label{spb}
  s&=&s(p,\kappa)=\beta(p,\kappa)+1-p,
 \end{eqnarray}  
 where  $\beta(p,\kappa)$ is given by \eqref{betaunb} and $s(p,\kappa)$ by \eqref{sp}. This choice for $s$ is precisely the one that insures that $\nu$ \eqref{nu} equals
 \begin{eqnarray}\label{nup}
 \nu(s(p,\kappa),\kappa)=p. \end{eqnarray} 
From \eqref{asymp} in Lemma \ref{3.2} we have for $w\in A_t$  the asymptotic behavior for large $t$
$$
\mathbb E\left[\big|g_t'(w)\big|^s\,1_{\mathcal H(\theta,t)} \right]\asymp \exp(-\nu\,t)\big(\sin(\theta/2)\big)^q,$$ so that 
 \begin{eqnarray}\label{Eintg}
  \mathbb E \big[\mathcal I_p(t)\big]\asymp\int_0^{2\pi} \sin^q(\theta/2) d\theta.  
  \end{eqnarray} 
 This integral converges as Eqs. \eqref{sp}, \eqref{nu}, and \eqref{nup} imply that 
 $$
 \frac{\kappa}{8}q(s,\kappa)=\nu(s,\kappa)-\frac{s}{2}=p-\frac{s}{2}=\frac{1}{4}\left(\sqrt{1+2\kappa p}-1\right) \geq 0.
 $$
 The integral \eqref{intg} over $A_t$ can be mapped back via the $g_t$ map to the subset $g_t(A_t)$ of the unit circle and equals 
  \begin{eqnarray}\label{intg-1}
 \mathcal I_p(t)&=& \int_{g_t(A_t)} e^{pt}\I (g_t^{-1}\big)'(z)\I^{p-\beta(p)}\, \I dz\I,\,\,\,  z=g_t(w),\\ \nonumber
 &=& \int_{g_t(A_t)}   \I\tilde f'_t(z) \I^p \,\I (g_t^{-1}\big)'(z)\I^{-\beta(p)}\, \I dz\I.
 \end{eqnarray} 
 Consider then the set of sub-arc integrals
   \begin{eqnarray}\label{intgC}
 \mathcal I_p(\mathcal C, t)&:=&\int_{\mathcal C}   \I\tilde f'_t(z) \I^p \,\I (g_t^{-1}\big)'(z)\I^{-\beta(p)}\, \I dz\I,\\ \label{ineqC} 
\mathcal I_p(\mathcal C, t)&\leq& \mathcal I_p(t),\,\,\,\, \forall \mathcal C \varsubsetneq g_t(A_t)  .
 \end{eqnarray}      
    Using Schwarz's reflection principle, the function  $g_t^{-1}(z)$ (therefore its blow-up $\tilde f_t(z)=e^t g_t^{-1}(z)$) can be analytically extended,  outside of the unit disk $\mathbb D$, by inversion with respect to the unit circle of any angular sector spanned by a strict sub-arc $ \mathcal C \varsubsetneq g_t(A_t)$.  Koebe's theorem then implies in all these angular sub-sectors the uniformly bounded behavior     \begin{eqnarray}\label{koebeineq}
    C^{-1}\leq \left |\frac{\tilde f'_t(rz)}{\tilde f'_t(z)}\right |=\left |\frac{g_t^{-1}(rz)}{g_t^{-1}(z)}\right |\leq C,\,\,\,  r\leq 1,\,\,\,\forall z\in \mathcal C\varsubsetneq g_t(A_t)\subset\partial\mathbb D,
    \end{eqnarray}
 where the constant $C$ depends on the sub-arc $\mathcal C$ of $g_t(A_t)$, and may go to infinity when $ \mathcal C\to g_t(A_t)$. 
 
 By Koebe's bounds \eqref{koebeineq}, one can extend the boundary integral \eqref{intgC} to the interior of $\mathbb D$:
    \begin{eqnarray}\label{intgCr}
 \mathcal I_p(\mathcal C, t)\asymp  \mathcal I_p(r\,\mathcal C, t):=\int_{\mathcal C}   \I\tilde f'_t(rz) \I^p \,\I (g_t^{-1}\big)'(rz)\I^{-\beta(p)}\, \I dz\I,\,\forall r \leq 1,\,\forall \mathcal C \varsubsetneq g_t(A_t)  .
 \end{eqnarray}  
         Introduce now the time-dependent (random) radius $r_t$ 
        \begin{eqnarray}\label{rt}
        r_t:=1-L_t ;\,\,\,\,r_t\to 1^{-}, \,\,\, L_t\to 0,\,\,\, t\to +\infty .      
         \end{eqnarray}     
    In the boundary arc $w\in g_t^{-1}(\mathcal C) \varsubsetneq A_t$,  and $z=g_t(w)\in \mathcal C\varsubsetneq g_t(A_t)$, we have seen that the derivative tends to zero uniformly as $|g'_t(w)|=\I (g_t^{-1}\big)'(z)\I^{-1}\asymp L_t=1-r_t$, for $t\to +\infty$. For the particular choice of radius $r=r_t$, the equivalence  \eqref{intgCr}  can then be rewritten as
      \begin{eqnarray}\label{intgCrt}
 \mathcal I_p(\mathcal C, t)\asymp  \mathcal I_p(r_t\,\mathcal C, t)\asymp (1-r_t)^{\beta(p)}\int_{\mathcal C}   \I\tilde f'_t(r_tz) \I^p  \I dz\I,\,\,\,\forall \mathcal C \varsubsetneq g_t(A_t)  .
 \end{eqnarray}    
 As suggested in Section \ref{motiv} above in the case of the whole-plane map $f_0$, and for $p\geq p^*(\kappa)$, one  assumes that the similar integral, extended to the whole circle of radius $r_t <1$,  is dominated by the localized integral \eqref {intgCrt} for $t\to +\infty$. This condensation of the integral's support is precisely the signal of the onset of the transition from the standard SLE bulk spectrum $\beta_0(p,\kappa)$ \eqref{beta000} to the unbounded whole-plane spectrum $\beta(p,\kappa)$ \eqref{betaunb}. We therefore expect for $p\geq p^*(\kappa)$
       \begin{eqnarray}\label{intgDrt}
 \mathcal I_p(r_t\,\mathcal C, t)\asymp (1-r_t)^{\beta(p)} \int_{\partial \mathbb D}   \I\tilde f'_t(r_tz) \I^p \I dz\I.
 \end{eqnarray}   
From Eqs.  \eqref{ineqC}, \eqref{intgCrt} and \eqref{intgDrt}, and from the finite expectation result \eqref{Eintg},  one therefore concludes that, in expectation,   \begin{eqnarray}\label{EintgDrt}
 \mathbb E\left[(1-r_t)^{\beta(p)} \int_{\partial \mathbb D}   \I\tilde f'_t(r_tz) \I^p  \I dz\I\right]\asymp 1,\,\,\,t\to +\infty.
 \end{eqnarray}    
Since the random radius $r_t\to 1$, this equivalence  is (formally) very similar to the equivalence \eqref{betapow} above, which can serve as a heuristic definition of the average integral means spectrum. This strongly suggests why the average integral means spectrum $\beta(p,\kappa)$ \eqref{betaunb}, which is specific to the \emph{unbounded whole-plane} SLE map, is intimately related, via the packing spectrum $s(p,\kappa)=\beta(p,\kappa)-p+1$ \eqref{packing}, to the  derivative exponents  \eqref{nu}, as derived by Lawler, Schramm and Werner in Ref. \cite{MR2002m:60159b}:   \emph{the derivative exponent $p=\nu(s,\kappa)$ is the inverse  function of the unbounded whole-plane packing spectrum $s(p,\kappa)$.} \qed
 
  
 \section{Appendices}\label{Appendices} 
\subsection{Appendix A:  A brief history of Bieberbach's conjecture}\label{appendix}
\subsubsection{Proof for $n=2$ (Bieberbach \cite{Bi}, 1916)}
First, let us introduce the normalized, so-called {\it schlicht} class of univalent functions $$ \mathcal S=\{f:\,\bb D\to\bb C\, \mathrm{holomorphic \,and \,injective}\,;\,f(0)=0, \,f'(0)=1\}.$$
The Bieberbach conjecture is clearly equivalent to $\I a_n\I\leq n,\,n\geq 2$ for $f\in \mathcal S$. A related class of normalized functions is
$$\Sigma=\left\{f:\,\Delta=\overline{\bb C}\backslash\overline{\bb D}\to\overline{\bb C}\,\,\mathrm{holomorphic\,and\,injective}\,;\,f(z)=z+\sum_{n=0}^{\infty}b_n z^{-n}\;\mathrm{at}\;\infty\right\}.$$
The mapping $f\mapsto F$, where $F(z)=1/f(1/z)$, is clearly a bijection from $\mathcal S$ onto $\Sigma'$, the subclass of $\Sigma$ consisting of functions that do not vanish in $\Delta$. A simple application of the Stokes formula shows that if $f\in\Sigma$ then, denoting by $\I B\I$ the Lebesgue measure (area) of the Borelian subset $B$ of the plane, $$\I\bb C\backslash f(\Delta)\I=\pi\big(1-\sum_{n\geq 1}n\I b_n\I^2\big).$$
Since the area is a positive quantity, a consequence of this equality is that $\I b_1\I\leq 1$. But applying this inequality directly to the function $F$, image  in $\Sigma'$ of $f\in \mathcal S$,  does bring anything conclusive.  Bieberbach's idea was then to apply this inequality to an odd function in $\mathcal S$.

Let $f\in \mathcal S$ then $z\mapsto f(z)/z$ does not vanish in the disc and thus it has a unique holomorphic square root $g$ which is equal to $1$ at $0$. Then,  $h(z)=zg(z^2)$, such that $f(z^2)=h(z)^2$, is still in class $\mathcal S$ but is moreover odd. This establishes a bijection ($f\mapsto h$) between $\mathcal S$ and the set of odd functions in $\mathcal S$. Now, if $f(z)=z+a_2z^2+a_3z^3+\ldots$ belongs to class $\mathcal S$, then  $h(z)=z+\frac{1}{2}{a_2}{z^3}+O(z^5)$ and  the associated $H\in\Sigma$ satisfies $H(z)=1/h(1/z)=z-\frac{a_2}{2z}+\cdots$ By the area proposition, $\I a_2\I\leq 2.$ 
{\rkk The fact that $\vert a_2\vert $ is bounded above for functions in class $\mathcal S$ 
implies (see \cite{Pommerenkeuniv}) that the class $\mathcal S$ 
 is compact.}
\begin{rkk}\label{weakbiber}
 As a corollary, one can state a 
weak form of the Bieberbach conjecture, namely that for each $n\geq 2$ there
exists a positive constant $C_n < +\infty$ such that for any $f\in
 \mathcal S,\,f(z)=z+\sum_{n\geq 2}a_n z^n$,  then  $\vert a_n\vert\leq C_n$.
\end{rkk}

\subsubsection{Proof for $n=3$ (Loewner \cite{Lo}, 1923)}
Replacing $f(z)$ by $f(rz)$ with $r<1$ but close to $1$, one sees that it suffices to prove the estimate for conformal mappings onto smooth Jordan domains containing $0$. Consider such a domain $\Omega$ and let $\gamma:\, [0,t_0]\to \bb C$ be a parametrization of its boundary. Introduce then $\Gamma:\,[0,\infty)\to\bb C$, a Jordan arc joining $\gamma(0)=\gamma(t_0)$ to $\infty$ inside the outer Jordan component. We then define $$\Lambda(t):=\gamma(t), 0\leq t\leq t_0;\,\Lambda(t):=\Gamma(t-t_0),\,t\geq t_0,$$ and define for $t>0$, $$\Omega_t=\bb C\backslash \Lambda([t,\infty)).$$ The domain $\Omega_t$ is a simply connected domain containing $0$ and we can thus consider its Riemann mapping $f_t: \bb D\to \Omega_t,\,f_t(0)=0, f_t'(0)>0.$ By the Caratheodory convergence theorem, $f_t$ converges as $t\to 0$ to $f$,  the Riemann mapping of $\Omega$. We may assume without loss of generality that $f'(0)=1$ and, by a change  of time $t$ if necessary, that $f_t'(0)=e^t$.

The key idea of Loewner is to observe that the sequence of domains $\Omega_t$ is increasing, which translates into $ \Re\left({\frac{\partial f_t}{\partial t}}/{z\frac{\partial f_t}{\partial z}}\right)>0$ or, equivalently, that the same quantity is the Poisson integral of a positive measure, actually a probability measure because of the choice of parametrization $f_t'(0)=e^t$. Now, the fact that the domains $\Omega_t$ are slit domains implies that for every $t$ this probability measure must be, on the unit circle, the Dirac mass at $\lambda(t)=f_t^{-1}(\Lambda(t))$. Even if this is not needed in Loewner's proof, it is worthwhile to notice that $\lambda$ is a continuous function. The process $\Omega_t$ is then driven by the function $\lambda$, in the sense that $(f_t)$ satisfies the Loewner differential equation
\begin{equation}\label{loewnerbis}
\frac{\partial f_t}{\partial t}=z\frac{\partial f_t}{\partial z}\frac{\lambda(t)+z}{\lambda(t)-z}.
\end{equation}
To finish Loewner's proof, one extends both sides of the last equation as power series, with $f_t(z)=e^t(z+a_2z^2+a_3z^3+\cdots)$, and simply identifies the coefficients, as was done in Section \ref{analytic}. This leads to:
\begin{eqnarray*} \dot{a}_2-a_2&=&2\overline{\lambda},\\
\dot{a}_3-2a_3&=&4a_2\overline{\lambda}+2\overline{\lambda}^2.
\end{eqnarray*}
 As seen above (Eqs. \eqref{a2bis} \& \eqref{a3ter}), this is solved by  
\begin{eqnarray*}a_2(t)&=&-2e^t\int_t^\infty\overline{\lambda}(s)e^{-s}ds,\\
a_3(t)&=&4e^{2t}\left(\int_t^\infty \overline{\lambda}(s)e^{-s}ds\right)^2-2e^{2t}\int_t^\infty e^{-2s}\overline{\lambda}^2(s)ds.
\end{eqnarray*}
The first equation gives a new proof that $|a_2|\leq 2|a_1|=2$. 
For $a_3$, by considering $e^{-i\alpha}f(e^{i\alpha} z)$, one remarks that it suffices to prove that $\Re(a_3)\leq 3$. To this aim, write $\lambda(s)=e^{i\theta(s)}$. The Cauchy-Schwarz inequality,
$$\left(e^t\int_t^\infty e^{-s}\cos\theta(s)ds\right)^2\leq e^t\int_t^\infty e^{-s}\cos^2\theta(s)ds\,,$$ gives 
\begin{align*}
\Re(a_3)=\;&4e^{2t}\left(\int_t^\infty e^{-s}\cos{\theta(s)}ds\right)^2\\
&-4e^{2t}\left(\int_t^\infty e^{-s}\sin{\theta(s)}ds\right)^2-2e^{2t}\int_t^{\infty}e^{-2s}\cos{2\theta(s)}ds\\
\leq\;&4\int_t^{\infty}\left(e^{t-s}-e^{2(t-s)}\right)\cos^2\theta(s)ds +1\\
\leq \;&4\int_t^{\infty}\left(e^{t-s}-e^{2(t-s)}\right)ds +1 =3.
\end{align*}

\subsubsection{ The Bieberbach conjecture after Loewner}
 The next milestone after the 1923 theorem by Loewner is the proof in 1925 by Littlewood \cite{Li} that in class $\mathcal S,\,\I a_n\I\leq en$. In 1931, Dieudonn\'e \cite{Di}  proved the conjecture for functions with real coefficients. In 1932, Littlewood and Paley \cite{LP}  proved that the coefficients of an odd function in $\mathcal S$ are bounded by $14$, and they conjectured that the best bound is $1$, a conjecture that implies Bieberbach's. This conjecture was disproved in 1933 by Fekete and Szeg\H o \cite{FS}  for $n=5$. In 1935, Robertson \cite{R} stated the weaker conjecture $$\sum_{k=1}^{n}\I a_{2k+1}\I^2\leq n,$$ which also implies the Bieberbach conjecture.  
The next milestone was due in the sixties to Lebedev and Milin \cite{LM}. It had already been observed by Grunsky \cite{Gr} in 1939 that the logarithmic coefficients $\gamma_n$ defined by
$$\log [{f(z)}/{z}]=2\sum_{n=1}^{\infty}\gamma_nz^n$$ can easily be estimated.  Lebedev and Milin \cite{LM}  showed, through three inequalities, how to pass from those estimates to estimates for $f$. This allowed Milin \cite{M} to prove that $\I a_n\I\leq 1.243\, n $. He  then stated what has become known as  Milin conjecture: $$ \sum_{m=1}^{n}\sum_{k=1}^{m}\big(k\I\gamma_k\I^2-{1}/{k}\big)\leq 0.$$
It should be noticed that $\gamma_n=1/n$ for the Koebe function but the stronger conjecture $\I\gamma_n\I\leq 1/n$ is false, even as an order of magnitude. It happens that Milin $\Rightarrow $ Robertson $\Rightarrow $ Bieberbach, and de Branges actually proved the Milin conjecture.
\subsection{Appendix B: Coefficient quadratic expectations}\label{appendcoeff}
\subsubsection{Quadratic third order coefficient}\label{Appthird}
\medskip
\textcolor{black}{For calculations involving $a_3$ as given by \eqref{a3ter},  we compute  $\mathbb E (\I a_3-\mu a_2^2\I^2)$ for all $\mu$  real constant, and prove Proposition \eqref{theo-a3mu}.} 
\begin{proof} 
We write $$e^{-4t}\I a_3-\mu a_2^2\I^2=16(1-\mu)^2I_1-16(1-\mu)\Re{I_2}+4I_3,$$ where 
\begin{align*}
&I_1=\int_t^\infty \int_t^\infty \int_t^\infty \int_t^\infty e^{-(s_1+s_2+s_3+s_4)}\overline{\lambda}(s_1)\lambda(s_2)\overline{\lambda}(s_3)\lambda(s_4)ds_1ds_2ds_3ds_4,\\
&I_2=\int_t^\infty \int_t^\infty \int_t^\infty e^{-(s_1+s_2+2s_3)}\overline{\lambda}(s_1)\overline{\lambda}(s_2)\lambda(s_3)^2ds_1ds_2ds_3,\\
&I_3=\int_t^\infty \int_t^\infty e^{-2(s_1+s_2)}\overline{\lambda}(s_1)^2\lambda(s_2)^2 ds_1ds_2.
\end{align*}
From now on, we set  the parameter $t=0$ in the above formulae. The computation of $I_3$ follows the same lines as that in Proposition \ref{theo-a2} and we find $$\mathbb E(I_3)=\Re\left(\frac{1}{2(2+\eta_2)}\right).$$
To compute $\mathbb E(I_2)$ we use the strong Markov property. First, we may write by symmetry $$I_2=2\int_{s_1=0}^{\infty} \int_{s_2=s_1}^{\infty} \int_{s_3=0}^{\infty} e^{-(s_1+s_2+2s_3)}e^{i(L_{s_3}-L_{s_1})}e^{i(L_{s_3}-L_{s_2})}ds_1ds_2ds_3;$$  we cut this integral into  $I_2=2(I_{2,1}+I_{2,2}+I_{2,3})$, where in $ I_{2,1}$ (resp. in $I_{2,2},I_{2,3}$), $s_3$ lies in $[0,s_1]$ (resp. in $[s_1,s_2], [s_2,\infty)$).
For $I_{2,1}$,  write $$e^{i(L_{s_3}-L_{s_1})}e^{i(L_{s_3}-L_{s_2})}=e^{-2i(L_{s_1}-L_{s_3})}e^{-i(L_{s_2}-L_{s_1})},$$ so that the Markov property can be used to get its expectation as $e^{-\overline{\eta_2}(s_1-s_3)}e^{-\overline{\eta_1}(s_2-s_1)}.$ From this, the value of $\mathbb E(I_{2,1})$ easily follows as $$\mathbb E(I_{2,1})=\frac{1}{4(1+\overline{\eta_1})(2+\overline{\eta_2})}.$$ Similar considerations lead to
$$\mathbb E(I_{2,2})=\frac{1}{4(1+\overline{\eta_1})(3+\eta_1)},\,
\mathbb E(I_{2,3})=\frac{1}{4(2+\eta_2)(3+\eta_1)}.
$$
By combining these computations we get $$\Re\mathbb E (I_2)=\Re\left(\frac{1}{2(1+{\eta_1})(2+{\eta_2})}+\frac{1}{2(1+\overline{\eta_1})(3+\eta_1)}+\frac{1}{2(2+\eta_2)(3+\eta_1)}\right).$$
The computation of $I_1$ follows the same lines. First, by symmetry, $$I_1=4\int_0^\infty \int_{s_1}^\infty \int_0^\infty \int_{s_3}^\infty e^{-(s_1+s_2+s_3+s_4)}e^{i(L_{s_3}-L_{s_1})}e^{i(L_{s_4}-L_{s_2})}ds_1ds_2ds_3ds_4.$$
We then split this integral into the sum of six pieces, respectively associated with the domains 
(I) $s_3<s_4<s_1<s_2$; 
(II) $s_3<s_1<s_4<s_2$; 
(III) $s_3<s_1<s_2<s_4$; 
(IV) $s_1<s_3<s_4<s_2$; 
(V) $s_1<s_3<s_2<s_4$; 
(VI) $s_1<s_2<s_3<s_4$.

Clearly, the respective contributions of (I) and (VI), (II) and (V), (III) and (IV), are complex conjugate of each other. 
The same arguments as above give,  in a short-hand notation,
\begin{align*}
&\mathbb E(\text{I})=\frac{1}{4(1+\overline{\eta_1})(2+\overline{\eta_2})(3+\overline{\eta_1})},\\
&\mathbb E(\text{II})=\frac{1}{8(1+\overline{\eta_1})(3+\overline{\eta_1})},\,\,\,\,\mathbb E(\text{III})=\frac{1}{8(1+{\eta_1})(3+\overline{\eta_1})}.
\end{align*}
Altogether, we get $$\mathbb E(I_1)=\Re\left( \frac{2}{(1+{\eta_1})(2+{\eta_2})(3+{\eta_1})}+\frac{1}{(1+{\eta_1})(3+{\eta_1})}+\frac{1}{(1+\overline{\eta_1})(3+{\eta_1})}    \right).$$
\end{proof}
\subsubsection{Higher orders}\label{Apphigh}
Using dynamic programing, we performed computations of $\mathbb E(|a_n^2|)$ (formal up to $n=8$ and numerical up to $n=19$) on a usual computer. \textcolor{black}{The results for $a_3$ and  $a_4$ in the LLE case are given in Eqs. \eqref{a3eta} and  \eqref{a4eta}, respectively, whereas  for $a_5$:}
\begin{eqnarray}\nonumber
\mathbb E\left(\left|a_5\right|^2\right)&=&\,\frac{5!2^4}{(\eta_1+1)(\eta_1+3)(\eta_1+5)(\eta_1+7)}\\ \nonumber
&+&\frac{(\eta_1-1)(\eta_1-3)}{(\eta_1+1)(\eta_1+3)(\eta_1+5)(\eta_1+7)(\eta_2+2)(\eta_2+4)(\eta_2+6)(\eta_3+3)(\eta_3+5)}\\ &\times&\left[\frac{4\eta_2(\eta_2-4)(\eta_1+3)(\eta_3+1)(\eta_3-5)(\eta_1+3)(\eta_1+5)(\eta_2+4)}{3(\eta_4+4)}+Q\right] \label{a5square}
\end{eqnarray}
\ \begin{eqnarray}\nonumber
Q&=&\frac{4}{3}(24\eta_1^2\eta_2^2+9\eta_1^2\eta_2\eta_3^2+72\eta_1^2\eta_2\eta_3+39\eta_1^2\eta_2+36\eta_1^2\eta_3^2+288\eta_1^2\eta_3+520\eta_1^2\\ \nonumber
&+&19\eta_1\eta_2^3\eta_3+77\eta_1\eta_2^3+56\eta_1\eta_2^2\eta_3+472\eta_1\eta_2^2-36\eta_1\eta_2\eta_3^2-816\eta_1\eta_2\eta_3-3660\eta_1\eta_2\\ \nonumber
&-&144\eta_1\eta_3^2 -1152\eta_1\eta_3-2160\eta_1+75\eta_2^3\eta_3+285\eta_2^3+348\eta_2^2\eta_3^2+2952\eta_2^2\eta_3\\ \nonumber
&+&6420\eta_2^2+3507\eta_2\eta_3^2+26184\eta_2\eta_3+43245\eta_2+8460\eta_3^2+67680\eta_3+126900).
\end{eqnarray}
\noi In each  expression for $\mathbb E (|a_n|^2)$, and after the first term there, notice the presence   of the common factors $(\eta_1-1)(\eta_1-3)$ in the numerators. The first term, hence $\mathbb E (|a_n|^2)$ itself, equals  $1$ for $\eta_1=3$ (or $\kappa=6$), or equals $n$ for $\eta_1=1$ (or $\kappa=2$). \textcolor{black}{We  checked these results explicitly in symbolic computations up to $n=8$, and in numerical ones up to $n=19$.}\\
Let us end this Appendix with the results for $a_5$ to $a_8$ in the SLE case:
\begin{align*}
\mathbb E\left(\left|a_5\right|^2\right)=&\,(27 \kappa ^8 + 3242 \kappa ^7 + 194336 \kappa ^6 + 6142312 \kappa ^5 + 42644896 \kappa ^4 \\
&\,+ 119492832 \kappa ^3 + 153156096 \kappa ^2 + 87882624 \kappa + 18144000)\\
&\,\;/[36(\kappa + 14)(3 \kappa + 2)(\kappa + 10)(2 \kappa + 1)(\kappa + 6)(\kappa + 3)(\kappa + 1)(\kappa + 2)^2]\,;\\
\\
\mathbb E\left(\left|a_6\right|^2\right)=&\,\frac{2}{225}(216 \kappa ^{10}+29563 \kappa ^9+2062556 \kappa ^8+90749820 \kappa ^7+2277912280 \kappa ^6 \\
&\,+16419864848 \kappa ^5+50825787744 \kappa ^4+76716664128 \kappa ^3 \\
&\,+58263304320 \kappa ^2+21233664000 \kappa +2939328000) \\
&\;\,/[(\kappa +18)(3 \kappa +2)(\kappa +14)(2 \kappa +1)(\kappa +10)(\kappa +6)(5 \kappa +2)\\
&\,\;(\kappa +3)(\kappa +1)(\kappa +2)^2]\,;\\
\\
\mathbb E\left(\left|a_7\right|^2\right)=&\,\frac{1}{8100}(27000 \kappa ^{15}+4479353 \kappa ^{14}+373838334 \kappa ^{13}+20594712527 \kappa ^{12}\\
&\,+787796136854 \kappa ^{11}+19121503739240 \kappa ^{10}+221861771218136 \kappa ^{9}\\
&\,+1386550697705712 \kappa ^{8}+5130607642056896 \kappa ^7+11854768997862912 \kappa ^6\\
&\,+17547915006086400 \kappa ^5+16725481436226816 \kappa ^4+10110569026936320 \kappa ^3\\
&\,+3711483045734400 \kappa ^2+749049576192000 \kappa +63371911680000) \\
&\;\,/[(\kappa +22)(3 \kappa +1)(5 \kappa +2)(\kappa +18)(2 \kappa +1)(\kappa +14)(3 \kappa +2)\\
&\;\,(\kappa +10)(\kappa +6)(\kappa +5)(\kappa +3)(\kappa +1)^2(\kappa +2)^3]\,;
\end{align*}
\begin{align*}
\mathbb E\left(\left|a_8\right|^2\right)=&\,\frac{2}{99225}(729000 \kappa ^{18} + 143757261 \kappa ^{17} + 14031668642 \kappa ^{16} + 906444920407 \kappa ^{15} \\
&\,+ 42715714646750 \kappa ^{14} + 1476227672190480 \kappa ^{13}+ 34674813906653712 \kappa ^{12} \\
&\, + 471116720002819536 \kappa ^{11} + 3802657434377773600 \kappa ^{10} \\
&\,+ 19218418658636100992 \kappa ^9 + 63191729416067875840 \kappa ^8 \\
&\,+ 138392538501661946112 \kappa ^7 + 204258207932541043200 \kappa ^6 \\
&\,+ 203508494170475323392 \kappa ^5 + 135640094878259859456 \kappa ^4 \\
&\,+ 59063686024095313920 \kappa ^3 + 16005106174366310400 \kappa ^2 \\
&\,+ 2435069931098112000 \kappa + 158176291553280000) \\
&\,\;/[(7 \kappa + 2)(5 \kappa + 2)(\kappa + 26)(3 \kappa + 1)(\kappa + 22)(2 \kappa + 1)(\kappa + 18)(\kappa + 14)\\
&\,\;(3 \kappa + 2)(\kappa + 10)(\kappa + 5)(\kappa + 3)(\kappa + 6)^2(\kappa + 1)^2(\kappa + 2)^3]\,.
\end{align*}
These results call for two observations:\\
--All the coefficients of the polynomial expansions in $\kappa$ are positive.\\
--For $\kappa \to \infty$ (or $\eta\to\infty$), \textcolor{black}{the coefficients' quadratic moments vanish as $\kappa^{-1}$.}

\subsection{Appendix C:  A proof of Theorem \ref{FMcGo}}\label{McGo}
\begin{proof}
The proof closely follows the steps of that of Feng-McGregor's theorem.  
 We start with the computation of $h'$:
$$h'(z)=zf'(z^2)(f(z^2))^{-1/2}.$$
We may then write, putting $\rho=r^2,$
$$\int_0^{2\pi}\I h'(re^{i\theta})\I^p d\theta\leq\int_0^{2\pi}\frac{\I f'(\rho e^{i\theta})\I^p}{\I f(\rho e^{i\theta})\I^{p/2}} d\theta.$$
Consider now two positive reals $a,\,b$ such that $a-b=1$ and fix $0<p<2$. By H\"older inequality,
$$\int_0^{2\pi}\frac{\I f'(\rho e^{i\theta}\I^p}{\I f(\rho e^{i\theta})\I^{p/2}} d\theta\leq \left(\int_0^{2\pi}\frac{\I f'(\rho e^{i\theta})\I^2}{\I f(\rho e^{i\theta})\I^a} d\theta\right)^{p/2}\left(\int_0^{2\pi}\I f(\rho e^{i\theta})\I^{\frac{b p}{2-p}}d\theta\right)^{(2-p)/2}.$$
But
\begin{equation}\left(\int_0^{2\pi}\frac{\I f'(\rho e^{i\theta})\I^2}{\I f(\rho e^{i\theta})\I^a }d\theta\right)^{p/2}=\left(\int_0^{2\pi}\I f'(\rho e^{i\theta})\I^2\I f(\rho e^{i\theta})\I^{(2-a)-2}d\theta\right)^{p/2},\label{pom}\end{equation}
and we invoke the following Lemma, which is a consequence of Hardy's identity and Koebe's theorem (see \cite{Pommerenke}):
\begin{lemma}\label{tpom}
There exists a universal constant $C>0$ such that, for $f$  holomorphic and injective in the unit disk, with $f(0)=0,\,f'(0)=1$,
\begin{enumerate}[(i)]
\item if $p>0$, 
$$ \int_0^{2\pi}\I f'(\rho e^{i\theta})\I^2\I f(\rho e^{i\theta})\I^{p-2} d\theta\leq \frac{C}{(1-\rho)^{2p+1}};$$ 
\item if $p>1/2$, 
$$\int_0^{2\pi}\I f(\rho e^{i\theta})\I^p d\theta\leq\frac{C}{(1-\rho)^{2p-1}}.$$
\end{enumerate}\end{lemma}
\noindent Consider the last Lemma and \eqref{pom}; we seek for $a$ such that $2-a>0\Leftrightarrow  b<1$, together with
$$\frac{bp}{2-p}>1/2\Leftrightarrow b>\frac{1}{p}-\frac{1}{2},$$
and we may find such a pair $(a, \,b)$ iff $\frac{1}{p}-\frac{1}{2}<1\Leftrightarrow p>2/3.$\\
With this condition on $p$ satisfied, and for that choice of $(a,\,b)$, Theorem $\ref{tpom}$ implies that [note: $1-r\leq 1-\rho\leq 2(1-r)$],
$$\int_{0}^{2\pi}\I h'(re^{i\theta})\I^pd\theta\leq C(1-r)^{-[2(2-a)+1)p/2]}(1-r)^{-[bp-(2-p)/2]}\leq C(1-r)^{-(2p-1)}.$$
The last statement shows that $\beta_h(p)\leq 2p-1$ for $2/3<p<2$. In order to prove it for all $p>0$, we first need the
\begin{lemma}
There exists a universal constant $C>0$ such that
$$\I h'(z)\I\leq C(1-r)^{-2}.$$
\end{lemma}
Proof:  We write
$$\I h'(z)\I=\left\I z\sqrt{\frac{f'(z^2)}{f(z^2)}}\sqrt{f'(z^2)}\right\I \leq C(1-\rho)^{-1/2}(1-\rho)^{-3/2},$$
by the Koebe distortion theorem. (Use inequalities (11) and (13) on page 21 of \cite{Pommerenkeuniv}.) An immediate corollary of this Lemma is that
$\beta_h(p)\leq 2p$ if $p>0$.

We now argue as in \cite{Pommerenke}, using the fact that the function $\beta_h$ is convex. Any number bigger than $2/3$ may be written as $p+q$ with $2/3<p<2$ and $q>0$. We can then write
$$ p+q=t\frac{p}{t}+(1-t)\frac{q}{1-t}$$ where $t\in[0,1]$ is close to $1$ and 
$$\beta_h(p+q)\leq t\beta_h\left(\frac{p}{t}\right)+(1-t)\beta_h\left(\frac{q}{1-t}\right)\leq t\left(2\frac{p}{t}-1\right)+2q$$
by the Lemma above. Finally
$$\beta_h(p+q)\leq 2(p+q)-t$$
which converges to $2(p+q)-1$ as $t\to 1$.
\end{proof}
\bibliographystyle{plain}
\bibliography{biblio}

\end{document}